\documentclass[11pt]{article}
\usepackage[T1]{fontenc}
\usepackage{mathrsfs}
\usepackage{amsfonts}
\usepackage{lmodern}
\usepackage{microtype}
\usepackage{amsmath}
\usepackage{amssymb}
\usepackage{amsthm}
\usepackage{bm}
\usepackage{array}
\usepackage{appendix}
\usepackage[english]{babel}
\usepackage{bbm}
\usepackage{color}
\usepackage{comment}
\usepackage{float}
\usepackage{geometry} 
\usepackage{graphicx}
\usepackage{hyperref}
\usepackage[utf8]{inputenc}
\usepackage{longtable}
\usepackage{lscape}
\usepackage{natbib}
\setcitestyle{authoryear,open={(},close={)}} 
\usepackage{pdflscape}
\usepackage{siunitx}
\usepackage{tikz}
\usepackage{wrapfig}
\usepackage{latexsym}
\allowdisplaybreaks
\usepackage{paralist}
\usepackage{csquotes}
\MakeOuterQuote{"}
\usepackage{enumitem}
\usepackage{wasysym} 

\usepackage[font=small,labelfont=bf]{caption}
\usepackage{subcaption}

\tikzstyle{solid node}=[circle,draw,inner sep=2,fill=black] 
\tikzstyle{hollow node}=[circle,draw,inner sep=1.5] 
\tikzstyle{hollow node1}=[rectangle,draw,inner sep=2] 
\theoremstyle{plain}
\newtheorem{thm}{Theorem}[section]
\newtheorem{prop}[thm]{Proposition}

\newtheorem{definition}[thm]{Definition}
\newtheorem{cor}[thm]{Corollary}
\newtheorem{lemma}[thm]{Lemma}
\newtheorem{assump}{Assumption}[section]
\theoremstyle{remark}
\graphicspath{ {pictures/} }

\renewcommand{\P}{\operatorname{\mathbb{P}}}

\newcommand{\T}{\mathcal{T}}

\newcommand{\indep}{\perp\!\!\!\perp}
\newcommand{\notindep}{\not\!\perp\!\!\!\perp}

\newcommand{\pa}{\mathrm{pa}}
\newcommand{\Pa}{\mathrm{Pa}}
\newcommand{\ch}{\mathrm{ch}}
\newcommand{\Ch}{\mathrm{Ch}}
\newcommand{\an}{\mathrm{an}}
\newcommand{\An}{\mathrm{An}}
\newcommand{\desc}{\mathrm{desc}}
\newcommand{\Desc}{\mathrm{Desc}}

\newcommand{\diff}{\mathrm{d}}

\newcommand{\norm}[1]{\|{#1}\|}
\newcommand{\cB}{\mathcal{B}}
\newcommand{\point}{\,\cdot\,}

\newcommand{\inlaw}{\stackrel{d}{\longrightarrow}}
\newcommand{\bari}{\bar{\imath}}
\definecolor{ggr}{HTML}{006633}
\definecolor{bbl}{HTML}{0066CC}
\definecolor{ppn}{HTML}{CC0099}
\definecolor{darkgreen}{rgb}{0, .5, 0}
\definecolor{darkteal}{rgb}{0, .35, .35}
\definecolor{cerulean}{rgb}{0.0, 0.48, 0.65}

\theoremstyle{lema1}
\newtheorem{lema1}{Lemma}[subsection] 

\providecommand{\keywords}[1]
{
	\small	
	\textbf{\textit{Keywords}} --- #1
}
\newcommand{\js}[1]{\textcolor{blue}{\sffamily\footnotesize\upshape [#1]}}
\newcommand{\sa}[1]{\textcolor{purple}{\sffamily\footnotesize [#1]}}

\definecolor{mygreen}{rgb}{0.10, 0.6,0 }
\definecolor{violet}{rgb}{0.6, 0, 0.2}

\newcounter{example}
\newenvironment{example}[1][]{\refstepcounter{example}\par
	\bigskip\noindent\textbf{Example~\theexample{} (#1).} \rmfamily }{\medskip \qedhere\hfill$\diamondsuit$\medskip}

\usepackage{subfiles}

\title{Max-linear graphical models with heavy-tailed factors on trees of transitive tournaments}
\author{Stefka Asenova\thanks{Corresponding author. UCLouvain, LIDAM/ISBA, Voie du Roman Pays 20, 1348 Louvain-la-Neuve, Belgium. E-mail: stefka.asenova@uclouvain.be} \and Johan Segers\thanks{UCLouvain, LIDAM/ISBA, Voie du Roman Pays 20, 1348 Louvain-la-Neuve, Belgium. E-mail: johan.segers@uclouvain.be}}
\date{\today}
\begin{document}
	
	\maketitle

\begin{abstract}
	Graphical models with heavy-tailed factors can be used to model extremal dependence or causality between extreme events. In a Bayesian network, variables are recursively defined in terms of their parents according to a directed acyclic graph (DAG). We focus on max-linear graphical models with respect to a special type of graphs, which we call a \emph{tree of transitive tournaments}. The latter are block graphs combining in a tree-like structure a finite number of transitive tournaments, each of which is a DAG in which every two nodes are connected. We study the limit of the joint tails of the max-linear model conditionally on the event that a given variable exceeds a high threshold. Under a suitable condition, the limiting distribution involves the factorization into independent increments along the shortest trail between two variables, thereby imitating the behavior of a Markov random field. We are also interested in the identifiability of the model parameters in case some variables are latent and only a subvector is observed. It turns out that the parameters are identifiable under a criterion on the nodes carrying the latent variables which is easy and quick to check. 
\end{abstract}

	
	\section{Introduction}


	Dependence in multivariate linear factor models is determined by a collection of independent random variables, called factors, which are shared by the modelled variables. In extreme value analysis there are the max-linear and the additive factor models with heavy-tailed factors. In \citet{einmahl2012an}, it is shown that both have the same max-domain of attraction. 
	
	In \citet{gissibl2018max}, a link is made between such factor models and  probabilistic graphical models via a max-linear recursively defined structural equation model on a directed acyclic graph (DAG). Each node carries a variable defined as a \emph{weighted maximum} of its parent variables and an independent factor. This leads to a representation of the graphical model as a (max-)factor model as in \citet{einmahl2012an}, the factors relevant for a given variable being limited to the set of its ancestors. More recent is the linear causally structured model in \citet{gnecco2020causal}: each variable is the \emph{weighted sum} of the variables on all its parent nodes plus an independent factor. This leads to a representation where a single variable is a weighted sum of all its ancestral factors. 
	
	In this paper, we study a type of graph that, to the best of our knowledge, is not yet known and which we gave a name  
	that reflects its most important properties: a \emph{tree of transitive tournaments (ttt)}, denoted by $\T$. A tournament is a graph obtained by directing a complete graph, while a tournament is said to be transitive if it has no directed cycles. The name reflects the interpretation of such a graph as a competition where every node is a player and a directed edge points from the winner to the loser.  
	Some examples are hierarchical relations between members of animal and bird societies, brand preferences, and votes between two alternative policies \citep{harari1966the}. A ttt links up several such transitive tournaments in a tree-like structure. It is acyclic by construction. If there is a directed path from one node to another one, there is a unique shortest such path. Moreover, between any pair of nodes, there is a unique shortest undirected path. 
	
	In this paper, we study max-linear graphical models with respect to a ttt as defined in \eqref{eqn:mlgm} below. In particular, for a max-linear random vector $X=(X_v, v\in V)$ with node set $V$, we study the limit in distribution
	\begin{equation} \label{eqn:XvXu}
		\left( X_v/X_u, v\in V\setminus u \mid X_u>t \right)
		\inlaw (A_{uv}, v\in V), 
		\qquad t\rightarrow\infty.
	\end{equation}
	It is not hard to show that the limit distribution in \eqref{eqn:XvXu} is discrete \citep[Example~1]{segers2020one}. We show that if the ttt has a unique node without parents, a so-called source node, the joint distribution of $(A_{uv}, v\in V)$ is determined by products of independent multiplicative increments along the unique shortest undirected paths between the node $u$ at which the high threshold is exceeded on the one hand and the rest of the nodes on the other hand. Such behaviour is analogous to that of Markov random fields on block graphs in \citet{asenova2021extremes} and of Markov trees in \citet[Theorem~1]{segers2020one}. In turn, these results go back to the extensive literature on the additive or multiplicative structure of extremes for Markov chains \citep[e.g.][]{smith1992the, yun1998the, segers2007multivariate, janssen2014markov, resnick2013asymptotics}. 
	
	An underlying reason for the factorization into independent increments is the fact that a max-linear graphical model with respect to a ttt is a Markov random field with respect to the undirected graph associated to the original, directed graph when the ttt has a unique source. A ttt with unique source has no v-structures, that is, no nodes with non-adjacent parents. Both properties, the factorization of the limiting variables and the Markovianity with respect to the undirected graph, are lost if the graph contains v-structures. To show this, we rely on recent theory of conditional independence in max-linear Bayesian networks based on the notion of $*$-connectedness \citep{amendola2021markov, amendola2022conditional}. This theory diverges from classical results on conditional independence in Bayesian networks based on the notion of d-separation \citep{lauritzen1996graphical, koller2009probabilistic}.
	
	
	In our paper the graph is given. A significant line of research in the context of extremal dependence is graph discovery. Given observations on a number of variables represented as nodes in a graph, the task is to estimate the edges. For Bayesian networks we can also talk about causality discovery because directed edges show the direction of influence. 
	A first attempt to identify the DAG in the context of max-linear models is \citet{gissibl2018tail}, followed by several papers focusing on this topic: \citet{kluppelberg2021estimating}, \citet{buck2021recursive}, \citet{gissibl2021identifiability}, \citet{tran2021estimating} and \citet{tran2021causal}. The problems related to identifiability of the true graph and to the estimation of the edge weights are discussed in \citet{kluppelberg2019bayesian}. 
	\citet{gnecco2020causal} study a new metric called causal tail coefficient which is shown to reveal the structure of a linear causal recursive model with heavy-tailed noise. 
	Graph discovery for non-directed graphs is studied in \citet{engelke2020graphical}, \citet{engelke2020structure} and \citet{hu2022modelling}. 
	

	Inspired from practice, and more specifically river network applications \citep{asenova2021inference}, we study a different identifiability problem. If the structure of the graph is known, it may happen that on some nodes the variables are latent, i.e., unobserved. The identifiability problem in this case is whether two different parameter vectors can still generate the same distribution of the observable part of the model. If this is possible then we cannot uniquely identify all tail dependence parameters that characterize the full distribution. Similarly to \citet{asenova2021extremes}, the identifiability criterion involves properties of the nodes with latent variables. The criterion is specific for a ttt with unique source and is easy to check. Our identifiability problem resembles the "method of path coefficients" of Sewall Wright which uses a system of equations involving correlations to solve for the edge coefficients \citep{wright1934themethod}.

	
	The novelty of the paper lies in several directions. First, a new class of graphs is introduced, called a tree of transitive tournaments (ttt), which is the directed acyclic analogue of a block graph. It can be seen as a generalization of a directed tree, where edges are replaced by transitive tournaments. Second, we show that a max-linear graphical model over a ttt with unique source exhibits properties known for other graphical models, namely Markov trees \citep{segers2020one} and Markov block graphs \citep{asenova2021extremes}. In particular, when the ttt has a unique source, the model is Markov with respect to the skeleton of the graph. This property underlies the factorization of the tail limit into independent increments along the unique shortest trails. Finally, we study a problem of identifiability of the edge weights from the angular measure both when all variables are observed and also when some of them are latent. 
	
	The structure of the paper is as follows. In Section~\ref{sec:def} we introduce the ttt, the max-linear model, and its angular measure, which plays a key role in almost all proofs. In Section~\ref{sec:ctc} we discuss the limiting distribution of \eqref{eqn:XvXu} and give four equivalent characterizations of a max-linear graphical model with respect to a ttt with unique source. The identifiability problem is covered in Section~\ref{sec:ident}. The discussion summarizes the main points of the paper. The appendices contain some additional lemmas and the proofs that are not presented in the main text.  
	
	
	
	
	
	
	
	
	
	\section{Notions and definitions} \label{sec:def}
	
	\subsection{Directed graphs}
	Let $\T=(V,E)$ be a directed acyclic graph (DAG) with finite vertex (node) set $V$ and edge set $E\subset V\times V$. An edge $e:=(u,v)\in E$ is directed meaning  $(u,v)\neq (v,u)$; it is outgoing with respect to the parent node $u$ and incoming with respect to the child node $v$. The graph $\T$ excludes loops, i.e., edges of the form $(u, u)$, and as $\T$ is directed, we cannot have both $(u,v)\in E$ and $(v,u)\in E$. Two nodes $u$ and $v$ are adjacent if $(u, v)$ or $(v, u)$ is an edge. A cycle is a sequence of edges $e_1,\ldots,e_n$ with $e_k = (u_k, u_{k+1})$ and $u_1 = u_{n+1}$ for some nodes $u_1,\ldots,u_n$. The property that $\T$ is acyclic means that it does not contain any cycle. The graph $\T$ is assumed connected, i.e., for any two distinct nodes $u$ and $v$ we can find nodes $u_1=u,u_2,\ldots,u_{n+1}=v$ such that $u_k$ and $u_{k+1}$ are adjacent for every $k = 1,\ldots,n$; we call the associated edge sequence an undirected path or a trail between $u$ and $v$. If all edges are directed in the same sense, i.e., $(u_k, u_{k+1}) \in E$ for all $k = 1,\ldots,n$, we talk about a (directed) path from the ancestor $u$ to the descendant $v$. Recall that a path is directed by convention, so when we need non-directed paths this will be indicated explicitly. Between a pair of nodes there may be several paths. The set of all paths between two nodes $u,v\in V$ is denoted by $\pi(u,v)$. An element, say $p$, of $\pi(u,v)$ is a collection of edges, $\{(v_1, v_2), (v_2,v_3),\ldots, (v_{n-1},v_n)\}$ for a path that involves the non-repeating nodes $\{v_1=u, v_2, \ldots, v_{n-1}, v_{n}=v\}$. Note that $\pi(u, u) = \varnothing$ in an acyclic graph. 
	
	A source is a node without parents. If a DAG has a unique source, this node is an ancestor of every other node. This property follows from the following reasoning: let $u_0$ denote the unique source node of the DAG, and let $v$ be any other node different from $u_0$. Then $v$ must have a parent, say $u$. If $u = u_0$, we are done. Otherwise, replace $v$ by $u$ and restart. Since the graph is finite and has no cycles, this chain must stop at some moment at a node without parents. But this node is necessarily equal to $u_0$ by assumption. 

	A graph, directed or not, is complete if there is an edge between any pair of distinct nodes.
	A subgraph of a graph is biconnected if the removal of any of its nodes will not disconnect the subgraph. A maximal biconnected subgraph, also known as a biconnected component, is a subgraph that cannot be extended by adding one adjacent node without violating this principle.
	
	A directed complete graph is called a tournament. A tournament $\tau=(V_\tau, E_\tau)$ is transitive if $(u, v), (v, w) \in E_\tau$ implies $(u, w) \in E_\tau$. A transitive tournament is necessarily acyclic. The graph-theoretic properties of transitive tournaments are studied in \citet{harari1966the}. The property most used here is that the set of out-degrees of the $d$ nodes of a transitive tournament is $\{d-1, d-2, \ldots, 0\}$; the in- and out-degrees of a node are the numbers of incoming and outgoing edges, respectively. 
	
	A subgraph of a graph is a maximal transitive tournament if it is not properly contained in another subgraph which is also a transitive tournament. The set of maximal transitive tournaments that are subgraphs of a DAG $\T$ will be denoted by $\mathbb{T}$. For brevity we will just write tournament when we mean a maximal transitive tournament and denote it by $\tau$. 
	
	\subsection{Tree of transitive tournaments}
	


	A block graph is an undirected graph where every maximal biconnected subgraph is a complete graph \citep{le2010the}. Let $T$ denote the non-directed version of $\T$, also called the skeleton of $\T$. It shares the same node set as $\T$, and for every edge $(u, v)$ in the original graph $\T$, the reverse edge $(v, u)$ is added to form the edge set of the skeleton graph $T$, after which each pair of edges $\{(u, v), (v, u)\}$ is identified with the undirected edge $\{u, v\}$ of $T$.
	
	\begin{definition}[Tree of transitive tournaments (ttt)] \label{def:ttt}
		A tree of transitive tournaments is a connected directed acyclic graph whose skeleton is a block graph.
	\end{definition}

	A ttt enjoys three key properties. They all follow from the link with block graphs, whose characteristics can be found in \citet{le2010the}.
	\begin{lemma}[Properties I]
		For a ttt, the following properties hold:
		\begin{enumerate}[label=(P\arabic*)]
			\item \label{ttt:prop1} two or more maximal transitive tournaments can have at most one common node, referred to as a \emph{separator node};
			\item \label{ttt:prop2} there is no undirected cycle that passes through nodes in different maximal transitive tournaments;
			\item \label{ttt:prop3} between every pair of nodes there is a unique shortest trail (undirected path).
		\end{enumerate}
	\end{lemma}

\begin{proof}
	All properties are direct consequences of the fact that removing directions from the ttt we obtain a block graph. In a block graph the minimal separators sets are singletons \citep[Theorem~B]{harari1962acharacterization}; there is a unique shortest path between two nodes \citep[Theorem~1.a)]{behtoei2010a}; and the graph is acyclic up to blocks, a property that follows from the first one.  
\end{proof}

	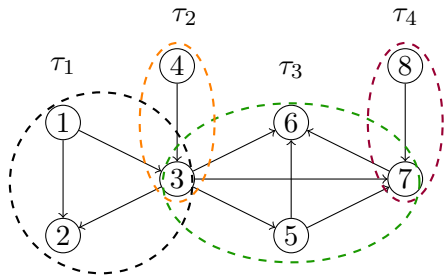
\begin{figure}{} 
		\centering 
		\begin{tikzpicture}
			\node[hollow node] (1) at (0,0) {1};
			\node[hollow node] (2) at (0,-1.5)  {2};
			\node[hollow node] (3) at (1.5,-0.75)  {3};
			\node[hollow node] (4) at (1.5,0.75)  {4};
			\node[hollow node] (7) at (4.5,-0.75)  {7};
			\node[hollow node] (5) at (3,-1.5)  {5};
			\node[hollow node] (6) at (3,0)  {6};
			\node[hollow node] (8) at (4.5,0.75)  {8};
			\path[->] (1) edge (2)  ;
			\path[->] (1) edge (3);
			\path[<-] (2) edge (3)  ;
			\path[->] (3) edge (5);
			\path[<-] (7) edge (3);
			\path[->] (3) edge (6);
			\path[->] (5) edge (6);
			\path[->] (5) edge (7);
			\path[->] (4) edge (3);
			\path[<-] (7) edge (8);
			\path[->] (7) edge (6);
			\node[] (t1) at (0,0.75)  {$\tau_1$};
			\node[] (t2) at (1.6,1.4)  {$\tau_2$};
			\node[] (t3) at (3,0.7)  {$\tau_3$};
			\node[] (t4) at (4.5,1.4)  {$\tau_4$};
			\begin{scope}[dashed, thick]
				\draw[color=black] (0.5,-0.8) circle (1.2cm);
				\draw[rotate=0, color=mygreen] (3,-0.8) ellipse (48pt and 30pt);
				\draw[rotate=0, color=violet] (4.5,0) ellipse (14pt and 30pt);
				\draw[rotate=0, color=orange] (1.5,0) ellipse (14pt and 30pt);
			\end{scope}
		\end{tikzpicture}
		
		\caption{
			A tree of four maximal transitive tournaments: $\tau_1$, $\tau_2$, $\tau_3$, and $\tau_4$. The skeleton graph is the same, but with arrow heads removed. Node $3$ is a separator node between tournaments $\tau_1$, $\tau_2$ and $\tau_3$, while node $7$ is a separator node between tournaments $\tau_3$ and $\tau_4$. Nodes $1$, $4$ and $8$ are source nodes, i.e., have no parents. The subgraph with node set $\{1, 3, 4\}$ is a v-structure: node $3$ has non-adjacent parents $1$ and $4$. Other v-structures within the ttt are the subgraphs with node sets $\{3, 7, 8\}$ and $\{5, 7, 8\}$. Between any pair of distinct nodes there is a unique shortest trail; for instance, nodes $4$ and $8$ are connected by the trail passing through nodes $3$ and $7$. There is no undirected cycle encompassing several tournaments.}
		\label{fig:ttt}
	\end{figure}
	
Similarly to block graphs \citep{le2010the} a ttt can be seen as a tree whose edges are replaced by transitive tournaments.

	In a ttt, if there is at least one (directed) path between distinct nodes $u$ and $v$, there is a \emph{unique shortest path} (see Lemma~\ref{lem:oned-1}-1) between them, which we denote by $p(u,v)$, and which belongs to $\pi(u, v)$. We also set $p(u,u)=\varnothing$ by convention. 
	


A key object in the paper is a ttt with unique source. In Lemma~\ref{lem:oned-1}-2 below it is shown that in this case there are no nodes with parents that are not adjacent or `married', also known as a v-structure \citep{koller2009probabilistic}. 

Consider the ttt in Figure~\ref{fig:ttt}, which presents some of the notions introduced above.  Each tournament is acyclic and we cannot find a cycle passing through different tournaments either. This is why we call such a graph a \emph{tree} of transitive tournaments. There are three v-structures: one on nodes $1,3,4$, one on $3, 7, 8$ and one on $5,7,8$. The main results in this paper require a ttt without v-structures. According to Lemma~\ref{lem:oned-1}, there are no v-structures in a ttt with unique source. This is illustrated in Figure~\ref{fig:ttt_unique}. 

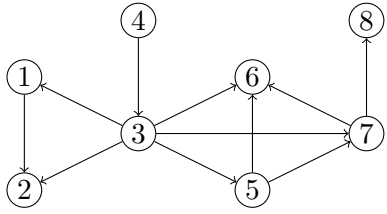
\begin{figure}{} 
	\centering 
	\begin{tikzpicture}
		\node[hollow node] (1) at (0,0) {1};
		\node[hollow node] (2) at (0,-1.5)  {2};
		\node[hollow node] (3) at (1.5,-0.75)  {3};
		\node[hollow node] (4) at (1.5,0.75)  {4};
		\node[hollow node] (7) at (4.5,-0.75)  {7};
		\node[hollow node] (5) at (3,-1.5)  {5};
		\node[hollow node] (6) at (3,0)  {6};
		\node[hollow node] (8) at (4.5,0.75)  {8};
		\path[->] (1) edge (2)  ;
		\path[<-] (1) edge (3);
		\path[<-] (2) edge (3)  ;
		\path[->] (3) edge (5);
		\path[<-] (7) edge (3);
		\path[->] (3) edge (6);
		\path[->] (5) edge (6);
		\path[->] (5) edge (7);
		\path[->] (4) edge (3);
		\path[->] (7) edge (8);
		\path[->] (7) edge (6);
	\end{tikzpicture}
	\caption{
		A tree of four maximal transitive tournaments. The skeleton graph is the same as the one in Figure~\ref{fig:ttt}, but the graph here has a single source, on node $4$. We see that there are no v-structures anymore in the graph.}
	\label{fig:ttt_unique}
\end{figure}


Considered on its own, every tournament in a ttt has a unique source; this follows from the ordering of the out-degrees due to \citet{harari1966the} mentioned earlier. When we talk about a source node, we will always state if we refer to the whole graph or with respect to a particular tournament.  

In a general directed graph $(V, E)$, let $\pa(v)\in V$ denote the set of parents of $v\in V$ and put $\Pa(v)=\pa(v)\cup \{v\}$. In a similar way let $\an(v)$, $\desc(v)$, and $\ch(v)$ denote the sets of ancestors, descendants, and children, respectively, excluding $v$, while $\An(v)$, $\Desc(v)$, and $\Ch(v)$ denote the same sets but including $v$. 

Below we present some additional properties used often in the paper.
\begin{lemma}[Properties II] \label{lem:oned-1}
	Let $\T$ be a tree of transitive tournaments as in Definition~\ref{def:ttt}. We have the following statements:
	\begin{compactenum} 
		\item If there is a path between two nodes, then there is a unique shortest path between them.
		\item The ttt $\T$ has a unique source if and only if it possesses no v-structures.
		\item If $\T$ has a unique source, then for any two nodes $i\neq j$, the sets $\Desc(i)$ and $\Desc(j)$ are either disjoint or one contains the other, that is, $i$ is an ancestor of $j$ or vice versa.
	\end{compactenum}  
\end{lemma}

\begin{lemma}[Properties III] \label{lem:addition_to_oned-1}
	Consider a ttt $\T=(V,E)$ as in Definition~\ref{def:ttt} with unique source.
	\begin{compactenum}
		\item If $\{v_1, v_2, \ldots, v_n\}$ is the node sequence of a unique shortest path between nodes $v_1$ and $v_n$, then all nodes except for possibly $v_1$ and $v_n$ are the source node of the tournament shared with the next node in the sequence.  
		\item Between any two distinct nodes $u,v$ in $\T$ the unique shortest trail between them is either $p(u,v)$ or $p(v,u)$ or there exists a node $w \in V \setminus \{u, v\}$ such that the trail is composed of the two shortest paths $p(w,u)$ and $p(w,v)$.
	\end{compactenum}
\end{lemma}
	

	\subsection{Max-linear structural equation model on a ttt}
	
	
Consider a directed graph, $(V,E)$. To each edge $e = (i, j) \in E$ we associate a weight $c_e = c_{ij} \in [0, \infty)$. 
	The product of the edge parameters over a directed path $p = \{e_1,\ldots,e_m\}$ is denoted by
	\[
	c_p=\prod_{r=1}^m c_{e_r}. 
	\] 
	When the product is over the unique shortest path from $u$ to $v$, we write $c_{p(u,v)}$. 
	The product over the empty set being one by convention, we have $c_{p(i,i)} = 1$.

	Let $(Z_i, i \in V)$ be a vector of independent unit-Fréchet random variables, i.e., $\P(Z_i \le z) = \exp(-1/z)$ for $z > 0$. In the spirit of \citep{gissibl2018max}, a recursive max-linear model on a ttt, $\T$, 
	is defined by
	\begin{equation} \label{eqn:mlgm}
		X_v=\bigvee_{i\in \pa(v)}c_{iv}X_i\vee c_{vv}Z_v, \qquad v\in V.
	\end{equation}
where the parameters $c_e$, for $e\in E$, and $c_{vv}$, for $v\in V$, are positive. We interpret this constraint as follows: if $c_{ij} = 0$, the variable $X_j$ cannot be influenced by $X_i$ through edge $(i,j)$, and the edge could be removed from the graph. If $c_{vv} = 0$, the factor variable $Z_v$ does not influence $X_v$. We don't want to deal with such border cases, so we assume that all parameters in the model definition~\eqref{eqn:mlgm} are positive.
According to \citet[Theorem~2.2]{gissibl2018max} the expression in \eqref{eqn:mlgm} is equal also to 
\begin{equation}
	\label{eq:Xv}
	X_v = \bigvee_{i \in V} b_{vi} Z_i,
\end{equation}
with
\begin{equation} \label{eq:bvi}
	b_{vi} =
	\begin{cases}
		0 & \text{if $i \not\in \An(v)$,} \\
		c_{vv} & \text{if $i = v$,} \\
		c_{ii} \max_{p \in \pi(i, v)} c_p
		& \text{if $i \in \an(v)$.}
	\end{cases}
\end{equation}
The cumulative distribution function (cdf) of $X_v$ is $\P(X_v \le x) = \exp(-\sum_{i \in V} b_{vi} / x)$ for $x > 0$. 	We assume that $(X_v, v\in V)$ are unit-Fréchet, yielding the constraint
\begin{equation}
	\label{eq:sumibvi}
	\sum_{i \in V} b_{vi} = \sum_{i \in \An(v)} b_{vi} = 1,
	\qquad \forall v \in V,
\end{equation}
since $b_{vi} = 0$ whenever $i \notin \An(v)$. It is thus necessary and sufficient to have
\begin{equation}
	\label{eq:cvv}
	c_{vv} = 
	1 - \sum_{i \in \an(v)} c_{ii} \max_{p \in \pi(i,v)} c_p,
\end{equation}
with $c_{vv}=1$ if $\an(v)=\varnothing$.
By~\eqref{eq:cvv}, the coefficients $c_{vv}$ for $v \in V$ are determined recursively by the edge weights $c_{e}$ for $e \in E$.
If $c_{iv} \ge 1$ for some $(i, v) \in E$, then \eqref{eqn:mlgm} implies that $X_v \ge X_i \vee c_{vv} Z_v$, and the constraint that $X_i$ and $X_v$ are unit-Fréchet distributed implies that $c_{vv} = 0$, a case we want to exclude, as explained above.
This is why we impose $0 < c_{e} < 1$ for all $e \in E$ from the start, yielding the parameter space
\begin{equation*}
	\mathring{\Theta}	
	= \left\{ 
	\theta = (c_e, e \in E) \in (0,1)^{E} : \ 
	\forall v \in V, \, c_{vv} > 0
	\right\}.
\end{equation*}

The notion of \emph{criticality} is important for max-linear structural equation models. We refer to \citet*{gissibl2018max}, \citet*{amendola2022conditional}, \citet*{gissibl2021identifiability} and \citet*{kluppelberg2019bayesian}  
for examples where different conditional independence relations arise depending on which path is critical, or for illustrations in the context of graph learning. 
According to \citet[Definition~3.1]{gissibl2018max}, a path $p \in \pi(i, v)$ is \emph{max-weighted} under $\theta \in \Theta$ if it realizes the maximum $\max_{p' \in \pi(i, v)} c_{p'}$, where $p'$ is any path in $\pi(i,v)$.  In \citet{amendola2022conditional} the term \emph{critical} is preferred.

If there is a (directed) path between two nodes, there is a unique shortest (directed) path between them (Lemma~\ref{lem:oned-1}-1). This is crucial for our parametric model. We define the \emph{critical parameter space} $\Theta_* \subset (0,1)^{E}$ as the set of parameters $\theta =(c_e, e \in E)$, 
such that for every $v \in V$ and every $i \in \an(v)$, the unique shortest directed path from $i$ to $v$ is the only critical path. Therefore we have $c_{p(i,v)} > c_{p}$, with strict inequality for any $p \in \pi(i,v)$ different from $p(i, v)$. Formally,
\begin{equation*}
	\label{eq:Theta_star}
	\Theta_*
	= \left\{ 
	\theta \in (0,1)^{E} : \ 
		\forall v \in V, \, 
		\forall i \in \an(v), \, 
		\forall p \in \pi(i, v) \setminus \{p(i, v)\}, \; 
		c_{p(i, v)} > c_p  
	\right\}.
\end{equation*}
	Next, we consider the intersection of the two spaces as an appropriate parameter space for our max-linear structural equation model:
	\begin{equation} \label{eqn:theta_ring_star}
		\mathring{\Theta}_*=\mathring{\Theta}\cap \Theta_*.
	\end{equation}
	For $\theta \in \mathring{\Theta}_*$, every element of the max-linear coefficient matrix $B_{\theta}=(b_{vi})_{v,i\in V}$ can be rewritten using an edge weight product over the unique shortest path $p(i,v)$ via
	\begin{equation}
		\label{eq:bvi-rev1}
		b_{vi} =
		\begin{cases}
			0 & \text{if $i \not\in \An(v)$,} \\
			c_{vv} & \text{if $i = v$,} \\
			c_{ii}c_{p(i,v)}
			& \text{if $i \in \an(v)$},
		\end{cases}
		\qquad \text{and} \qquad
		c_{vv} = 
			1 - \sum_{i \in \an(v)} c_{ii}c_{p(i,v)}.
	\end{equation}
	Also, note that $b_{ii}=c_{ii}$, leading to the frequently used expression
	\[
		b_{vi} = c_{p(i,v)} b_{ii}, \qquad i \in \an(v).
	\] 
	
\begin{example}[Criticality]
	The following example shows what happens if the assumption that all shortest paths are critical is omitted.	Consider a max-linear model on three nodes $\{1,2,3\}$ and three edges $\{(1,2), (2,3),(1,3)\}$. The corresponding edge weights are $c_{12}, c_{23}, c_{13}$. We have
	\begin{align*}
		X_1=c_{11}Z_1, \qquad 
		X_2=c_{12}X_1\vee c_{22}Z_2, \qquad 
		X_3=c_{13}X_1\vee c_{23}X_2 \vee c_{33}Z_3.
	\end{align*}  
	The coefficient matrix $B=\{b_{iv}\}$ from \eqref{eq:bvi} together with \eqref{eq:sumibvi} and \eqref{eq:cvv} is
	\begin{equation*}
		B=\begin{bmatrix}
			1&0&0\\
			c_{12}&(1-c_{12})&0\\
			c_{12}c_{23}\vee c_{13}&c_{23}(1-c_{12})
			&1-c_{12}c_{23}\vee c_{13}-c_{23}(1-c_{12})
		\end{bmatrix}.
	\end{equation*}
	If the shortest path $p=\{(1,3)\}$ from node~$1$ to node~$3$ is not critical then we have $b_{31}=c_{12}c_{23}$ and also $b_{33}=1-c_{23}$. In this way the coefficient $c_{13}$ has completely left the model. When considering the identifiability problem, we cannot hope to identify a coefficient from some marginal distribution if it is not even identifiable from the full one.
\end{example}
	
	Now, all elements are in place to describe our main object of interest.
	
	\begin{assump}[Max-linear structural equation model on a ttt]
		\label{ass:max-mod}
		The random vector $X = (X_v, v \in V)$ has the max-linear representation in \eqref{eq:Xv} and \eqref{eq:bvi-rev1} with respect to the ttt $\T = (V, E)$ (Definition~\ref{def:ttt}) where $(Z_v, v \in V)$ is a vector of independent unit-Fréchet random variables and the edge weight vector $\theta = (c_e, e \in E)$ belongs to $\mathring{\Theta}_*$ in \eqref{eqn:theta_ring_star}.
	\end{assump}

The following identity for nodes with a unique parent will be useful:
\begin{equation}
\label{eqn:unique_parent} 
	\pa(v) = \{i\} \implies b_{vv}=1-c_{iv}.
\end{equation}
Indeed, if $i$ is the only parent of $v$, then $X_v = c_{iv} X_i \vee c_{vv} Z_v$ by \eqref{eqn:mlgm}. The variables $X_v, X_i, Z_v$ are unit-Fréchet distributed and $X_i$ is independent of $Z_v$, since $X_i$ is a function of $(Z_u, u \in \An(i))$ and $v \not\in \An(i)$. Hence $ c_{iv} + c_{vv} = 1 $, and because  $c_{vv} = b_{vv}$, Eq.~\eqref{eqn:unique_parent} follows.

A notational convention: in case of double subscripts, we may also write $x_{i_1,i_2}$ instead of $x_{i_1i_2}$.


	

\subsection{The angular measure}
	
Let $X$ follow a max-linear model with parameter vector $\theta$ as in Assumption~\ref{ass:max-mod}. The joint distribution $P_\theta$ of $X$ on $[0, \infty)^V$ is max-stable and has unit-Fréchet margins. It is determined by
	\[
		P_\theta([0, z]) = \P(X \le z)  
	= \exp\left( - l_\theta\left((1/z_v)_{v \in V}\right)\right), 
	\qquad z \in (0, \infty]^V,
	\]
	where the stable tail dependence function (stdf) $l_{\theta} : [0, \infty)^V \to [0, \infty)$ is
	\begin{equation} \label{eqn:ltheta}
		l_\theta(x) 
		= \sum_{i \in V} \max_{v \in V} \left( b_{vi} x_v \right)
	\end{equation}
	for $x = (x_v)_{v \in V} \in [0, \infty)^V$ \citep{einmahl2012an}.
	
	Let $H_{\theta}$ be the angular measure on the unit simplex $\Delta_V = \{ a \in [0, 1]^V : \sum_{v \in V} a^{(v)} = 1\}$ corresponding to the stdf $l_{\theta}$. The link between the stdf and the angular measure is detailed in \citet{haan2007extreme} for the bivariate case and in \citet[Chapter~5]{resnick1987extreme} and \citet[Chapters~7--8]{beirlant2004statistics} for higher dimensions: we have
	\begin{equation*} 
		l_\theta(x) = \int_{\Delta_V} \max_{v \in V} {( a^{(v)} x_v )} \, \diff H_\theta(a).
	\end{equation*}
	In view of the expression of $l_{\theta}$ in \eqref{eqn:ltheta}, the angular measure is discrete and satisfies
	\begin{equation}
		\label{eq:Hth}
		H_\theta = \sum_{i \in V} m_i \delta_{a_i},
	\end{equation}
	with masses $m_i = \sum_{v \in V} b_{vi}$ and atoms $a_i = (b_{vi} / m_i)_{v \in V} \in \Delta_V$ for $i \in V$ \citep*[page~1779]{einmahl2012an}. The notation $\delta_{x}$ refers to a unit point mass at $x$. 
	
	If $X$ follows a max-linear model, the angular measure of $X$ is identifiable from its distribution $P_\theta$ via the limit relation
	\begin{equation*}
		\label{eq:tPth2Hth}
		t \P \left(
		\frac{1}{\norm{X}_1} X \in \point, \, 
		\norm{X}_1 > t
		\right)
		\xrightarrow{w} H_{\theta}(\point), \qquad t \to \infty,
	\end{equation*}
	where $\norm{x}_1 = \sum_i |x_i|$ for a vector $x$ in Euclidean space, while the arrow $\xrightarrow{w}$ denotes weak convergence of finite Borel measures, in this case on $\Delta_V$.

	When we discuss latent variables and identifiability in Section~\ref{sec:ident}, we have to deal with the angular measure of a subvector of $X$, say $X_U = (X_v)_{v \in U}$, for non-empty $U \subset V$. Its stdf $l_{\theta, U}$ arises from $l_\theta$ by setting $x_v = 0$ for all $v \not\in U$: for $x \in [0, \infty)^U$ we have
	\begin{equation*}
		l_{\theta,U}(x) 
		= \sum_{i \in V} \max_{v \in U} {(b_{vi} x_v)}
		= \int_{\Delta_U} \max_{v \in U} {(a^{(v)} x_v)} \, \diff H_{\theta,U}(a).
	\end{equation*}
	The distribution of $X_U$ is max-linear too, so that its angular measure $H_{\theta,U}$ on $\Delta_U$ has a similar form as the one of $X$:
	\begin{equation}
		\label{eq:HthU}
		H_{\theta,U} = \sum_{i \in V} m_{i,U} \delta_{a_{i,U}},
	\end{equation}
	with masses $m_{i,U} = \sum_{v \in U} b_{vi}$ and atoms $a_{i,U} = (b_{vi} / m_{i,U})_{v \in U} \in \Delta_U$ for $i \in V$.

	
	\section{Conditional tail limit and the ttt with unique source} \label{sec:ctc}

	Here we study the limit distribution of 
	\begin{equation} \label{eqn:XuXv1}
		\left(\frac{X_v}{X_u}, v\in V \mathrel{\Big|} X_u>t \right), \qquad t\rightarrow\infty,
	\end{equation}
	when $X$ is a max-linear model with respect to a ttt $\T = (V, E)$ as in Assumption~\ref{ass:max-mod}. In particular, we are interested to know whether the elements of the limiting vector of \eqref{eqn:XuXv1} can be factorized into products of independent increments, similarly to other models with this property as in \citet{segers2020conditional} and \citet{asenova2021extremes}. 
	In Proposition~\ref{prop:factorize} below, we show that the limit variables factorize according to the unique shortest trails under the condition that the ttt has a unique source (node without parents). Moreover, by Proposition~\ref{prop:markov_ml}, the latter criterion is necessary and sufficient for $X$ to satisfy the global Markov property with respect to the skeleton graph associated to $\T$, i.e., the undirected counterpart of $\T$.

		Even though Proposition~\ref{prop:factorize} below looks similar to Theorem~3.5 in \cite{asenova2021extremes}, it does not follow from it. The reason is that we have not been able to verify Assumptions~3.1 and~3.4 in that article for the recursive max-linear model. In these assumptions, the conditioning event involves equality, i.e., $\{X_u=t\}$, and calculating the conditional distributions and their limits is not easy. This is why we have opted here for a different route: in~\eqref{eqn:XuXv1}, the conditioning event is $\{X_u>t\}$ and the limit conditional distribution as $t \to \infty$ is found from \citet[Example~1]{segers2020one}.
	
		According to property~\ref{ttt:prop3}, any pair of distinct nodes in a ttt is connected by a unique shortest trail. Let $t(u,v)$ denote the set of edges along the unique shortest trail between two distinct nodes $u$ and $v$. Consider for instance the shortest trail between nodes 2 and 8 on Figure~\ref{fig:ttt_unique}: $t(2,8)=\{(3,7), (7,8), (3,2)\}$. In contrast, let $t_{u}(u,v)$ be the set of edges incident to the same node set but directed from $u$ to $v$, irrespective of their original directions, e.g., $t_2(2,8)=\{(2,3), (3,7), (7, 8)\}$.

	For a given node $u \in V$, let $E_u$ be the set of all edges in such unique shortest paths directed away from $u$, that is,
	\begin{equation} \label{eqn:Eu}
		E_u=\bigcup_{v\in V\setminus u} t_u(u,v).
	\end{equation}
	Recall from Section~\ref{sec:def} that $\mathbb{T}$ denotes the set of tournaments within the ttt $\T$.
	For fixed $u \in V$ there is for every tournament $\tau = (V_\tau, E_\tau) \in \mathbb{T}$ a node, say $w_{u,\tau}$, which is the unique node in $V_\tau$ such that the trail $t(u, w_{u, \tau})$ is the shortest one among all trails between $u$ and a node $v$ in $V_\tau$. As an example, consider Figure~\ref{fig:conv_on_ttt}: starting from node $u = 8$, the closest node from the node set $V_{\tau_1}=\{1,2,3\}$ is $3$, hence $w_{8,\tau_1}=3$.  

	With these definitions we are ready to state the condition under which the limiting variables factorize into independent increments.
	
	\begin{prop}[Factorization in max-linear model] \label{prop:factorize}
		Let $(X_v, v\in V)$ follow a max-linear model as in Assumption~\ref{ass:max-mod}. Fix $u\in V$. Let $E_u$ be as in \eqref{eqn:Eu} and let $(M_e, e\in E_u)$ be a random vector composed of mutually independent subvectors $M^{(u,\tau)}=\bigl( M_{w_{u,\tau},j}: j \in V_\tau, (w_{u,\tau},j)\in E_u \bigr)$, one for every transitive tournament $\tau\in \mathbb{T}$, and with marginal distribution as in Lemma~\ref{lem:Muv-rev1}.
		
		
		The following statements are equivalent:
		\begin{enumerate}[label=(\roman*)]
			\item $\T$ has a unique source.
			\item For every $u \in V$, we have, as $t \to \infty$, the weak convergence P5Ttk2JNYkwUnwBcR6bu
			\begin{equation} \label{eqn:LXvXu}
				\mathcal{L}(X_v/X_u, v\in V\mid X_u>t)\inlaw
				\mathcal{L}(A^{(u)})
				=\mathcal{L}(A_{uv}, v\in V)
			\end{equation}
			with 
			\begin{equation} \label{eqn:Auvprod}
				A_{uv}=\prod_{e\in t_u(u,v)}M_e, \qquad v\in V.
			\end{equation}
			\item There exists $u\in V$ such that the limit in \eqref{eqn:LXvXu} and \eqref{eqn:Auvprod} holds.
		\end{enumerate}
	\end{prop}

The following lemma provides the distribution of $M^{(u,\tau)}$ in Proposition~\ref{prop:factorize}.

	\begin{lemma} \label{lem:Muv-rev1}
	Let $(X_v, v\in V)$ follow a max-linear model as in Assumption~\ref{ass:max-mod}.
	Let $\tau \in \mathbb{T}$ be a transitive tournament on nodes $V_{\tau}$.
	Then for $u\in V_{\tau}$, we have 
	\begin{align} \label{eqn:LXLM}
		\mathcal{L}\left(\frac{X_v}{X_u}, v\in V_{\tau} \mathrel{\Big|} X_u>t\right)
		\inlaw
		\mathcal{L}(M^{(u,\tau)})
		=\mathcal{L}(M_{uv}, v\in V_{\tau})
		=\sum_{j\in \An(u)}b_{uj}
		\delta_{\left\{\frac{c_{p(j,v)}}{c_{p(j,u)}}, v\in V_{\tau}\right\}}.
	\end{align}
	The vector $M^{(u,\tau)}=(M_{uv}, v\in V_{\tau})$ has dependent variables and the distribution of a single element is as follows.
	\begin{compactenum}
		\item The distribution of $M_{uv}$ when $(u,v)\in E$.
		\begin{compactenum}
			\item If $u$ is the source node of $\tau$, the distribution is given by $\mathcal{L}(M_{uv})=\delta_{\{c_{uv}\}}$.
			
			\item If $u$ is not the source node of $\tau$, the distribution is given by 
			\[
			\mathcal{L}(M_{uv})=\sum_{j\in \An(u)}b_{uj}
			\delta_{\left\{\frac{c_{p(j,v)}}{c_{p(j,u)}}\right\}}.
			\]
		\end{compactenum}
		\item The distribution of $M_{uv}$ when $(v,u)\in E$.
		\begin{compactenum}
			\item If $v$ is the source node of $\tau$, the distribution is given by 
			\[
			\mathcal{L}(M_{uv})=c_{vu}\delta_{\{1/c_{vu}\}}
			+(1-c_{vu})\delta_{\{0\}}.
			\]
			\item If $v$ is not the source node of $\tau$, the distribution is given by
			\[
			\mathcal{L}(M_{uv})
			=\sum_{j\in \An(v)}
			b_{uj}\delta_{\left\{\frac{c_{p(j,v)}}{c_{p(j,u)}}\right\}}
			+\sum_{j\in \An(u)\setminus \An(v)}
			b_{uj}\delta_{\{0\}}.
			\]
		\end{compactenum}
	\end{compactenum}
\end{lemma}

	According to Proposition~\ref{prop:factorize}, the factorization property \eqref{eqn:Auvprod} holds either for all nodes or for no node at all, a necessary and sufficient condition being that the ttt has a unique source.	
	The principle of \eqref{eqn:Auvprod} is illustrated in Figure~\ref{fig:conv_on_ttt} for $u = 8$. 
	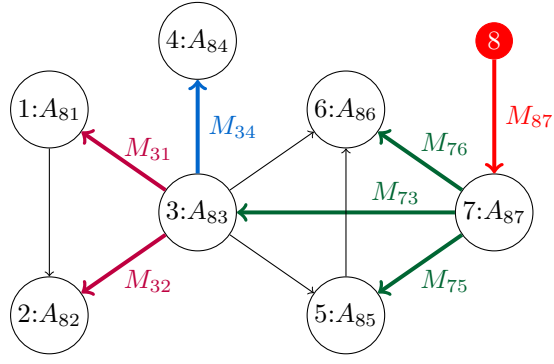
\begin{figure}[]
		\centering
		\small 
		\begin{tikzpicture}[scale=1.3]
			\node[hollow node] (1) at (0,0.3) {1:$A_{81}$};
			\node[hollow node] (2) at (0,-1.8)  {2:$A_{82}$};
			\node[hollow node] (3) at (1.5,-0.75)  {3:$A_{83}$};
			\node[hollow node] (4) at (1.5,1)  {4:$A_{84}$};
			\node[hollow node] (7) at (4.5,-0.75)  {7:$A_{87}$};
			\node[hollow node] (5) at (3,-1.8)  {5:$A_{85}$};
			\node[hollow node] (6) at (3,0.3)  {6:$A_{86}$};
			\node[solid node, red] (8) at (4.5,1)  {\textcolor{white}{8}};
			\path[->] (1) edge (2)  ;
			\path[<-, color=purple, line width=1.5] (1) edge (3);
			\node[color=purple] (m1) at (1,-0.1)  {$M_{31}$};
			\path[<-, color=purple, line width=1.5] (2) edge (3)  ;
			\path[->] (3) edge (5);
			\path[->, color=ggr, line width=1.5] (7) edge (3);
			\node[color=ggr] (m3) at (3.5,-0.55)  {$M_{73}$};
			\node[color=ggr] (m5) at (4,-1.45)  {$M_{75}$};
			\path[->] (3) edge (6);
			\path[->] (5) edge (6);
			\path[<-, color=ggr, line width=1.5] (5) edge (7);
			\node[color=ggr] (m6) at (4,-0.05)  {$M_{76}$};
			\node[color=purple] (m2) at (1,-1.45)  {$M_{32}$};
			\path[<-, color=bbl, line width=1.5] (4) edge node[anchor=west] {$M_{34}$} (3);
			\path[<-, color=red, line width=1.5] (7) edge node[anchor=west] {$M_{87}$} (8);
			\path[->, color=ggr, line width=1.5] (7) edge (6);
		\end{tikzpicture}
		\centering
		\caption{A ttt on four tournaments: $\tau_1$ on node set $\{1,2,3\}$, $\tau_2$ on $\{3,4\}$, $\tau_3$ on $\{3,5,6,7\}$ and $\tau_4$ on $\{8,7\}$. The variable exceeding a high threshold is at node 8. The set $E_u$ is composed of the coloured edges which do not necessarily have the same directions in the original graph. On the nodes we have $A^{(8)}=(A_{8i}, i=1, \ldots, 7)$ and on the edges we have the multiplicative increments $(M_{e}, e\in E_u)$. Increments in different colours are mutually independent, while those in the same color are dependent according to Lemma~\ref{lem:Muv-rev1}.} 
		\label{fig:conv_on_ttt}
	\end{figure}
	The limit $A^{(u)}=(A_{uv}, v\in V\setminus u)$ is given by
	\begin{align*}
		A_{87}&=\textcolor{red}{M_{87}}, 		&A_{83}&=\textcolor{red}{M_{87}}\textcolor{ggr}{M_{73}},
		&A_{82}&=\textcolor{red}{M_{87}}\textcolor{ggr}{M_{73}}\textcolor{purple}{M_{32}},\\
		&&A_{86}&=\textcolor{red}{M_{87}}\textcolor{ggr}{M_{76}},
		&A_{81}&=\textcolor{red}{M_{87}}\textcolor{ggr}{M_{73}}\textcolor{purple}{M_{31}},\\
		&&A_{85}&=\textcolor{red}{M_{87}}\textcolor{ggr}{M_{75}},
		&A_{84}&=\textcolor{red}{M_{87}}\textcolor{ggr}{M_{73}}\textcolor{bbl}{M_{34}},
	\end{align*}
	where $M^{(8,\tau_4)}=\textcolor{red}{M_{87}}$, $M^{(8,\tau_3)}=(\textcolor{ggr}{M_{76}}, \textcolor{ggr}{M_{73}}, \textcolor{ggr}{M_{75}})$, $M^{(8,\tau_2)}=\textcolor{bbl}{M_{34}}$ and $M^{(8,\tau_1)}=(\textcolor{purple}{M_{31}}, \textcolor{purple}{M_{32}})$ are independent sub-vectors by construction.
	
	What underlies the link between the factorization of the limiting variables from Proposition~\ref{prop:factorize} on the one hand and the uniqueness of the source of the ttt on the other hand is the Markovianity of $X$ with respect to the skeleton graph $T$. The Markov property states that for any three non-empty and disjoint sets $A,B,C\subset V$ such that in the graph $T$ the nodes in $A$ are separated from the nodes in $B$ by the nodes in $C$, the vector $X_A=(X_v, v\in A)$ is conditionally independent from $X_B$ given $X_C$ \citep{lauritzen1996graphical}. Another equivalence condition can be added to the list in Proposition~\ref{prop:factorize}.
	
	
	\begin{prop} \label{prop:markov_ml}
		Let $X$ follow a max-linear model with respect to the ttt $\T$ as in Assumption~\ref{ass:max-mod}. 
Then $X$ satisfies the global Markov property with respect to the skeleton graph $T$ if and only if $\T$ has a unique source. 
	\end{prop}

		Even though in Proposition~\ref{prop:markov_ml} we consider the undirected graph $T$ associated to the ttt $\T$, the recursive max-linear specification of $X$ is still with respect to the directed graph $\T$ itself. Indeed, the latter edges' directions are intrinsically determined by the recursive max-linear model specification of $X = (X_v)_{v \in V}$. When we also consider the associated skeleton graph $T$, i.e., without directions, it is because, in view of the factorization property in Proposition~\ref{prop:factorize}, we are interested in whether the global Markov property holds, a property which is most easily described in terms of the skeleton graph $T$.
	
	The proof of Proposition~\ref{prop:markov_ml} is based on notions and results from \citet{amendola2022conditional} which provides an extensive study of conditional independence properties of max-linear models. In particular, the notion of \emph{$*$-connecting} path between two nodes in a DAG is introduced, a notion which is similar to the one of an \emph{active} path \citep[Definition~3.6]{koller2009probabilistic} between two nodes.

\section{Latent variables and parameter identifiability} 
\label{sec:ident}
	
In practice, it is possible that on some of the nodes, the variables of interest are not observed (latent). Examples from the literature are water heights on certain locations on the river networks of the Danube in \citet{asadi2015extremes} and the Seine in \citet{asenova2021inference}.
We look at the problem of recovering all parameters of the distribution of the complete vector, based on the distribution of the observed variables only. If this is possible, we can study the parametric model as if all variables were observed: in particular, 
we are able to compute measures of tail dependence for sets including the unobserved variables. The latter is important as it may be the only possible way to quantify tail dependence, because non-parametric estimates are not available when dealing with unobserved variables. 
	
Consider for instance the network in Figure~\ref{fig:ident_ml}. The max-linear model on $\T=(V,E)$ has eight variables and eleven parameters $\theta=(c_e, e\in E)$. By Proposition~\ref{prop:ident_block} below, the parameter $\theta \in \mathring{\Theta}_*$ can be uniquely identified in case $X_1, X_3, X_7$ are not observed on the basis of the joint distribution of the remaining five variables, $X_U=(X_{2}, X_4,X_5,X_6,X_8)$.

The problem of parameter identifiability will be formalized on the level of the angular measure $H_{\theta}$ and is presented in detail in the next two subsections.

\subsection{Graph-induced characteristics of the angular measure} \label{ssec:ang_compl}
	
In this subsection, we argue that the condition $\theta=(c_e, e\in E)\in \mathring{\Theta}_*$ guarantees that all edge weights in $\theta$ are uniquely identifiable from the angular measure $H_{\theta}$ of $X = (X_v, v\in V)$ and thus from the distribution $P_\theta$ of $X$.
Recall from \eqref{eq:Hth} that $H_{\theta}$ is discrete with atoms $a_i = (a_{vi})_{v \in V} \in \Delta_V$ and masses $m_i > 0$.
	
	

	Thanks to the assumption $\theta\in \mathring{\Theta}_*$, 
we have
\begin{equation}
	\label{eq:equivs}
	a_{vi} > 0 \iff b_{vi} > 0 \iff i \in \An(v) \iff v \in \Desc(i).
\end{equation}
For any DAG, all nodes have a different set of descendants, i.e.,
\begin{equation}
	\label{eq:Descdistinct}
	\forall i, j \in V : i \ne j \implies \Desc(i) \neq \Desc(j).
\end{equation}
Indeed, if $i \ne j$ and $\Desc(i) \subseteq \Desc(j)$, then $i \in \desc(j)$ and hence $j \not\in \desc(i)$, so that $\Desc(j) \not\subseteq \Desc(i)$.

\begin{lemma} 
	\label{lem:Hth}
	Let $(X_v, v \in V)$ follow a max-linear model as in Assumption~\ref{ass:max-mod}, with parameter vector $\theta \in \mathring{\Theta}_*$ and induced coefficient matrix $(b_{vi})_{i,v \in V}$. Let $H_\theta = \sum_{i \in V} m_i \delta_{a_i}$ in \eqref{eq:Hth} be its angular measure. Then
	\begin{enumerate}[label=(\arabic*)]
	\item $m_i > 0$ for all $i \in V$; 
	\item for any atom $a_i = (a_{vi})_{v \in V}$, we have $a_{vi} > 0$ if and only if $v \in \Desc(i)$. Specifically, all $|V|$ vectors $a_i$ are different and every atom can be matched uniquely to a node in $V$;
	\item for each edge $(i, v) \in E$, we have $c_{iv} = b_{vi} / b_{ii} = a_{vi} / a_{ii}$.
	\end{enumerate}
	In particular, $\theta \in \mathring{\Theta}_*$ is identifiable from $H_\theta$ and thus from $P_\theta$, i.e., for $\theta_1 \neq \theta_2\in \mathring{\Theta}_*$ we have $H_{\theta_1} \neq H_{\theta_2}$ and thus $P_{\theta_1} \neq P_{\theta_2}$.
\end{lemma}

In Lemma~\ref{lem:Hth}, if the edge $(i, v)$ is not critical, then there is another path, say $p'$, from $i$ to $v$ with path product $c_{p'} \ge c_{iv}$, and then we can further lower the value of $c_{iv}$ without changing the coefficients in \eqref{eq:bvi}, because they involve $c_{p'}$ rather than $c_{iv}$, thus yielding the same measure $H_\theta$. This shows that without the criticality assumption, some edge weights may not be identifiable from $H_\theta$.

	\begin{example}[Unique zero patterns] \label{ex:ident} 
		In dimension $d = 3$, consider an angular measure given by the following atoms and masses:
		\begin{equation*}
			\omega_1=\frac{1}{2.2}\begin{bmatrix}0.8\\1\\0.4\end{bmatrix}, \; \mu_1=2.2, 
			\qquad 
			\omega_2=\frac{1}{0.5}\begin{bmatrix}0\\0\\0.5\end{bmatrix}, \;
			\mu_2= 0.5,
			\qquad
			\omega_3= \frac{1}{0.3}\begin{bmatrix}0.2\\0\\0.1\end{bmatrix}, \;
			\mu_3= 0.3.
		\end{equation*}	
		Consider the vectors $\beta_j = \mu_j \omega_j$ for $j \in \{1,2,3\}$.
		By Lemma~\ref{lem:Hth}, the unordered collection $\{ \beta_1, \beta_2, \beta_3 \} = \{(0.8, 1, 0.4)^\top, (0, 0, 0.5)^\top, (0.2, 0, 0.1)^\top\}$ permits to recover the values of the coefficients in the max-linear model
		\begin{align*}
			X_1=c_{11}Z\vee c_{21}c_{22}Y,\quad
			X_2=c_{22}Y,\quad
			X_3=c_{13}c_{11}Z\vee c_{13}c_{21}c_{22}Y\vee c_{33}T.
		\end{align*}
		with (known) edge set $E = \{(2, 1), (1, 3)\}$, and this due the presence of zeroes in the vectors. For the current example, argue as follows. The angular measure $H_{\theta}$ of $(X_1,X_2, X_3)$ has three atoms: atom $a_Z=b_Z/m_Z$ with $b_{Z}=(c_{11}, 0, c_{13}c_{11})^\top$, atom $a_Y=b_Y/m_Y$ with $b_{Y}=(c_{21}c_{22},\, c_{22},\, c_{13}c_{21}c_{22})^\top$, and atom $a_T=b_T/m_T$ with $b_{T}=(0,0,c_{33})^\top$.
As unordered sets, $\{\beta_1, \beta_2, \beta_3\}$ and $\{b_Z, b_Y, b_T\}$ are equal, but the question is which vector $\beta_j$ corresponds to which vector $b_{*}$.
From an inspection of the zero entries of the vectors, it is easily seen that the only possibility to identify the three coefficient vectors $\beta_1, \beta_2, \beta_3$ with the vectors $b_Z, b_Y, b_T$ of the angular measure $H_\theta$ is 
		\begin{align*}
			\beta_1=\begin{bmatrix}0.8\\1\\0.4\end{bmatrix} &= 
			\begin{bmatrix}c_{21}c_{22}\\ c_{22}\\ c_{13}c_{21}c_{22}\end{bmatrix} = b_Y, &
			\beta_2=\begin{bmatrix}0\\0\\0.5\end{bmatrix} &=
			\begin{bmatrix} 0\\0\\c_{33}\end{bmatrix} = b_T, &
			\beta_3=\begin{bmatrix}0.2\\0\\0.1\end{bmatrix} &=
			\begin{bmatrix} c_{11}\\ 0\\ c_{13}c_{11}\end{bmatrix} = b_Z.
		\end{align*}
		Solving the equations yields $(c_{11}, c_{21}, c_{22}, c_{13}, c_{33})=(0.2, 0.8, 1, 0.5, 0.5)$.
	\end{example}
	
	
	\subsection{Identifiability issues with the angular measure of a subvector}
	
	When we deal with latent variables, we know the distribution of the observable variables only, $X_U=(X_v, v\in U)$ for non-empty $U \subset V$. 
	The angular measure, say  $H_{\theta,U}$, of $X_U$ in \eqref{eq:HthU} is discrete and takes the form
	\begin{equation} \label{eqn:Hmuomega}
		H_{\theta, U} = \sum_{r=1}^s \mu_r \delta_{\omega_r},
	\end{equation}
	with masses $\mu_r > 0$ and $s$ distinct atoms $\omega_r \in \Delta_U$. Combining \eqref{eq:HthU} and \eqref{eqn:Hmuomega}, we should have 
	\begin{equation} \label{eqn:muom_ma}
		\sum_{r=1}^s \mu_r \delta_{\omega_r}
		=
		\sum_{i \in V} m_{i,U} \delta_{a_{i,U}},
	\end{equation}
	which means that, as sets, we should have $\{ \omega_1, \ldots, \omega_s \}
	=
	\{ a_{i,U} : i \in V \}$. In contrast to the situation in Lemma~\ref{lem:Hth}, the subvectors $a_{i,U}$ for $i \in V$ are not necessarily all different. Any atom $\omega_r$ of $H_{\theta, U}$ is of the form $a_{i,U} = (b_{vi}/m_{i,U})_{v \in U}$ for one or possibly several indices $i \in V$. For $r=1,\ldots, s$ and $i \in V$ such that $\omega_r = a_{i, U}$, we know from \eqref{eq:equivs} that
	\begin{equation}
		\label{eq:omegaiU}
		\{ v \in U : \omega_{r,v} > 0 \} = \Desc(i) \cap U.
	\end{equation} 
	The (unordered) collection of vectors $\{(b_{vi})_{v\in U} : i \in V\}$ will be denoted by $\cB_{\theta,U}$.
	
	With unobservable variables, there are several issues with the angular measure and its expression on the right hand-side of \eqref{eqn:muom_ma}.
	\begin{itemize}
		\item \emph{Zero masses.} We have $m_{i,U}=\sum_{v\in U}b_{vi}$, so that if all components of $(b_{vi})_{v\in U}$ are zero, then $m_{i,U}=0$. This happens when $\Desc(i)\cap U=\varnothing$. In this case, we have $s<|V|$, i.e., $H_{\theta,U}$ has less atoms than $H_\theta$.
		\item \emph{Equal atoms.} We may have $a_{i,U}=a_{j,U}$ for some indices $i,j\in V$ and $i\neq j$. In this case, the terms $i$ and $j$ in \eqref{eq:HthU} are to be aggregated and again, $H_{\theta,U}$ has less than $|V|$ atoms, $s<|V|$. This happens when the vectors $(b_{vi}, v\in U)$ and $(b_{vj}, v\in U)$ are proportional for some distinct $i,j\in V$.
		\item \emph{Zeroes on the same positions}.  A more subtle problem occurs when for two distinct vectors $b,b'\in \cB_{\theta,U}$, the supports $\{v\in U:b_v>0\}$ and $\{v\in U:b'_v>0\}$ are equal. Such a situation arises when two distinct nodes $i,j\in V$ satisfy $\Desc(i)\cap U=\Desc(j)\cap U$. The latter equality is only possible in the presence of latent variables and is to be contrasted with property~\eqref{eq:Descdistinct} when all variables are observable.  
	\end{itemize}

	\subsection{Identifiability criterion}

	For a max-linear model with respect to a ttt $\T=(V,E)$ with unique source, we need conditions that ensure that the minimal representation of the angular measure of $X_U$ is the one in \eqref{eq:HthU}. Consider the following two conditions for the set of nodes $\bar{U} = V \setminus U$ carrying latent variables:
	\begin{enumerate}[label=(I\arabic*)]
		\item \label{ident1} any $u \in \bar{U}$ has at least two children;
		\item \label{ident2} any $u \in \bar{U}$ is the source of some tournament in $\T$.
	\end{enumerate}
	\begin{prop} \label{prop:ident_block}
		Let $X$ follow a max-linear model as in Assumption~\ref{ass:max-mod} with respect to a ttt $\T=(V,E)$ with unique source. 
		For a non-empty node set $U \subset V$, the parameter $\theta\in \mathring{\Theta}_*$ is uniquely identifiable from the distribution of $(X_v, v \in U)$ if and only if conditions~\ref{ident1} and~\ref{ident2} are satisfied.
	\end{prop}

	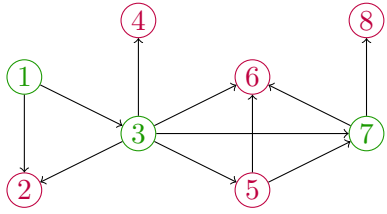
\begin{figure} 
		\centering 
		\begin{tikzpicture}
			\node[hollow node, color=mygreen] (1) at (0,0) {1};
			\node[hollow node, color=purple] (2) at (0,-1.5)  {2};
			\node[hollow node, color=mygreen] (3) at (1.5,-0.75)  {3};
			\node[hollow node, color=purple] (4) at (1.5,0.75)  {4};
			\node[hollow node, color=mygreen] (7) at (4.5,-0.75)  {7};
			\node[hollow node, purple] (5) at (3,-1.5)  {5};
			\node[color=purple, hollow node] (6) at (3,0)  {6};
			\node[hollow node, purple] (8) at (4.5,0.75)  {8};
			\path[->] (1) edge (2)  ;
			\path[->] (1) edge (3);
			\path[<-] (2) edge (3)  ;
			\path[->] (3) edge (5);
			\path[<-] (7) edge (3);
			\path[->] (3) edge (6);
			\path[->] (5) edge (6);
			\path[->] (5) edge (7);
			\path[<-] (4) edge (3);
			\path[->] (7) edge (8);
			\path[->] (7) edge (6);
		\end{tikzpicture}
		\caption{In the following ttt, the nodes that are allowed to contain a latent variable while the edge parameters remain identifiable are $1,3,7$. These are the only nodes where each of them satisfies both \ref{ident1} and \ref{ident2}. For instance, if node 2 has unobserved variables, the parameters attached to edges $(1,2),(3,2)$ are not identifiable. This is because the edge weights $c_{12},c_{32}$ take part only in products over paths ending at $2$. But if $2\in\bar{U}$ these coefficients disappear from the atoms of the angular measure $a_{i,U}=(b_{vi}/m_{i,U}, v\in U)$ and accordingly from the collection of vectors $\cB_{\theta,U}$.
} 
		\label{fig:ident_ml}
	\end{figure}
	
	Figure~\ref{fig:ident_ml} illustrates the identifiability criterion. 

	\section{Discussion}

	In this paper we have considered a Bayesian max-linear network over a special type of graph which we called a tree of transitive tournaments (ttt). It is a graph which collects in an acyclic manner transitive tournaments which are themselves complete DAGs. The max-linear model is defined on a particular parameter space which ensures that the impact from one variable to another takes place along the shortest path, a consideration that has been defined in the literature as the path's criticality. It turns out that a ttt with unique source leads to a graph without v-structures, that is, no node has non-adjacent parents. The limit of the scaled random vector, conditional on the event that a high threshold is exceeded at a particular node, is shown to be factorizable in independent multiplicative increments if and only if the ttt has a unique source. This result is analogous to that for Markov trees in \citet{segers2020one} and for Markov random fields on undirected block graphs in \citet{asenova2021extremes}. The property that the Bayesian max-linear model on a ttt with unique source shares with these two other models is that it satisfies the global Markov property with respect to the undirected counterpart or skeleton graph of the ttt.
	
		In addition, we have provided a simple necessary and sufficient criterion guaranteeing the identifiability of the edge coefficients in case some variables are latent. As suggested by a Reviewer, it may be possible to extend the criterion to partial identifiability of some edge weights in case the criterion is fulfilled only locally.
	
	Upon appropriate modifications, we expect the results presented in this paper to hold equally for the linear additive causal model introduced in \citet{gnecco2020causal}. One of the reasons is that the max-domain of attraction of a linear model with heavy-tailed factors is the same as that of a max-linear one \citep{einmahl2012an}. However, the relation between the edge weights $\theta = (c_e)_{e \in E}$ and the coefficient matrix $B_\theta = (b_{ij})_{i,j \in V}$ is different between the max-linear and additive linear versions, and this may ask for different approaches in showing the same properties for the additive version. 
	
	\appendix
	
	
	\section{Trees of transitive tournaments} \label{app:ttt}
	
	Recall that in a directed acyclic graph, a v-structure refers to a node with parents that are not adjacent, see Figure~\ref{fig:ttt}. 
	

	\subsection{Proof of Lemma~\ref{lem:oned-1}}
	\begin{proof}
		1. Let $a,b\in V$. If $a$ and $b$ share the same tournament, they must be connected by an arrow, which is then the unique shortest path between them, since all other possible paths have length larger than one.
		
		Let $a,b$ be nonadjacent. 
		If there is a unique directed path between $a$ and $b$ then this is the unique shortest path. Suppose now there are two shortest paths: $p_1, p_2\in \pi(a,b)$. Let the path $p_1$ be along the vertices $\{v_1=a, v_2, \ldots, v_n=b\}$ and the path $p_2$ on along the vertices $\{u_1=a, u_2, \ldots, u_n=b\}$. 
		
		We will proceed by contradiction. Assume $v_2\neq u_2$. If $v_2$ and $u_2$ belong to two different tournaments, then there exists a non-directed cycle through nodes in different tournaments, namely $\{a,v_2, \ldots, b,\ldots, u_2,a\}$. But this is impossible by property~\ref{ttt:prop2} of a ttt. 
		 Hence, $v_2$ and $u_2$ must belong to the same tournament, say $\tau_a$, because $a$ is part of the same tournament too. Now consider $u_3$ and $v_3$. Then either $u_3=v_3$ or they share a tournament, say $\tau_3$, because otherwise there exists a non-directed cycle through nodes in different tournaments. Since $(v_2, v_3)\in E$ and $(u_2, u_3)\in E$ and by the assumption that $v_2\neq u_2$, all four nodes $\{a, v_2, u_2, v_3=u_3\}$ or all five nodes $\{a, v_2, u_2, v_3, u_3\}$ belong to $\tau_a$. This is because by property~\ref{ttt:prop1}, two tournaments can share only one node, hence it is impossible to have $\tau_3\cap \tau_a=\{v_2, u_2\}$. Because all four or five nodes belong to the same tournaments and since $(a,v_2), (v_2,v_3), (a,u_2), (u_2, u_3)\in E$ we must have $(a,v_3)\in E$ and $(a,u_3)\in E$ to avoid inter-tournament undirected cycles. Hence the paths $\{a=v_1, v_3, \ldots, v_n=b\}$ and $\{u_1=a, u_3, \ldots, u_n=b\}$ are shorter then $p_1$ and $p_2$, a contradiction. Hence we must have $v_2 = u_2$. 
		
		We apply the same strategy to the nodes $v_3, u_3$ and $v_4, u_4$ to find that $v_3 = u_3$. Proceeding recursively, we conclude that $p_1 = p_2$. \smallskip
		
		2. First we show that if the ttt has a unique source, there cannot be a v-structure. We proceed by contraposition. Assume that there is a node, $v$, with parents in two different tournaments $\tau_a$ and $\tau_b$. Let $a$ and $b$ be the sources of $\tau_a$ and $\tau_b$ respectively \citep[Corollary~5a]{harari1966the}. Note that we definitely have $v\neq a$ and $v\neq b$. 
		From node $v$ go to node $a$. If $a$ doesn't have a parent from another tournament we have found one node with zero in-degree with respect to the whole graph. If $a$ has parent(s) from another tournament, say $\tau_{a}'$, then go to the node that within $\tau_a'$ has in-degree zero, say node $a'$. Keep on going until you find a node with in-degree zero within the whole graph---such a node must exist because the graph is finite. Repeat the same for $\tau_b$, yielding two different nodes having zero in-degree with respect to whole graph. These nodes must be different because of the definition of $\T$: since we have started in two different tournaments $\tau_a$ and $\tau_b$ we cannot end up in the same node, or otherwise there would be a non-directed cycle passing through $v$ and that node. Hence we have found two nodes with zero in-degree, hence $\T$ does not have a unique source node. 
		
		Next we show that if $\T$ has two or more source nodes, $u$ and $v$, then there is a v-structure. Because $u$ and $v$ are sources they have in-degree zero, so that they cannot belong to the same tournament, and thus they belong to two different tournaments. Consider the unique shortest trail between $u,v$ on a sequence of nodes $\{u=v_1, v_2, \ldots, v_n=v\}$. Such a trail exists as, by definition of a ttt, the skeleton of $\T$ is a block graph and the fact that in a block graph there is a unique shortest path between every two nodes \citep[Theorem~1]{behtoei2010a}. For every two consecutive nodes in the shortest path, $v_i,v_{i+1}$,  we have either $(v_i,v_{i+1})\in E$ or $(v_{i+1}, v_i)\in E$. Because $u$ and $v$ are sources of $\T$, we have $(u,v_2)\in E$ and $(v,v_{n-1})\in E$. Note that $n \ge 3$, since $u$ and $v$ cannot be adjacent. We need three nodes $v_i, v_{i+1}, v_{i+2}$ such that $(v_i,v_{i+1})\in E$ and $(v_{i+2}, v_{i+1})\in E$. If $n = 3$, then the triple $(u,v_2,v)$ already fulfils the requirement. If $n\geq 4$, then continue from $v_2$ as follows. 
		Let $i = \max \{ j = 1,\ldots,n-2 :(v_{j}, v_{j+1}) \in E\}$; then $(v_i,v_{i+1}) \in E$ and $(v_{i+2},v_{i+1}) \in E$, as required. Because this is the shortest trail, $v_i$ and $v_{i+2}$ cannot belong to the same tournament, since otherwise there would exist a shorter trail passing only through $v_i$ and $v_{i+2}$. \smallskip
		
		
		3. Suppose that $v \in \Desc(i) \cap \Desc(j)$ but also both $i \not\in \an(j)$ and $j \not\in \an(i)$; in particular, $i$ and $j$ do not belong to the same tournament.
		Consider the paths $p(i,v)$ and $p(j,v)$. Along each path, continue walking upwards considering successive parents. 
		Since the graph is finite, this walk must end for both paths to a node without parents.
		By assumption, this must be the same unique source node of the ttt, say $u_0$.
		We will thus have found two different paths from $u_0$ to $v$, one passing via $i$ and the other one via $j$. However, as $i$ and $j$ do not belong to the same tournament, this is in contradiction to property~\ref{ttt:prop2} of a ttt.
		%
		%
	\end{proof}

	\subsection{Proof of Lemma~\ref{lem:addition_to_oned-1}}
	\begin{proof}
		
		
		1. 
		Suppose that there is a node $v_r$, for $r \in \{2, \ldots, n-1\}$, which is not the source node in the tournament shared with $v_{r+1}$, say $\tau$. Let $\bar{v}$ be a parent of $v_r$ in $\tau$. Note that $\bar{v}$ must be a parent of $v_{r+1}$ too, because of the out-degree ordering in a tournament. Because $v_{r-1}$ is a parent of $v_r$ too, both $v_{r-1}$ and $\bar{v}$ must belong to $\tau$, since otherwise $v_r$ would have parents from different tournaments, which is impossible according to Lemma~\ref{lem:oned-1}-2. Hence $v_{r-1}, v_r, \bar{v}, v_{r+1}$ all belong to the same tournament, i.e., to $\tau$. Necessarily $v_{r-1}$ is a parent of $v_{r+1}$, because otherwise there would be a directed cycle $\{v_{r-1},v_r, v_{r+1}, v_{r-1}\}$. But then $\{v_1, \ldots, v_{r-1}, v_{r+1}, \ldots, v_n\}$ is a shorter path between $v_1$ and $v_n$, in contradiction to the hypothesis. \smallskip

		
		2. Let the shortest trail between $u$ and $v$ be the one along the node sequence $\{v_1=u, \ldots, v_n=v\}$. It is sufficient to show that there cannot exist a node $v_{r}$ for $r \in \{2, \ldots, n-1\}$ such that $(v_{r-1}, v_r)\in E$ and $(v_{r+1},v_r)\in E$. Suppose indeed that the converse were true, i.e., there exists $r \in \{2, \ldots, n-1\}$ such that both $v_{r-1}$ and $v_{r+1}$ are parents of $v_r$. Then $v_{r-1}$ and $v_{r+1}$ must be adjacent because v-structures are excluded by statement~2 of Lemma~\ref{lem:oned-1}. But then $\{v_1, \ldots, v_{r-1}, v_{r+1}, \ldots, v_n\}$ is a shorter trail between $u$ and $v$, yielding a contradiction.  
	\end{proof}

	\section{Proofs and additional results for Section~\ref{sec:ctc}}

	\begin{proof}
		From \citet[Example~1]{segers2020one} we have the limit
		\[
		\sum_{j\in V} b_{uj}
		\delta_{\left\{\frac{b_{vj}}{b_{uj}}, v\in V_{\tau}\right\}}.
		\]
		Adapting this representation to a model where we have $b_{uj}=0$ for $j\notin \An(u)$ and $b_{ij}=c_{p(j,i)}b_{jj}$ for $j\in \An(i)$  we obtain
		\[
		\sum_{j\in \An(u)} b_{uj}
		\delta_{\left\{\frac{c_{p(j,v)}b_{jj}}{c_{p(j,u)}b_{jj}}, v\in V_{\tau}\right\}}.
		\]
		Recall that $c_{p(i,i)}=1$ and $c_{p(i,j)}=0$ if $i\notin \An(j)$.
		
		Next we show that $(M_{uv}, v\in V_{\tau})$ are mutually dependent. When $u$ is the source of $\tau$ then for every $j\in \An(u)$ the atom
		\[
		\left(\frac{c_{p(j,v)}}{c_{p(j,u)}}, v\in V_{\tau}\right)
		=\left(\frac{c_{p(j,u)}c_{uv}}{c_{p(j,u)}}, v\in V_{\tau}\right)
		=(1; c_{uv}, v\in V_{\tau}\setminus u)
		\]
		gets probability $\sum_{j\in \An(u)}b_{uj}=1$. Hence $(M_{uv}, v\in V_{\tau})$ are at the same time perfectly dependent and independent. 
		
		For $u$ which is not the source node the general idea is to take a collection of coordinates with joint probability zero, 
		and positive product of the marginal probabilities, thus showing that the joint probability does not equal the product of marginal probabilities for selected possible value of the random vector.
		
		Let for brevity $V_{\tau}=\{1,2,\ldots,m \}$: the nodes are labelled according to their order of out-degrees within $\tau$: the source node of $\tau$ has $m-1$ (largest) out-degree and is labelled by $1$, the node with out-degree $m-2$ is labelled as $2$, etc.   
		
		Consider $u$ being the node $2$. We have, thanks to the no-cycle property within a tournament $\An(2)=\An(1)\cup \{2\}$. 
		For all $j\in \An(1)$ we have 
		\begin{equation} \label{eqn:jvju}
			\left(\frac{c_{p(j,v)}}{c_{p(j,2)}}, v=1, \ldots, m \right)
			=\left(
			\frac{1}{c_{12}}; 1; \frac{c_{p(j,1)}c_{1v}}{c_{p(j,1)}c_{12}},
			v = 3, \ldots, m
			\right),
		\end{equation}
		which is an atom of $(M_{2v},v=1,\ldots, m)$ with mass $\sum_{j\in \An(1)}b_{2j}$. This means that for the marginal distribution of $M_{21}$ we have $\P(M_{21}=1/c_{12})\geq\sum_{j\in \An(1)}b_{2j}$.
		For $j=2$ we have an atom $(0, 1, c_{23}, \ldots, c_{2m})$ with mass $b_{22}$. This means that for the marginal probabilities of $(M_{23}, \ldots, M_{2m})$ we have $\P(M_{2v}=c_{2v})\geq b_{22}$ for all $v=3, \ldots, m$. Take a vector of coordinates $(1/c_{12}, 1,  c_{23}, \ldots, c_{2m})$. Note that this vector cannot be the same as the one in \eqref{eqn:jvju}. For any $v=3, \ldots, m$ we cannot have $c_{1v}/c_{12}=c_{2v}$ because of the criticality assumption, according to which $c_{1v} > c_{12}c_{2v}$ for any $v = 3,\ldots,m$. The joint probability of this vector of coordinates is 
		\[
		\P(M_{21}=1/c_{12}, M_{22}=1, 
		M_{23}=c_{23},\ldots , M_{2m}=c_{2m})
		=0.
		\]	 
		However the product of marginal probabilities is positive:
		\[
		\P(M_{21}=1/c_{12})
		\P( M_{22}=1)
		\prod_{v=3}^m\P( M_{2v}=c_{2v})
		\geq \sum_{j\in \An(1)}b_{2j}\times b_{22}^{m-1} 
		>0.
		\]
		
		Now let $u\geq 3$. Take the vector of coordinates in \eqref{eqn:LXLM} corresponding to $j=1$ which is equal to $(1/c_{1u},c_{12}/c_{1u}, \ldots, c_{1m}/c_{1u} )$ and has probability at least $b_{u1}$. Consider also the vector of coordinates for $j=u$ which is $(0, \ldots, 0, 1; c_{uv}, v=u+1, \ldots, m)$ with mass at least $b_{uu}$. Replace the first coordinate by $1/c_{1u}$. The vector obtained in this way has joint probability zero. For every $j\in \pa(u)$ we have $b_{vj}/b_{uj}=0$ when $v$ is not child of $j$ or equivalently, given the order in the node labelling, when $v<j$. So for fixed $u\geq3$, for $j=1$ the vector $(b_{vj}/b_{uj}, v=1, \ldots, m)$ has no zeros. For $j=2$ the vector $(b_{vj}/b_{uj}, v=1, \ldots, m)$ has one zero, namely $(0;b_{vj}/b_{uj}, v=2, \ldots, m)$, for $j=3$ the vector $(b_{vj}/b_{uj}, v=1, \ldots, m)$ has two zeros, namely  $(0,0;b_{vj}/b_{uj}, v=3, \ldots, m)$ and so on until $j=u$ with the corresponding vector $(b_{vj}/b_{uj}, v=1, \ldots, m)=(0,\ldots, 0;b_{vj}/b_{uj}, v=u, \ldots, m)$. By replacing the first coordinate by a non-zero value in this vector we get an impossible value for the random vector $(M_{uv}, v=1, \ldots, m)$ or a value with probability zero. 
		Considering the univariate marginal distributions of $(M_{uv}, v=1, \ldots, m)$ we obtain for the product of marginal probabilities a positive value:
		\[
		\P(M_{u1}=1/c_{1u})
		\left[\prod_{v=2}^{u-1}\P( M_{uv}=0)\right]
		\P(M_{uu}=1)
		\prod_{v=u+1}^m\P( M_{uv}=c_{uv})
		\geq b_{u1}\times b_{uu}^{m-1} 
		>0
		\]
		This shows that for any $u\in V_{\tau}$ the vector $(M_{u1}, \ldots, M_{um})$ has jointly dependent elements.
		
		
		\paragraph{} Next we show the distribution of a single element $M_{uv}, v\in V_{\tau}\setminus u$.
		
		1. Consider first when $u$ is the source node in $\tau$. Since $(u, v) \in E$, we have $\An(u) \subset \An(v)$ and thus $\An(v) \cap \An(u) = \An(u)$. We have $b_{vj}>0, j\in \An(u)$, hence zero is not a possible value of $M_{uv}$. For $j\in \An(u)$
		\[
		\frac{b_{vj}}{b_{uj}}=\frac{c_{p(j,u)}c_{uv}b_{jj}}{c_{p(j,u)}b_{jj}}=c_{uv},
		\] 
		and since $\sum_{j\in \An(u)}b_{uj}=1$ we obtain the desired result under~1.(a).
		
		When $u$ is not the source node in $\tau$ not all shortest paths to $v$ pass through $u$ hence for $j\in\An(u)$ we have 
		\[
		\frac{b_{vj}}{b_{uj}}=\frac{c_{p(j,v)}b_{jj}}{c_{p(j,u)}b_{jj}}
		=\frac{c_{p(j,v)}}{c_{p(j,u)}} >0
		\] 
		with mass $b_{uj}$. Hence the result in~1.(b). Note that zero is not possible value as we still have $\An(u)\subset \An(v)$. Also $c_{p(u,u)}=1$ by convention.
		
		2. Let us have now $(v,u)\in E$. In this case $\An(u)\setminus\An(v)$ is not empty because it contains at least the node $u$, so zero is a possible value of $M_{uv}$. We need to distinguish only the zero atoms from the non-zero ones. 
		When $v$ is a source node in $\tau$, we have, for $j\in \An(v)$ 
		\[
		\frac{b_{vj}}{b_{uj}}
		=\frac{c_{p(j,v)}b_{jj}}{c_{p(j,u)}b_{jj}}
		=\frac{c_{p(j,v)}}{c_{p(j,v)}c_{vu}}
		=\frac{1}{c_{vu}} >0,
		\] 
		which is an atom with probability 
		\[
		\sum_{j\in \An(v)}b_{uj}
		=\sum_{j\in \An(v)}c_{p(j,v)}c_{vu}b_{jj}
		=c_{vu}\sum_{j \in \An(v)} c_{p(j,v)} b_{jj}
		=c_{vu}\sum_{j \in \An(v)} b_{vj} = c_{vu}.
		\] 
		The probability of the zero atom is $\sum_{j\in \An(u)\setminus \An(v)}b_{uj}=\sum_{j\in \An(u)}b_{uj}-\sum_{j\in \An(v)}b_{uj}=1-c_{vu}$. This shows~2.(a).
		
		When $v$ is not a source node of $\tau$ we have for $j\in \An(v)$
		\[
		\frac{b_{vj}}{b_{uj}}
		=\frac{c_{p(j,v)}b_{jj}}{c_{p(j,u)}b_{jj}}
		=\frac{c_{p(j,v)}}{c_{p(j,u)}}>0,
		\] 
		an atom with mass $b_{uj}$ and zero atom with probability $\sum_{j\in \An(u)\setminus \An(v)}b_{uj}$. This shows~2.(b).
	\end{proof}
	
	\paragraph{Remark.} From the results in Lemma~\ref{lem:Muv-rev1} we see that a multiplicative increment does not have a degenerate distribution at zero, so that a product of several such multiplicative increments cannot be degenerate at zero either. This is an important observation that we will use in further proofs.

	\begin{lemma} \label{lem:Auv_rev1}
		Let $(X_v, v\in V)$ follow a max-linear model as in Assumption~\ref{ass:max-mod}.
		Let $\T$ have a unique source. For any $u\in V$ we have 
		\begin{align} \label{eqn:LXLA}
			\mathcal{L}\left(\frac{X_v}{X_u}, v\in V \mathrel{\Big|} X_u>t\right)
			\inlaw
			\mathcal{L}(A_{uv}, v\in V)
			=\sum_{j\in \An(u)}b_{uj}
			\delta_{\left\{\frac{c_{p(j,v)}}{c_{p(j,u)}}, v\in V\right\}}.
		\end{align}
		The distribution of $A_{uv}$ depends on the three types of possible trails according to Lemma~\ref{lem:addition_to_oned-1}-2. In what follows we assume $(u,v)\notin E$. For the case $(u,v)\in E$ see Lemma~\ref{lem:Muv-rev1}.
		\begin{compactenum}
			\item Distribution of $A_{uv}$ on a path $\{u=v_1, r=v_2, \ldots, v=v_n\}$ with $u,r\in \tau$, one of the tournaments of $\T$.
			\begin{compactenum}
				\item If $u$ is a source node in $\tau$ then $\mathcal{L}(A_{uv})=\delta_{\{c_{p(u,v)}\}}$.
				
				\item If $u$ is not a source node in $\tau$ we have 
				\[
				\mathcal{L}(A_{uv})
				=\sum_{j\in \An(u)} b_{uj}\delta_{\left\{\frac{c_{p(j,r)}}{c_{p(j,u)}}c_{p(r,v)}\right\}}.
				\]			
			\end{compactenum}
			
			\item  Distribution of $A_{uv}$ on a path $\{v=v_1, r=v_2, \ldots, u=v_n\}$ with $v,r\in \tau$.
			\begin{compactenum}
				\item If $v$ is a source node in $\tau$ then
				\[
				\mathcal{L}(A_{uv})
				=c_{p(v,u)}\delta_{\left\{\frac{1}{c_{p(v,u)}}\right\}}
				+(1-c_{p(v,u)})\delta_{\{0\}}.
				\]
				\item If $v$ is not a source node in $\tau$ then
				\[
				\mathcal{L}(A_{uv})
				=\sum_{j\in \An(v)} c_{p(r,u)}b_{rj}\delta_{\left\{
					\frac{c_{p(j,v)}}{c_{p(j,r)}c_{p(r,u)}}\right\}}
				+\sum_{j\in \An(u)\setminus \An(v)}b_{uj}\delta_{\{0\}}.
				\]	
			\end{compactenum}
			\item The distribution of $A_{uv}$ on a trail composed of two paths $p(r,u)$ and $p(r,v)$. Let the trail be on nodes $\{u, \ldots, m,r, n, \ldots, v\}$. Let also $\tau_m, \tau_n$ be two tournaments with $r,m\in \tau_m$ and $r,n\in \tau_n$.
			\begin{compactenum}
				\item If $r$ is source in both $\tau_m$ and $\tau_n$, then
				\[
				\mathcal{L}(A_{uv})=c_{p(r,u)}
				\delta_{\left\{\frac{c_{p(r,v)}}{c_{p(r,u)}}\right\}}
				+(1-c_{p(r,u)})\delta_{\{0\}}.
				\]
				
				\item If $r$ is source in $\tau_m$, but not in $\tau_n$, then
				\[
				\mathcal{L}(A_{uv})=\sum_{j\in\An(r)}c_{p(r,u)}b_{rj}
				\delta_{\left\{\frac{c_{p(j,n)}c_{p(n,v)}}{c_{p(j,r)}c_{p(r,u)}}\right\}}
				+\sum_{j\in \An(u)\setminus \An(r)}b_{uj}\delta_{\{0\}}.
				\]
				
				\item If $r$ is source in $\tau_n$, but not in $\tau_m$, then
				\[
				\mathcal{L}(A_{uv})=\sum_{j\in\An(r)}c_{p(m,u)}b_{mj}
				\delta_{\left\{\frac{c_{p(j,r)}c_{p(r,v)}}{c_{p(j,m)}c_{p(m,u)}}\right\}}
				+\sum_{j\in \An(u)\setminus \An(r)}b_{uj}\delta_{\{0\}}.
				\]
			\end{compactenum}
		\end{compactenum}
	\end{lemma}
	
	\begin{proof}
		We have already seen that from \citet[Example~1]{segers2020one} we have the limit
		\[
		\mathcal{L}\left(\frac{X_v}{X_u}, v\in V\mid X_u>t\right)
		\inlaw
		\sum_{j\in V} b_{uj}
		\delta_{\left\{\frac{b_{vj}}{b_{uj}}, v\in V\right\}}.
		\]
		Adapting this representation to a model where we have $b_{uj}=0$ for $j\notin \An(u)$ and $b_{ij}=c_{p(j,i)}b_{jj}$ for $j\in \An(i)$  we obtain
		\[
		\sum_{j\in \An(u)} b_{uj}
		\delta_{\left\{\frac{c_{p(j,v)}}{c_{p(j,u)}}, v\in V\right\}}.
		\]
		Recall that $c_{p(i,i)}=1$ and $c_{p(i,j)}=0$ if $i\notin \An(j)$. For a single $v\in V\setminus u$ we have the marginal distribution
		\begin{equation} \label{eqn:onlyAuv}
			\mathcal{L}(A_{uv})=\sum_{j\in\An(u)}b_{uj}
			\delta_{\left\{\frac{b_{vj}}{b_{uj}}\right\}}.
		\end{equation}
		
		
		The distribution of $A_{uv}$ depends deterministically on properties of the ttt. When $\T$ has a unique source, according to Lemma~\ref{lem:addition_to_oned-1}-2 there are three possible shortest trails between two nodes. In addition we have also the property under Lemma~\ref{lem:addition_to_oned-1}-1. We look at the different distributions of $A_{uv}$ that arise due to these two properties of the ttt.  
		
		First we deal with~1.(a).
		Since $\An(u)\subset \An(v)$ all atoms in \eqref{eqn:onlyAuv} are positive and zero is not a possible value of $A_{uv}$. All paths from $\An(u)$ to $v$ pass through $u$ because $u$ is source in $\tau$ and because by property~\ref{ttt:prop2} of a ttt no cycle involving several tournaments is allowed. The case is illustrated by the graph below.
		
		
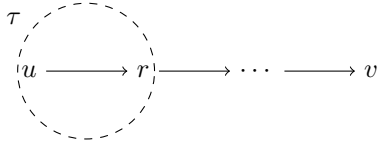
\begin{center} 
			\begin{tikzpicture}
				\node[] (u) at (0,0)  {$u$};
				\node[] (m) at (1.5,0)  {$r$};
				\node[] (n) at (3,0)  {$\cdots$};
				\node[] (v) at (4.5,0)  {$v$};
				\path[->] (u) edge (m)  ;
				\path[->] (m) edge (n)  ;
				\path[->] (n) edge (v)  ;
				\node[] (tau) at (-0.2,0.7)  {$\tau$};
				\begin{scope}[dashed]
					\draw[color=black] (0.75,0) circle (0.9cm);
				\end{scope}
			\end{tikzpicture}
		\end{center}	
		Hence for all $j\in \An(u)$ we have 
		\[
		\frac{b_{vj}}{b_{uj}}
		=\frac{c_{p(j,v)}b_{jj}}{c_{p(j,u)}b_{jj}}
		=\frac{c_{p(j,u)}c_{p(u,v)}}{c_{p(j,u)}}
		=c_{p(u,v)}>0,
		\]
		with mass $\sum_{j\in \An(u)}b_{uj}=1$. 
		
		Next we show~1.(b). Because $\An(u)\subset \An(v)$, zero is not possible value of $A_{uv}$. Not all shortest paths from $\An(u)$ to $v$ pass through $u$ because $u$ is not source in $\tau$.  However all paths from $\An(u)$ to $v$ pass through $r$, as shown in the picture. Paths from $\An(u)$ to $v$ other than these passing through $u$ or $r$ are impossible because of the property~\ref{ttt:prop2} of a ttt.  
		
		
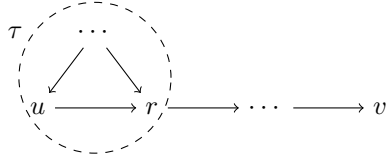
\begin{center} 
			\begin{tikzpicture}
				\node[] (u) at (0,0)  {$u$};
				\node[] (m) at (1.5,0)  {$r$};
				\node[] (n) at (3,0)  {$\cdots$};
				\node[] (v) at (4.5,0)  {$v$};
				\node[] (p) at (0.75,1)  {$\cdots$};
				\path[->] (u) edge (m)  ;
				\path[->] (m) edge (n)  ;
				\path[->] (n) edge (v)  ;
				\path[->] (p) edge (m)  ;
				\path[->] (p) edge (u)  ;
				\node[] (tau) at (-0.3,1)  {$\tau$};
				\begin{scope}[dashed]
					\draw[color=black] (0.75,0.4) circle (1cm);
				\end{scope}
			\end{tikzpicture}
		\end{center}
		We have for $j\in \An(u)$
		\[
		\frac{b_{vj}}{b_{uj}}
		=\frac{c_{p(j,v)}b_{jj}}{c_{p(j,u)}b_{jj}}
		=\frac{c_{p(j,r)}}{c_{p(j,u)}}c_{p(r,v)}>0,
		\]
		with mass $b_{uj}$, hence the expression in~1.(b).

		Next we show~2.(a). 
		When the directed path is from $v$ to $u$ the set $\An(u)\setminus \An(v)$ contains at least $u$ hence we have $b_{vj}=0$ for all $j\in \An(u)\setminus \An(v)$. This means that zero is a possible value of $A_{uv}$. All shortest paths from $j\in \An(v)$ to $u$ pass through $v$ as $v$ is source in $\tau$. Otherwise, there would be cycle encompassing multiple tournaments, which is not allowed under property~\ref{ttt:prop2} of a ttt.
		
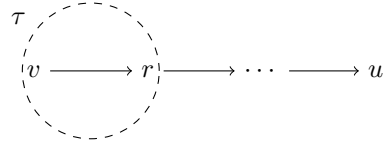
\begin{center} 
			\begin{tikzpicture}
				\node[] (u) at (0,0)  {$v$};
				\node[] (m) at (1.5,0)  {$r$};
				\node[] (n) at (3,0)  {$\cdots$};
				\node[] (v) at (4.5,0)  {$u$};
				\path[->] (u) edge (m)  ;
				\path[->] (m) edge (n)  ;
				\path[->] (n) edge (v)  ;
				\node[] (tau) at (-0.2,0.7)  {$\tau$};
				\begin{scope}[dashed]
					\draw[color=black] (0.75,0) circle (0.9cm);
				\end{scope}
			\end{tikzpicture}
		\end{center}		
		For $j\in \An(v)$ the non-zero atom is given by
		\[
		\frac{b_{vj}}{b_{uj}}
		=\frac{c_{p(j,v)}b_{jj}}{c_{p(j,u)}b_{jj}}
		=\frac{c_{p(j,v)}}{c_{p(j,v)}c_{p(v,u)}}=\frac{1}{c_{p(v,u)}}>0, 
		\qquad j\in\An(v),
		\]
		with mass 
		\[
		\sum_{j\in \An(v)}b_{uj}
		=\sum_{j\in \An(v)}c_{p(j,v)}c_{p(v,u)}b_{jj}
		=c_{p(v,u)}\sum_{j\in \An(v)}c_{p(j,v)}b_{jj}
		=c_{p(v,u)}\sum_{j\in \An(v)}b_{vj}
		=c_{p(v,u)}.
		\] 
		For the zero atom we have probability $\sum_{j\in\An(u)\setminus \An(v)}b_{uj}=\sum_{j\in \An(u)}b_{uj}-\sum_{j\in \An(v)}b_{uj}=1-c_{p(v,u)}$. This shows~2.(a). 
		
		To show~2.(b) we note that when $v$ is not a source node of $\tau$ not all shortest paths from $j\in \An(v)$ to $u$ pass through $v$. However all paths from $j\in \An(v)$ to $u$ pass through $r$, as it can be seen from the figure here.
		
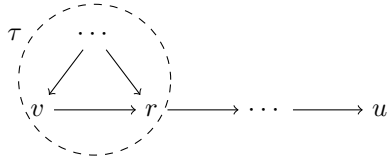
\begin{center} 
			\begin{tikzpicture}
				\node[] (u) at (0,0)  {$v$};
				\node[] (m) at (1.5,0)  {$r$};
				\node[] (n) at (3,0)  {$\cdots$};
				\node[] (v) at (4.5,0)  {$u$};
				\node[] (p) at (0.75,1)  {$\cdots$};
				\path[->] (u) edge (m)  ;
				\path[->] (m) edge (n)  ;
				\path[->] (n) edge (v)  ;
				\path[->] (p) edge (m)  ;
				\path[->] (p) edge (u)  ;
				\node[] (tau) at (-0.3,1)  {$\tau$};
				\begin{scope}[dashed]
					\draw[color=black] (0.75,0.4) circle (1cm);
				\end{scope}
			\end{tikzpicture}
		\end{center}	
		Hence for $j\in \An(v)$ we have
		\[
		\frac{b_{vj}}{b_{uj}}
		=\frac{c_{p(j,v)}b_{jj}}{c_{p(j,u)}b_{jj}}
		=\frac{c_{p(j,v)}}{c_{p(j,r)}c_{p(r,u)}}>0,
		\]
		which is an atom with mass $b_{uj}=c_{p(j,r)}c_{p(r,u)}b_{jj}=c_{p(r,u)}b_{rj}$. The zero atom comes from the fact that $b_{vj}=0$ for all $j\in \An(u)\setminus \An(v)$, and it has probability $\sum_{j\in\An(u)\setminus \An(v)}b_{uj}$. This shows the distribution under~2.(b).

		By Lemma~\ref{lem:addition_to_oned-1}-1 the node $r$, as part of the paths $p(r,u)$ is allowed not be a source node in $\tau_m$. Similarly considering the path $p(r,v)$. However when we combine $p(r,u)$ and $p(r,v)$ in one trail $t(u,v)$ the node $r$ should be a source in at least one of $\tau_m$ and $\tau_n$. If $r$ is not source of both $\tau_m$ and $\tau_n$ then there would be indeed a v-structure. However, Lemma~\ref{lem:oned-1}-2 excludes v-structures when $\T$ has a unique source, hence node $r$ should be source in at least one tournament, $\tau_m$ and/or $\tau_n$.
		
		To show~3.(a) we note that all paths from $j\in \An(r)$ to $u$ and to $v$ pass through $r$, as $r$ is source in both $\tau_n$ and $\tau_m$. The case is depicted in the following picture.
		
		
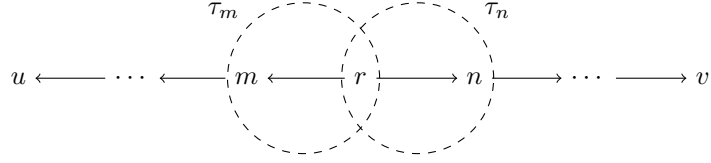
\begin{center}
			\begin{tikzpicture}
				\node[] (r) at (0,0)  {$r$};
				\node[] (n) at (1.5,0)  {$n$};
				\node[] (nn) at (3,0)  {$\cdots$};
				\node[] (v) at (4.5,0)  {$v$};
				\node[] (m) at (-1.5,0)  {$m$};
				\node[] (nnn) at (-3,0)  {$\cdots$};
				\node[] (u) at (-4.5,0)  {$u$};
				\path[->] (r) edge (n)  ;
				\path[->] (n) edge (nn)  ;
				\path[->] (nn) edge (v)  ;
				\path[->] (r) edge (m)  ;
				\path[->] (m) edge (nnn)  ;
				\path[->] (nnn) edge (u)  ;
				\node[] (tau) at (1.8,0.9)  {$\tau_n$};
				\node[] (taum) at (-1.8,0.9)  {$\tau_m$};
				\begin{scope}[dashed]
					\draw[color=black] (-0.75,0) circle (1cm);
					\draw[color=black] (0.75,0) circle (1cm);
				\end{scope}
			\end{tikzpicture}
		\end{center}
		Also we have $b_{vj}=0$ for all $j\in \An(u)\setminus \An(r)$. For $j\in \An(r)$ we have
		\[
		\frac{b_{vj}}{b_{uj}}
		=\frac{c_{p(j,r)}c_{p(r,v)}b_{jj}}{c_{p(j,r)}c_{p(r,u)}b_{jj}}
		=\frac{c_{p(r,v)}}{c_{p(r,u)}}>0,
		\]
		with probability 
		\[
		\sum_{j\in \An(r)}b_{uj}
		=\sum_{j\in \An(r)}c_{p(j,r)}c_{p(r,u)}b_{jj}
		=c_{p(r,u)}\sum_{j\in \An(r)}c_{p(j,r)}b_{jj}
		=c_{p(r,u)}\sum_{j\in \An(r)}b_{rj}
		=c_{p(r,u)}.
		\]
		The probability of the zero atom is $\sum_{j\in \An(u)\setminus \An(r)}b_{uj}=\sum_{j\in \An(u)}b_{uj}-\sum_{j\in\An(r)}b_{uj}=1-c_{p(r,u)}$. 
		
		Next we show~3.(b). Because $r$ is not a source in $\tau_n$ not all paths from $\An(r)$ to $v$ pass through $r$, but they do all pass through $n$. Also all paths from $\An(r)$ to $u$ pass through $r$ because $r$ is source in $\tau_m$. 
		
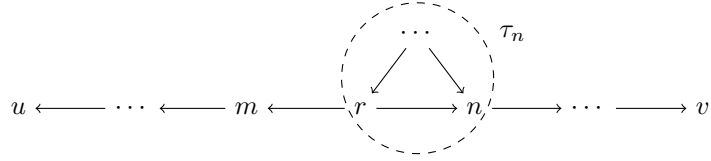
		\begin{center}
			\begin{tikzpicture}
				\node[] (r) at (0,0)  {$r$};
				\node[] (n) at (1.5,0)  {$n$};
				\node[] (nn) at (3,0)  {$\cdots$};
				\node[] (v) at (4.5,0)  {$v$};
				\node[] (p) at (0.75,1)  {$\cdots$};
				\node[] (m) at (-1.5,0)  {$m$};
				\node[] (nnn) at (-3,0)  {$\cdots$};
				\node[] (u) at (-4.5,0)  {$u$};
				\path[->] (r) edge (n)  ;
				\path[->] (n) edge (nn)  ;
				\path[->] (nn) edge (v)  ;
				\path[->] (p) edge (r)  ;
				\path[->] (p) edge (n)  ;
				\path[->] (r) edge (m)  ;
				\path[->] (m) edge (nnn)  ;
				\path[->] (nnn) edge (u)  ;
				\node[] (tau) at (2,1)  {$\tau_n$};
				\begin{scope}[dashed]
					\draw[color=black] (0.75,0.4) circle (1cm);
				\end{scope}
			\end{tikzpicture}
		\end{center}
		Hence for $j\in \An(r)$
		\[
		\frac{b_{vj}}{b_{uj}}
		=\frac{c_{p(j,n)}c_{p(n,v)}b_{jj}}{c_{p(j,r)}c_{p(r,u)}b_{jj}}
		=\frac{c_{p(j,n)}c_{p(n,v)}}{c_{p(j,r)}c_{p(r,u)}}>0,
		\]
		which is an atom with mass $b_{uj}=c_{p(j,r)}c_{p(r,u)}b_{jj}=b_{rj}c_{p(r,u)}$. The zero atom has probability equal to $\sum_{j\in \An(u)\setminus \An(r)}b_{uj}$. 
		
		Next we show~3.(c). When $r$ is source in $\tau_n$ it means that all paths from $\An(r)$ to $v$ pass through $r$. Because $r$ is not source in $\tau_m$ not all paths from $\An(r)$ to $u$ pass through $r$, but they do all pass through $m$. 
		
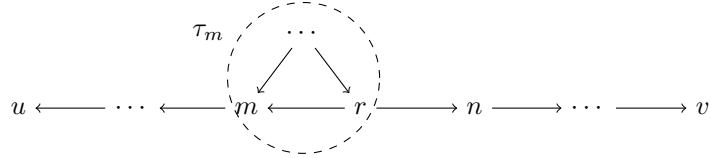
		\begin{center}
			\begin{tikzpicture}
				\node[] (r) at (0,0)  {$r$};
				\node[] (n) at (1.5,0)  {$n$};
				\node[] (nn) at (3,0)  {$\cdots$};
				\node[] (v) at (4.5,0)  {$v$};
				\node[] (m) at (-1.5,0)  {$m$};
				\node[] (nnn) at (-3,0)  {$\cdots$};
				\node[] (u) at (-4.5,0)  {$u$};
				\node[] (pp) at (-0.75,1)  {$\cdots$};
				\path[->] (r) edge (n)  ;
				\path[->] (n) edge (nn)  ;
				\path[->] (nn) edge (v)  ;
				\path[->] (r) edge (m)  ;
				\path[->] (m) edge (nnn)  ;
				\path[->] (nnn) edge (u)  ;
				\path[->] (pp) edge (r)  ;
				\path[->] (pp) edge (m)  ;
				\node[] (tau) at (-2,1)  {$\tau_m$};
				\begin{scope}[dashed]
					\draw[color=black] (-0.75,0.4) circle (1cm);
				\end{scope}
			\end{tikzpicture}
		\end{center}
		For $j\in \An(r)$ we have
		\[
		\frac{b_{vj}}{b_{uj}}
		=\frac{c_{p(j,r)}c_{p(r,v)}b_{jj}}{c_{p(j,m)}c_{p(m,u)}b_{jj}}
		=\frac{c_{p(j,r)}c_{p(r,v)}}{c_{p(j,m)}c_{p(m,u)}}>0,
		\]
		which is an atom with mass $b_{uj}=c_{p(j,m)}c_{p(m,u)}b_{jj}=b_{mj}c_{p(m,u)}$. The zero atom comes from $b_{uj}=0$ for all $j\in\An(u)\setminus \An(r)$. It gets probability $\sum_{j\in \An(u)\setminus \An(r)}b_{uj}$. 
	\end{proof}
	
	\subsubsection*{Proof of Proposition~\ref{prop:factorize}}
	
	\begin{proof}
		First we prove that (i) implies (ii). Assume $\T$ has a unique source. We have to prove that for any $u\in V$ an element from the limiting vector in \eqref{eqn:LXvXu} is given by \eqref{eqn:Auvprod}. 
		
		In Lemma~\ref{lem:Auv_rev1} we have seen a number of cases for the distribution of $A_{uv}$ depending on deterministic properties of the trail between $u$ and $v$. Below we consider each of these cases again.
		\smallskip
		
		\underline{\emph{Case 1.}} Let the unique shortest trail between $u$ and $v$ be a path on node sequence $\{u=v_1, r=v_2, \ldots, v_n=v \}$. Let $\tau$ be the tournament containing $u,r$.
		\smallskip
		
		\emph{Case~1.(a).} -- Let $u$ be source in $\tau$. From Lemma~\ref{lem:Auv_rev1}-1.(a) we have $P(A_{uv}=c_{p(u,v)})=1$. Consider the variables $(M_e, e\in p(u,v))$ which are by construction independent between each other because they belong to different tournaments. 
		Note that in this case all nodes $v_1,\ldots, v_{n-1}$ are source nodes in the tournament containing that node and the next one in the sequence. This follows from Lemma~\ref{lem:addition_to_oned-1}-1. Then according to Lemma~\ref{lem:Muv-rev1}~1.(a) for every $M_e, e\in p(u,v)$ we have $\P(M_e=c_e)=1$ and hence
		\[
		\P\left(\prod_{e\in p(u,v)}M_e=c_{p(u,v)}\right) =
		\prod_{e\in p(u,v)}\P(M_{e}=c_e)=1,
		\]
		which shows $A_{uv}=\prod_{e\in p(u,v)}M_e$. 
		\smallskip

		\emph{Case~1.(b).} -- If $u$ is not the source in $\tau$, the distribution of $M_{ur}$ is as in Lemma~\ref{lem:Muv-rev1}-1.(b). As in the case~1.(a)  all nodes $r=v_2,v_3,\ldots, v_{n-1}$ are source nodes in the tournament containing that node and the next one in the sequence. The variables $M_e, e\in p(r,v)$ are degenerate at $c_e$.
		As the case~1.(a) above the variables $(M_e, e\in p(u,v))$ are by construction independent between each other because they are indexed by edges which belong to different tournaments. Then we have  
		\begin{equation} \label{eqn:case1b}
			\begin{split} 
				\mathcal{L}\left(\prod_{e\in p(u,v)}M_e\right)
				=\mathcal{L}\left(M_{ur}\prod_{e\in p(r,v)}M_e\right)
				&=\left(\sum_{j\in \An(u)} b_{uj}\delta_{\left\{\frac{c_{p(j,r)}}{c_{p(j,u)}}\right\}}\right)
				\otimes 
				\delta_{\{c_{p(r,v)}\}}
				\\&=\sum_{j\in \An(u)} b_{uj}\delta_{\left\{\frac{c_{p(j,r)}}{c_{p(j,u)}}c_{p(r,v)}\right\}}.
			\end{split}
		\end{equation}
		The sign $\otimes$ denotes multiplication between two discrete probability measures, say $\mu$ and $\nu$ of two independent variables, say $\xi_1, \xi_2$ respectively. For two possible values $a_1,a_2$ of $\xi_1, \xi_2$ respectively we have $\mu(\{a_1\})\nu(\{a_1\})$ as a measure of the event $\{\xi_1\xi_2=a_1a_2\}=\{\xi_1=a_1, \xi_2=a_2\}$. 
		The last one expression in~\eqref{eqn:case1b} is the distribution of $A_{uv}$ in Lemma~\ref{lem:Auv_rev1}-1.(b).
		\smallskip
		
		\underline{\emph{Case 2.}} Let the unique shortest trail between $u$ and $v$ be a path from $v$ to $u$ on the node sequence $\{v=v_1, r=v_2, \ldots, v_n=u \}$.
		Let $\tau$ be the tournament containing $v,r$.
		\smallskip
		
		\emph{Case~2.(a).} -- Let $v$ be source in $\tau$.  Consider the random variables $M_{v_{i+1}v_{i}}, i=1,\ldots, n-1$ whose distributions are as in Lemma~\ref{lem:Muv-rev1}-2.(a). Since this is the unique shortest trail from $v$ to $u$, all edges on it belong to different tournaments and the vector $(M_{v_{i+1}v_{i}}, i=1,\ldots, n-1)$ contains independent variables by definition. Then 
		\begin{equation} \label{eqn:case2a1}
			\P\Big(\prod_{i=1}^{n-1}M_{v_{i+1}v_{i}}=\frac{1}{c_{p(v,u)}}\Big)=
			\prod_{i=1}^{n-1}\P\Big(M_{v_{i+1}v_{i}}=\frac{1}{c_{v_{i}v_{i+1}}}\Big)
			=\prod_{i=1}^{n-1}c_{v_{i}v_{i+1}}=c_{p(v,u)}.
		\end{equation}
		For the zero atom we have
		\begin{equation} \label{eqn:case:2a2}
			\P\left(\prod_{i=1}^{n-1}M_{v_{i+1}v_{i}}=0\right)
			=1-\prod_{i=1}^{n-1}\P(M_{v_{i+1}v_{i}}>0)
			=1-\prod_{i=1}^{n-1}\P\left(M_{v_{i+1}v_{i}}
			=\frac{1}{c_{v_{i}v_{i+1}}}\right)
			=1-c_{p(v,u)}.
		\end{equation}
		The expressions in \eqref{eqn:case2a1} and \eqref{eqn:case:2a2} represent indeed the distribution of $A_{uv}$ in Lemma~\ref{lem:Auv_rev1}-2.(a).
		\smallskip
		
		\emph{Case~2.(b).} -- If $v$ is not the source in $\tau$, consider a random variable $M_{rv}$ with distribution as in Lemma~\ref{lem:Muv-rev1}-2.(b) and a random variable $A_{ur}$ constructed as in~2.(a) here above, i.e., as the product $\prod_{i=2}^{n-1}M_{v_{i+1}v_i}$. By construction $M_{rv}$ is independent from $A_{ur}$ with the same argument as above. We have
		\begin{align*}
			\mathcal{L}(A_{ur}M_{rv})
			&=\left(c_{p(r,u)}\delta_{\left\{\frac{1}{c_{p(r,u)}}\right\}}
			+(1-c_{p(r,u)})\delta_{\{0\}}\right)
			\\&\otimes
			\left(\sum_{j\in \An(v)} b_{rj}
			\delta_{\left\{\frac{c_{p(j,v)}}{c_{p(j,r)}}\right\}}
			+\sum_{j\in \An(r)\setminus \An(v)}b_{rj}\delta_{\{0\}}\right)
		\end{align*}
		which gives non-zero atoms $c_{p(j,v)}/(c_{p(j,r)}c_{p(r,u)}), j\in \An(v)$ with masses $b_{rj}c_{p(r,u)}, j\in \An(v)$. To show the probability of the zero atom, consider 
		\begin{align*}
			&\P(A_{ur}M_{rv}=0)
			=1-\P(A_{ur}>0) \cdot \P(M_{rv}>0)
			=1-c_{p(r,u)}\sum_{j\in \An(v)}b_{rj}
			\\&=\sum_{j\in \An(u)}b_{uj}-\sum_{j\in \An(v)}c_{p(j,r)}b_{jj}c_{p(r,u)}
			=\sum_{j\in \An(u)}b_{uj}-\sum_{j\in \An(v)}b_{uj}
			=\sum_{j\in \An(u)\setminus \An(v)}b_{uj},
		\end{align*}
		which is what we need to confirm $A_{uv}=A_{ur}M_{rv}$ where $A_{uv}$ is as in Lemma~\ref{lem:Auv_rev1}-2.(b).
		\smallskip

		\underline{\emph{Case 3.}} In the three cases that follow let the unique shortest trail from $u$ to $v$ be given by two paths $p(r,u)$ and $p(r,v)$. Let the trail be on nodes $\{u, \ldots, m,r, n, \ldots, v\}$. Let also $\tau_m, \tau_n$ be two tournaments with $r,m\in \tau_m$ and $r,n\in \tau_n$.
		\smallskip
		
		\emph{Case~3.(a).} -- Let $r$ be source in both $\tau_m$ and $\tau_n$.
		Consider random variables $A_{rv}$ as in Lemma~\ref{lem:Auv_rev1}-1.(a) and $A_{ur}$ as in Lemma~\ref{lem:Auv_rev1}-2.(a). Above we have shown in cases~1.(a) and~2.(a) that $A_{rv}$ and $A_{ur}$ are factorizable in independent multiplicative increments. By construction $A_{rv}$ and $A_{ur}$ are independent from each other, because the multiplicative increments are independent. We have
		\begin{align*}
			\P\Big(A_{ur}A_{rv}=\frac{c_{p(r,v)}}{c_{p(r,u)}}\Big)
			=\P\Big(A_{ur}=\frac{1}{c_{p(r,u)}}\Big)\P(A_{rv}=c_{p(r,v)})
			=c_{p(r,u)}.
		\end{align*}
		For the probability of the zero atom we have
		\[
		\P(A_{ur}A_{rv}=0)=P(A_{ur}=0)=(1-c_{p(r,u)}). 
		\]
		The two displays above represent indeed the distribution of $A_{uv}$ in Lemma~\ref{lem:Auv_rev1}-3.(a).
		\smallskip

		\emph{Case~3.(b).} -- Let $r$ be source in $\tau_m$, but not source in $\tau_n$. 
		Consider three random variables $A_{ur}, M_{rn}, A_{nv}$ with distributions as in Lemma~\ref{lem:Auv_rev1}-2.(a), Lemma~\ref{lem:Muv-rev1}-1.(b) and Lemma~\ref{lem:Auv_rev1}-1.(a) respectively. For $A_{ur}$ and $A_{nv}$ we have shown in cases~2.(a) and~1.(a) in this proof that they are factorizable in independent multiplicative increments. By construction $M_{rn}$ is independent from the increments in $A_{ur}$ and $A_{nv}$. Then 
		\begin{align*}
			\mathcal{L}(A_{ur}M_{rn}A_{nv})
			&=\left(c_{p(r,u)}\delta_{\left\{\frac{1}{c_{p(r,u)}}\right\}}
			+(1-c_{p(r,u)})\delta_{\{0\}}\right)
			\otimes 
			\left(\sum_{j\in \An(r)}b_{rj}
			\delta_{\left\{\frac{c_{p(j,n)}}{c_{p(j,r)}}\right\}}\right)
			\otimes 
			\delta_{\{c_{p(n,v)}\}}
			\\&=\sum_{j\in \An(r)}b_{rj}c_{p(r,u)}
			\delta_{\left\{\frac{c_{p(j,n)}c_{p(n,v)}}{c_{p(j,r)}c_{p(r,u)}}\right\}}
			+(1-c_{p(r,u)})\delta_{\{0\}}.
		\end{align*}
		Note that 
		\begin{align*}
			\sum_{j\in \An(u)\setminus \An(r)}b_{uj}
			=\sum_{j\in\An(u)}b_{uj}-\sum_{j\in \An(r)}b_{uj}
			=1-\sum_{j\in \An(r)} c_{p(j,r)}c_{p(r,u)}b_{jj}
			&=1-c_{p(r,u)}\sum_{j\in \An(r)} b_{rj}
			\\&=1-c_{p(r,u)}.
		\end{align*}
		This shows that the distribution of $A_{ur}M_{rn}A_{nv}$ is the one of $A_{uv}$ in Lemma~\ref{lem:Auv_rev1}-3.(b). 
		\smallskip

		\emph{Case~3.(c).} -- Let $r$ be source in $\tau_n$, but not in $\tau_m$. Consider variables $A_{um}, M_{mr}, A_{rv}$ with distributions as in Lemma~\ref{lem:Auv_rev1}-2.(a), Lemma~\ref{lem:Muv-rev1}-2.(b) and Lemma~\ref{lem:Auv_rev1}-1.(a) respectively. The variables $A_{um}$ and $A_{rv}$ have been shown to factorize in independent increments in cases~2.(a) and~1.(a) of this proof respectively, hence they are independent from each other too. By construction $M_{mr}$ is independent from $A_{um}$ and $A_{rv}$. Then we have
		\begin{align*}
			\mathcal{L}(A_{um}M_{mr}A_{rv})
			&=\left(c_{p(m,u)}\delta_{\left\{\frac{1}{c_{p(m,u)}}\right\}}
			+(1-c_{p(m,u)})\delta_{\{0\}}\right)
			\\&	\otimes 
			\left(\sum_{j\in \An(r)}b_{mj}
			\delta_{\left\{\frac{c_{p(j,r)}}{c_{p(j,m)}}\right\}}
			+\sum_{j\in \An(m)\setminus \An(r)}b_{mj}\delta_{\{0\}}\right)
			\otimes 
			\delta_{\{c_{p(r,v)}\}}.
		\end{align*}
		The non-zero atoms are $c_{p(j,r)}c_{p(r,v)}/(c_{p(j,m)}c_{p(m,u)})$ for $j\in \An(r)$ with masses $c_{p(m,u)}b_{mj}=c_{p(j,m)}c_{p(m,u)}b_{jj}=b_{uj}$ for $j\in \An(r)$. The probability of the zero atom is given by
		\begin{align*}
			\P(A_{um}M_{mr}A_{rv}=0)
			&=1-\P(A_{um}>0)\P(M_{mr}>0)
			=1-c_{p(m,u)}\sum_{j\in \An(r)}b_{mj}
			\\&=1-\sum_{j\in \An(r)} c_{p(j,m)}c_{p(m,u)}b_{jj}
			=\sum_{j\in\An(u)}b_{uj}-\sum_{j\in\An(r)}b_{uj}
			=\sum_{j\in \An(u)\setminus \An(r)}b_{uj}.
		\end{align*} 
		Hence the distribution of $A_{um}M_{mr}A_{rv}$ is the one of $A_{uv}$ in Lemma~\ref{lem:Auv_rev1}-3.(c). This completes the proof that the statement in (i) implies (ii).	
		
		The statement in (iii) holds trivially from (ii).
		
		Next we prove that (iii) implies (i) by contraposition: we assume that $\T$ has at least two sources and we will show that it is not possible to obtain the factorization in \eqref{eqn:Auvprod}. If $\T$ has at least two sources, then by Lemma~\ref{lem:oned-1}-2 there is at least one v-structure, say on nodes $1,2,3$ and involving edges $(1,3), (2,3)\in E$. Consider the nodes $1,2$. For every $u\in V$ we have two possibilities:
		\begin{compactenum}
			
			\item[(a)] the v-structure belongs to only one of the trails $t(u, 1)$ and $t(u, 2)$: w.l.o.g. $(1,3), (2,3)\in t(u,2)$ and $(1,3), (2,3)\notin t(u,1)$;
			\item[(b)] each trail $t(u, 1)$ and $t(u, 2)$ contains one edge of the v-structure: 
			w.l.o.g. $(1,3)\in t(u,1)$ and $(2,3)\in t(u,2)$.
		\end{compactenum}
		If $u \in \{1, 2\}$, then we are in case~1, while if $u = 3$, we are in case~2. If $u \not\in \{1,2,3\}$, then node~3 must belong to at least one of the two trails $t(u, 1)$ or $t(u, 2)$, because otherwise the skeleton graph would have a cycle connecting nodes $u,1,2,3$ and passing through more than one block. The latter is impossible according to property~\ref{ttt:prop2}. The two possibilities are illustrated in Figure~\ref{fig:twocases}.
		\smallskip
		
		\begin{figure}
			\begin{subfigure}[b]{1\textwidth}
				\centering
				\begin{tikzpicture}
					\node[] (3) at (0,-1)  {$3$};
					\node[] (1) at (-1,0)  {$1$};
					\node[] (2) at (1,0)  {$2$};
					\path[->] (1) edge (3)  ;
					\path[->] (2) edge (3)  ;
					\node[] (u) at (-9,0)  {$u$};
					\node[] (u1) at (-7.5,0)  {$v_{n-1}$};
					\node[] (u2) at (-6,0)  {};
					\node[] (u3) at (-5,0)  {};	
					\node[] (u4) at (-4,0)  {};	
					\node[] (u5) at (-2.5,0)  {$v_2$};	
					\path[->] (1) edge (u5)  ;
					\path[<-] (u5) edge (u4)  ;
					\path[->] (u4) edge (u3)  ;
					\path[->] (u3) edge (u2)  ;
					\path[<-] (u2) edge (u1)  ;
					\path[->] (u1) edge (u)  ;
				\end{tikzpicture}
				\caption{When the v-structure belongs to only one of the two trails $t(u,1)$ or $t(u,2)$.}
			\end{subfigure}
			\vspace{0.2cm}
			
			\begin{subfigure}[b]{1\textwidth}
				\centering
				\begin{tikzpicture}
					\node[] (3) at (0,-1)  {$3$};
					\node[] (1) at (-1,0)  {$1$};
					\node[] (2) at (1,0)  {$2$};
					\path[->] (1) edge (3)  ;
					\path[->] (2) edge (3)  ;
					\node[] (u) at (-8,-1)  {$u$};
					\node[] (u1) at (-6.5,-1)  {$v_{n-1}$};
					\node[] (u2) at (-5,-1)  {};
					\node[] (u3) at (-4,-1)  {};	
					\node[] (u4) at (-3,-1)  {};	
					\node[] (u5) at (-1.5,-1)  {$v_2$};	
					\path[->] (3) edge (u5)  ;
					\path[<-] (u5) edge (u4)  ;
					\path[->] (u4) edge (u3)  ;
					\path[->] (u3) edge (u2)  ;
					\path[<-] (u2) edge (u1)  ;
					\path[->] (u1) edge (u)  ;
				\end{tikzpicture}
				\caption{When each node of the v-structure belongs to one of the two trails $t(u,1)$ and $t(u,2)$.}
			\end{subfigure}
			\caption{The two possible configurations of the trails $t(u,1)$ and $t(u,2)$ when nodes $1,2,3$ form a v-structure.}
			\label{fig:twocases}
		\end{figure}
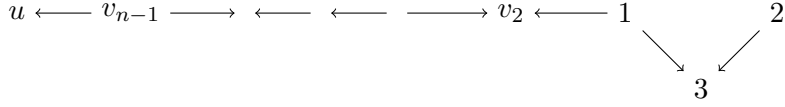
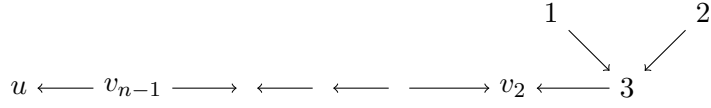
		
		\emph{Case 3.(c-i).}
		Consider first the case when, w.l.o.g., the v-structure belongs to $t(u,2)$ but not to $t(u,1)$, see Figure~\ref{fig:twocases}(a). Let the trail from $1$ to $u$ be on nodes $\{v_1=1, v_2, \ldots, v_n=u\}$. We can have any direction on the edges of $t(1,u)$.
		Recall the distribution of $A_{u2}$:
		\[
		\mathcal{L}(A_{u2})=\sum_{j\in\An(u)}b_{uj}\delta_{\{b_{2j}/b_{uj}\}}.
		\]	
		We have $b_{2j}=0$ for all $j\notin \An(2)$. 
		We claim that $\An(u)\cap \An(2)=\varnothing$. 
		According to property~\ref{ttt:prop2} of a ttt, $\T$ does not contain an undirected cycle involving several tournaments. This means that it is impossible to find a node from which there leave paths to $u$ and to $2$.  Also it is not possible to find a path passing through $3$ and going to $2$, because otherwise there would 
		be either an undirected cycle involving several tournaments, or a cycle within a tournament. Both are impossible for a ttt. 
		This leads to the conclusion that $A_{u2}$ is degenerate at zero. 
		Now we look at the variables $(M_{v_{i+1}v_i}, i=n-1, \ldots, 1; M_{13}, M_{32})$ which we take by construction to be independent as they belong to different tournaments. Each of them is one of the variables in Lemma~\ref{lem:Muv-rev1}, and none of these is degenerate at zero. Hence their product cannot be degenerate at zero too. 
		\smallskip
		
		\emph{Case 3.(c-ii).}
		Next we consider the second case, when w.l.o.g. $(1,3)\in t(u,1)$ and $(2,3)\in t(u,2)$, see Figure~\ref{fig:twocases}(b). 
		Let the trail from node $3$ to $u$ be on nodes $\{v_1=3, v_2, \ldots, v_n=u\}$. First we consider the case when we have at least one $i=1, \ldots, n-1$ for which $(v_{i+1},v_i)\in E$, i.e., we have at least one edge with direction from $u$ to $3$.
		Because $t(u,3)$ is a shortest trail, the edges incident to the nodes on the trail belong to different tournaments.  
		The distribution of $(A_{u1}, A_{u2})$ is given by	
		\[
		\mathcal{L}(A_{u1}, A_{u2})=
		\sum_{j\in\An(u)}b_{uj}
		\delta_{\left\{\frac{b_{1j}}{b_{uj}},\frac{b_{2j}}{b_{uj}}\right\}}, 
		\]
		where $b_{1j}=0$ and $b_{2j}=0$ if $j\notin \An(1)$ and $j\notin \An(2)$ respectively. When for some $i=1, \ldots, n-1$ we have $(v_{i+1}, v_i)\in E$ then necessarily $\An(1) \cap \An(u)=\varnothing$ and $\An(2) \cap \An(u)=\varnothing$. 
		There cannot be a path from $\An(1)$ or $\An(2)$ to any of the nodes $\{v_2,\ldots, v_n=u\}$, because otherwise there would be a cycle involving several tournaments in contradiction to the definition of a ttt. Because of the edge $(v_{i+1}, v_i)\in E$ all nodes in $\An(1)\cup\An(2)$ are not ancestors of $u$. And also because of the edges $(1,3), (2,3)\in E$ all nodes in $\An(u)$ cannot be ancestors of nodes $1$ or $2$. Thus when for some $i=1, \ldots, n-1$ there is a directed edge $(v_{i+1}, v_i)\in E$ we have $\mathcal{L}(A_{u1},A_{u2})=\delta_{\{0,0\}}$. We have found a node $v \in V$ such that $A_{uv} = 0$ almost surely, but then the factorisation \eqref{eqn:LXvXu}--\eqref{eqn:Auvprod} cannot hold, because these never degenerate at zero. 
		
		Now let the trail from node $3$ to $u$ be actually a path. Let also nodes $1$ and $2$ be sources with respect to the tournaments shared with node $3$, say $1,3\in V_{\tau_1}$ and $2,3\in V_{\tau_2}$. It is always possible to choose $1$ and $2$ in such a way they are the sources of $\tau_1$ and $\tau_2$. This is because node $3$ obviously is not a source in $\tau_1$ and $\tau_2$, so the sources of these must point to $3$. We can decompose $\An(u)$ into three disjoint sets, $\An(1), \An(2)$ and the rest, $\An(u)\setminus \{\An(1)\cup \An(2)\}$. For the distribution of $(A_{u1}, A_{u2})$ we have
		\begin{align*}
			\mathcal{L}(A_{u1}, A_{u2})
			&=\sum_{j\in \An(u)}b_{uj}
			\delta_{\left\{\frac{b_{1j}}{b_{uj}}, \frac{b_{2j}}{b_{uj}}\right\}}
			\\&=\sum_{j\in \An(1)}b_{uj}
			\delta_{\left\{\frac{b_{1j}}{b_{uj}}, \frac{b_{2j}}{b_{uj}}\right\}}
			+\sum_{j\in \An(2)}b_{uj}
			\delta_{\left\{\frac{b_{1j}}{b_{uj}}, \frac{b_{2j}}{b_{uj}}\right\}}
			+\sum_{j\in\An(u)\setminus \{\An(1)\cup \An(2)\}}b_{uj}
			\delta_{\left\{\frac{b_{1j}}{b_{uj}}, \frac{b_{2j}}{b_{uj}}\right\}}.
		\end{align*}
		For the atoms in the first summation we have 
		\[
		\frac{b_{1j}}{b_{uj}}=\frac{c_{p(j,1)}b_{jj}}{c_{p(j,u)}b_{jj}}
		=\frac{c_{p(j,1)}}{c_{p(j,1)}c_{13}c_{p(3,u)}}
		=\frac{1}{c_{13}c_{p(3,u)}}
		\]
		and $b_{2j}/b_{uj}=0$ as $b_{2j}=0$ for all $j\in \An(1)$. Hence we have an atom that does not depend on $j\in \An(1)$, i.e., $\big(1/(c_{13}c_{p(3,u)}),0\big)$ and its mass is 
		\[
		\sum_{j\in \An(1)}b_{uj}=\sum_{j\in \An(1)}c_{p(j,1)}c_{13}c_{p(3,u)}b_{jj}=c_{13}c_{p(3,u)}=c_{p(1,u)}.
		\] 
		In a similar way, from the second summation in the last display we have an atom $\big(0,1/(c_{23}c_{p(3,u)})\big)$ with mass $c_{23}c_{p(3,u)}=c_{p(2,u)}$. In the third summation term the atom is $(0,0)$ as $b_{1j}=b_{2j}=0$ for all $j\in\An(u)\setminus \{\An(1)\cup \An(2)\}$ and its mass is $1-c_{13}c_{p(3,u)}-c_{23}c_{p(3,u)}=1-c_{p(3,u)}(c_{13}+c_{23})$. 
		Consider now the multiplicative increments $(M_{31}; M_{32}, M_{v_{i+1}v_i} i=1, \ldots,n-1)$ which are mutually independent since they belong to different tournaments. Because node $1$ is a source node in the tournament $\tau_1$ the distribution of $M_{31}$ is $c_{13}\delta_{\{1/c_{13}\}}+(1-c_{13})\delta_{\{0\}}$ by Lemma~\ref{lem:Muv-rev1}-2.(a). Similarly for $M_{32}$. We have
		\begin{equation}  \label{eqn:M31M32}
			\begin{split} 
				&	\P\left(M_{31}\prod_{i=1}^{n-1}M_{v_{i+1}v_i}=0, M_{32}\prod_{i=1}^{n-1}M_{v_{i+1}v_i}=0\right)
				=1-\prod_{i=1}^{n-1}\P(M_{v_{i+1}v_i}>0)
				\\& +\P(M_{31}=0)\P(M_{32}=0)
				-\left(1-\prod_{i=1}^{n-1}\P(M_{v_{i+1}v_i}>0)\right)
				\P(M_{31}=0)\P(M_{32}=0).
			\end{split}
		\end{equation}
		After some rearranging of the expression above we obtain
		\begin{equation} \label{eqn:1-prod}
			1-\prod_{i=1}^{n-1}\P(M_{v_{i+1}v_i}>0)(c_{13}+c_{23}-c_{13}c_{23}).
		\end{equation}
		
		There are two further sub-cases: either all nodes $v_1,\ldots,v_{n-1}$ are source nodes with respect to the tournament involving the next node in the sequence, or not. In the first sub-case, namely when all nodes in $\{v_1=3, v_2, \ldots, v_{n-1}\}$ are source nodes with respect to the tournament involving the next node in the sequence, then $\P(M_{v_{i+1}v_i}>0)=c_{v_{i}v_{i+1}}$ for $i=1, \ldots, n-1$. This means that the probability in \eqref{eqn:1-prod} and accordingly in \eqref{eqn:M31M32} equals $1-c_{p(3,u)}(c_{13}+c_{23}-c_{13}c_{23})$, which is different than $\P(A_{u1}=0, A_{u2}=0)=1-c_{p(3,u)}(c_{13}+c_{23})$. 
		In the second sub-case, i.e., if at least one node from $\{v_1=3, v_2, \ldots, v_{n-1}\}$ is not source with respect to the tournament involving the next node in the sequence then the possible values for $M_{31}\prod_{i=1}^{n-1}M_{v_{i+1}v_i}$ are not only $\{0,1/c_{p(1,u)}\}$, 
		which are the only possible values of $A_{u1}$ as we showed in the previous paragraph.
		Let $i \in \{1,\ldots,n-1\}$ be such that node $v_i$ is not the source node in the tournament shared with $v_{i+1}$, say $\tau_i$. This is depicted in the following graph:
		
		
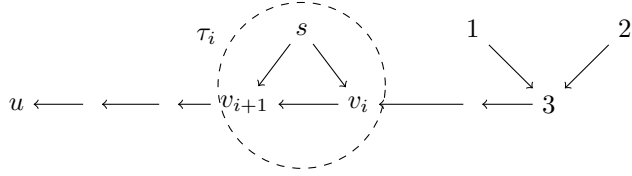
\begin{center} 
			\begin{tikzpicture}
				\node[] (3) at (0,-1)  {$3$};
				\node[] (1) at (-1,0)  {$1$};
				\node[] (2) at (1,0)  {$2$};
				\path[->] (1) edge (3)  ;
				\path[->] (2) edge (3)  ;
				\node[] (u) at (-7,-1)  {$u$};
				\node[] (u1) at (-6,-1)  {};
				\node[] (u2) at (-5,-1)  {};
				\node[] (u3) at (-4,-1)  {$v_{i+1}$};	
				\node[] (u4) at (-2.5,-1)  {$v_{i}$};	
				\node[] (u5) at (-1,-1)  {};
				\node[] (s) at (-3.25,0)  {$s$};	
				\path[->] (3) edge (u5)  ;
				\path[->] (u5) edge (u4)  ;
				\path[->] (u4) edge (u3)  ;
				\path[->] (u3) edge (u2)  ;
				\path[->] (u2) edge (u1)  ;
				\path[->] (u1) edge (u)  ;
				\path[->] (s) edge (u3)  ;
				\path[->] (s) edge (u4)  ;
				\node[] (tau) at (-4.5,-0.1)  {$\tau_i$};
				\begin{scope}[dashed]
					\draw[color=black] (-3.25,-0.75) circle (1.1cm);
				\end{scope}
			\end{tikzpicture}
		\end{center}
		Recall the distribution of $M_{v_{i+1}v_i}$ from Lemma~\ref{lem:Muv-rev1}-2.(b):
		\[
		\mathcal{L}(M_{v_{i+1}v_i})
		=\sum_{j\in \An(v_{i})}b_{v_{i+1}j}\delta_{\{b_{v_ij}/b_{v_{i+1}j}\}}
		+\sum_{j\in \An(v_{i+1})\setminus \An(v_i)}\delta_{\{0\}}.
		\]   
		Take for instance a node, say $s$, a parent of $v_i$ and accordingly in $\An(v_i)$. 
		Then 
		\[
		\frac{b_{v_is}}{b_{v_{i+1}s}}
		=\frac{c_{sv_i}}{c_{sv_{i+1}}}
		\]
		is a possible value of $M_{v_{i+1}v_i}$ with positive probability, namely at least $b_{v_{i+1}s}$. Another possible positive value is for $j=v_i\in \An(v_i)$, namely 
		\[
		\frac{b_{v_iv_i}}{b_{v_{i+1}v_i}}
		=\frac{1}{c_{v_iv_{i+1}}}
		\]
		with probability at least $b_{v_{i+1}v_i}$. The criticality assumption on edge weights guarantees $\frac{c_{sv_i}}{c_{sv_{i+1}}}\neq 1/c_{v_iv_{v+1}}$. This means that the product $M_{31}\prod_{i=1}^{n-1}M_{v_{i+1}v_i}$ has at least two different positive values - one involving $\frac{c_{sv_i}}{c_{sv_{i+1}}}$ and another $1/c_{v_iv_{i+1}}$. However $A_{u1}$ has only one possible positive value. 
%
	\end{proof}

	\subsubsection*{Proof of Proposition~\ref{prop:markov_ml}}
	
	\begin{proof} 
		\emph{Sufficiency.} Assume $\T$ has a unique source.
		We need to show that, for any disjoint and nonempty sets $A, B, S$, we have $X_A\indep X_B\mid X_S$, whenever $S$ is a separator of $A$ and $B$ in the skeleton $T$ of $\T$. We would like to use Theorem~5.15 in 
		\citet*{amendola2022conditional}, by which we need to show $A\perp_*B\mid S$ in $\mathcal{D}_S^*$, that is, there are no $*$-connecting paths between any pair of nodes in $A$ and $B$ in the \emph{conditional reachability} DAG $\mathcal{D}_S^*$. We will explain these notions further.
		
		Let $A, B,S\subset V$ be nonempty disjoint node sets, such that $S$ is a separator of $A$ and $B$ in the skeleton $T$. Consider Figure~\ref{fig:star-connect}. According to Definition~5.4 of \citet{amendola2022conditional}, a $*$-connecting path between $a \in A$ and $b \in B$ is one of the five configurations therein. Our goal is to show that for $a\in A$ and $b\in B$ it is impossible to find a $*$-connecting path in a certain graph $\mathcal{D}_S^*$, which is not $\T$, neither $T$, but constructed under particular rules given in \citet[Definition~5.1]{amendola2022conditional}.  
		
		\begin{figure}[h] 
			\centering
			\begin{tikzpicture} 
				\node[] (i) at (0,0) {$a$};
				\node[] (j) at (0,-1.5) {$b$};
				\path[->] (i) edge (j) ;
			\end{tikzpicture}
			\hspace{0.5cm}
			\begin{tikzpicture} 
				\node[] (i') at (0,0) {$a'$};
				\node[] (j) at (-0.75,-1.5) {$b$};
				\node[] (i) at (+0.75,-1.5) {$a$};
				\path[->] (i') edge (j) ;
				\path[->] (i') edge (i) ;
			\end{tikzpicture}
			\hspace{0.5cm}
			\begin{tikzpicture} 
				\node[] (i) at (0,0) {$a$};
				\node[] (j) at (1.5,0) {$b$};
				\node[color=red] (k) at (+0.75,-1.5) {$s$};
				\path[->] (i) edge (k) ;
				\path[->] (j) edge (k) ;
			\end{tikzpicture}
			\hspace{0.5cm}
			\begin{tikzpicture} 
				\node[] (i') at (0,0) {$a'$};
				\node[] (j) at (1.5,0) {$b$};
				\node[color=red] (k) at (+0.75,-1.5) {$s$};
				\node[] (i) at (-0.75,-1.5) {$a$};
				\path[->] (i') edge (k) ;
				\path[->] (j) edge (k) ;
				\path[->] (i') edge (i) ;
			\end{tikzpicture}
			\hspace{0.5cm}
			\begin{tikzpicture} 
				\node[] (i') at (0,0) {$a'$};
				\node[] (j') at (1.5,0) {$b'$};
				\node[color=red] (k) at (+0.75,-1.5) {$s$};
				\node[] (i) at (-0.75,-1.5) {$a$};
				\node[] (j) at (2.25,-1.5) {$b$};
				\path[->] (i') edge (k) ;
				\path[->] (j') edge (k) ;
				\path[->] (i') edge (i) ;
				\path[->] (j') edge (j) ;
			\end{tikzpicture}
			\caption{According to Definition~5.4 of \cite{amendola2022conditional}, a $*$-connected path between $a$ and $b$ relative to $S$ is one of the five configurations above. In the last three graphs we have $s\in S$.}
			\label{fig:star-connect}
		\end{figure}
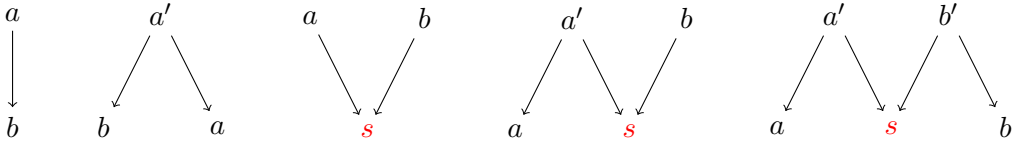
		
		According to this definition, the conditional reachability graph $\mathcal{D}_S^*$ is on the same vertex set, $V$. 
		Between two nodes $i$ and $j$ in $V$ there is an edge $(i, j)$ in $\mathcal{D}_S^*$ if and only if there is a directed path from $i$ to $j$ in $\T$ such that no node on that path belongs to $S$, except possibly for $i$ and $j$ themselves.
		
		Consider Figure~\ref{fig:star-connect}. We need to show that in the conditional reachability graph $\mathcal{D}_S^*$, there is no $*$-connecting path between a node $a \in A$ and a node $b \in B$.
		
		To obtain the first configuration in $\mathcal{D}_S^*$, there must be, in the skeleton $T$, nodes $a\in A$ and $b\in B$ such that no node on the path from $a$ to $b$ passes through $S$. But this is impossible, because we assumed that $S$ is a separator of $A$ and $B$ in $T$. Similarly for the second configuration.
		
		For the other three configurations in Figure~\ref{fig:star-connect} consider Figure~\ref{fig:star_connect}. 
		
		In Figure~\ref{fig:star_connect}, the left-hand and right-hand trails in the original graph $\T$ are the only possible one that give rise to the middle path in Figure~\ref{fig:star-connect} with respect to the graph $\mathcal{D}_S^*$: both the left-hand and right-hand graphs in Figure~\ref{fig:star_connect} show existing trails between $a$ and $b$ in $\T$, trails composed of a directed path from $a$ to $s$, and a directed path from $b$ to $s$. The only node on these trails which belongs to $S$ is $s$. Hence in $\mathcal{D}_S^*$ we put a directed edge from $a$ to $s$ and from $b$ to $s$. This gives the third $*$-connecting path in Figure~\ref{fig:star-connect}. But this configuration cannot occur, for the following reason.
		On the left-hand trail in Figure~\ref{fig:star_connect}, the separator node $s$ has parents $u_r$ and $v_q$ in different tournaments. This leads to a v-structure between the nodes $u_r,s,v_q$, in contradiction to Lemma~\ref{lem:oned-1}-2 and the hypothesis that $\T$ has a unique source.
		
		On the right-hand trail in Figure~\ref{fig:star_connect}, the node $s$ shares a tournament with its parents $u_r$ and $v_q$, but only $s$ belongs to $S$; on the trail $\{a=u_1, u_2, \ldots, u_r, v_q, \ldots, v_2,v_1=b\}$ none of the nodes are in $S$. In $T$, this means that there is a path between $A$ and $B$ that does not pass through $S$. This is in contradiction to the assumption that $S$ separates $A$ and $B$ in $T$.
		

		To show that the fourth type of $*$-connecting path in Figure~\ref{fig:star-connect} cannot occur, we can use the reasoning used for the third one by setting $a=a'$ in Figure~\ref{fig:star_connect}. Then either $s$ has parents from two different tournaments or there is a non-directed path from $a'$ to $b$ which does not pass through $S$. The first case is excluded by Lemma~\ref{lem:oned-1}-2 and the assumption that $\T$ has a unique source, and the second one by the assumption that $S$ is a separator of $A$ and $B$ in $T$. The impossibility of the fifth $*$-connected configuration follows analogously.
		\smallskip
		
		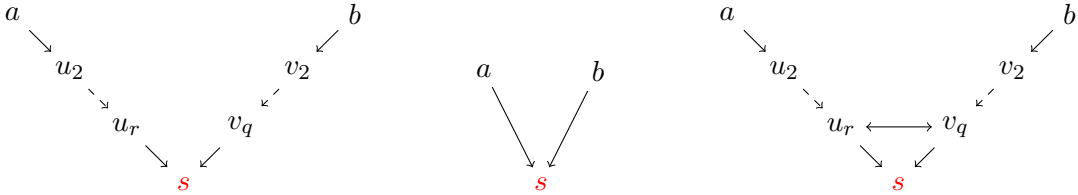
\begin{figure}[h]
			\centering
			\begin{tikzpicture} 
				\node[] (a) at (0,0) {$a$};
				\node[] (a1) at (0.75,-0.75) {$u_2$};
				\node[] (ar) at (1.5,-1.5) {$u_r$};
				\node[color=red] (s) at (2.25,-2.25) {$s$};
				\node[] (bq) at (3,-1.5) {$v_q$};
				\node[] (b2) at (3.75,-0.75) {$v_2$};
				\node[] (b) at (4.5,0) {$b$};
				\path[->] (a) edge (a1)  ;
				\path[->, dashed] (a1) edge (ar);
				\path[->] (ar) edge (s)  ;
				\path[->] (b) edge (b2)  ;
				\path[->, dashed] (b2) edge (bq);
				\path[->] (bq) edge (s)  ;
			\end{tikzpicture}
			\hspace{1cm}
			\begin{tikzpicture} 
				\node[] (a) at (0,0) {$a$};
				\node[color=red] (s) at (0.75,-1.5) {$s$};
				\node[] (b) at (1.5,0) {$b$};
				\path[->] (a) edge (s)  ;
				\path[->] (b) edge (s)  ;
			\end{tikzpicture}
			\hspace{1cm}
			\begin{tikzpicture} 
				\node[] (a) at (0,0) {$a$};
				\node[] (a1) at (0.75,-0.75) {$u_2$};
				\node[] (ar) at (1.5,-1.5) {$u_r$};
				\node[color=red] (s) at (2.25,-2.25) {$s$};
				\node[] (bq) at (3,-1.5) {$v_q$};
				\node[] (b2) at (3.75,-0.75) {$v_2$};
				\node[] (b) at (4.5,0) {$b$};
				\path[->] (a) edge (a1)  ;
				\path[->, dashed] (a1) edge (ar);
				\path[->] (ar) edge (s)  ;
				\path[->] (b) edge (b2)  ;
				\path[->, dashed] (b2) edge (bq);
				\path[->] (bq) edge (s)  ;
				\path[<->] (bq) edge (ar)  ;
			\end{tikzpicture}
			\caption{The left and right trails in the original graph $\T$ are the only possible trails that give rise to the middle path in the graph $\mathcal{D}_S^*$.}
			\label{fig:star_connect}
		\end{figure}
		
		\emph{Necessity.} We will show that if $\T$ has multiple source nodes, there is a triple of disjoint, non-empty sets $A,B,S\subset V$ such that $S$ is a separator of $A$ and $B$ in $T$, but $X_A$ and $X_B$ are conditionally dependent given $X_S$. In case $\T$ has at least two sources, we have at least one v-structure in $\T$ by Lemma~\ref{lem:oned-1}-2. Take a triple of nodes in a v-structure, say $u, v, w$, with $u$ and $w$ being parents of $v$. Then node $v$ separates nodes $u$ and $w$ in $T$, i.e., $S = \{v\}$ is a separator of $A = \{u\}$ and $B = \{w\}$ in $T$. All references below are from \citet{amendola2022conditional}.
		
		To show $X_u\notindep X_w\mid X_v$ 	we will use Theorem~6.18 (Context free completeness) of \citet{amendola2022conditional}. We need to show that there is an \emph{effective} $*$-connecting path in the critical DAG $\mathcal{D}^*_S(\theta)$ between nodes $u$ and $w$ as in their Definitions~5.2 and~6.5.
		
		The subgraph on nodes $u,v,w$ of $\mathcal{D}^*_S(\theta)$ is a v-structure, $u\longrightarrow \textcolor{red}{v}\longleftarrow w$, according to the definition of $\mathcal{D}^*_S(\theta)$. According to Definition~6.4, the $|S| \times |S|$ substitution matrix of $(u,v)\in E$ relative to $S=\{v\}$ is zero, because $S$ is a singleton and by definition all diagonal entries of the substitution matrix are zero, i.e., $\Xi^{vu}_S=0$. Similarly, $\Xi^{vw}_S=0$.
		Because the edges $(u,v), (w,v)$ form a $*$-connecting path between $u,w$ in $\mathcal{D}^*_S(\theta)$, the substitution matrix of this path relative to $S$, say $\Xi_S$, is zero too:
		\[
		\Xi_S=\max(\Xi^{vu}_S, \Xi^{vw}_S)=0.
		\]
		To find out if the edges $(u,v), (w,v)$ form an \emph{effective} $*$-connecting path between $u,w$, we need to compute the tropical eigenvalue of $\max(\Gamma_{SS},\Xi_S)$ where $\Gamma$ is as in Equation~(2.3) in \cite{amendola2022conditional} and $\Gamma_{SS}$ is the $vv$-element of $\Gamma$, i.e., $\{\Gamma\}_{vv}$. Because $\{\Gamma\}_{ij}>0$ if and only if there is a directed path from $j$ to $i$, we have $\Gamma_{SS}=\{\Gamma\}_{vv}=0$ and so 
		\[
		\max(\Gamma_{SS}, \Xi_S)=0.
		\]
		The tropical eigenvalue \citep[Equation~(2.7)]{amendola2022conditional} of the above matrix is trivially equal to zero and thus smaller than one. By Definition~6.5 in the cited reference, there is indeed an \emph{effective} $*$-connecting path between $u,w$ in $\mathcal{D}^*_{S}(\theta)$. In view of their Theorem~6.18, we conclude $X_u \notindep X_w \mid X_v$.
		
	\end{proof}

	\section{Proofs and additional results for Section~\ref{sec:ident}}
	\subsection{Auxiliary results}
	
	\begin{proof}[Proof of Lemma~\ref{lem:Hth}]
		The point masses satisfy $m_i > 0$ for all $i\in V$ because we have $m_i=0$ if and only if $c_{ii} = 0$. However $c_{ii}=0$ is impossible in view of the definition in \eqref{eq:cvv}. Therefore we cannot have undefined atoms, which would happen when $m_i=0$. This shows (i).
		
		Next we show (ii). To see why $a_i\neq a_j$ for $i\neq j$, let $i, v \in V$ and recall $b_{vi}$ in \eqref{eq:bvi-rev1}. From the line below \eqref{eq:Hth}, recall that we also have $b_{vi} = m_i a_{vi}$ for $i, v \in V$.
		Thanks to the assumption $\theta\in \mathring{\Theta}_*$, 
		we have \eqref{eq:equivs}. 
We also have for any DAG \eqref{eq:Descdistinct}. 
		
		The combination of the last two equations implies that in \eqref{eq:Hth}, all vectors $a_i$ for $i \in V$ are different and thus that $H_\theta$ has $|V|$ distinct atoms. Also, for every node $i \in V$, we can find out which of the $|V|$ atoms of $H_\theta$ is $a_i$ because it is the unique one that satisfies $\Desc(i) = \{ v \in V : a_{vi} > 0 \}$. Note that similarly, among the $|V|$ vectors in the set $\cB_\theta = \{ (b_{vj})_{v \in V} : j \in V \}$, the vector $b_i$ is the unique one such that $\Desc(i) = \{ v \in V : b_{vi} > 0 \}$.	 
		
		Finally consider (iii). By the criticality assumption, every edge is critical, because it is the shortest path between any pair of adjacent nodes. Since $(i, v)\in E$ is critical, we have $b_{vi} = b_{ii} c_{iv}$  and thus $c_{iv} = b_{vi} / b_{ii} = a_{vi} / a_{ii}$.
		
		In summary, the angular measure $H_\theta$ possesses $|V|$ distinct atoms that can be uniquely matched to the nodes. As a consequence, we can reconstruct the matrix $(b_{vi})_{i,v \in V}$. Thanks to (iii), this matrix allows us to recover all edge weights $c_{vi}$.
	\end{proof}
	
	
	\begin{lema1} \label{prop:properties2}
		Let $\T= (V, E)$ be a ttt as in Definition~\ref{def:ttt} and let $\T$ have a unique source, $u_0$. Let $U \subset V$ be non-empty and suppose that $\bar{U} = V \setminus U$ satisfies conditions~\ref{ident1} and \ref{ident2}. 
			For every $\bar{u} \in \bar{U}$ there exists $s \in \desc(\bar{u}) \cap U$ such that $\pi(\bar{u}, s)$ is a singleton and the unique path $p$ from $\bar{u}$ to $s$ satisfies the following two properties:
			\begin{compactenum}
				\item all nodes on $p$ except for $s$ are in $\bar{U}$;
				\item all nodes on $p$ except possibly for $\bar{u}$ have only one parent.
			\end{compactenum}  
		As a consequence, any path with destination $s$ must either start in one of the nodes of $p$ or contain $p$ as a sub-path.
		\end{lema1}
	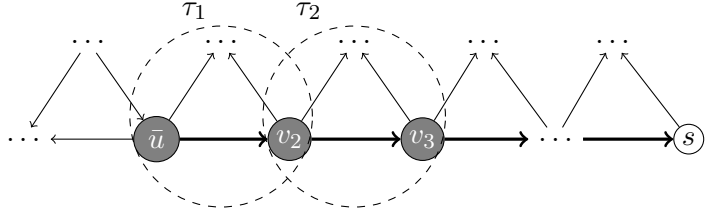
\begin{figure} 
		\centering
	\begin{tikzpicture}		
		\node[hollow node, fill=gray, minimum size=0.6cm] (ubar) at (0,0) {$\textcolor{white}{\bar{u}}$};
		\node[] (v1) at (-0.9,1.3) {$\ldots$};
		\node[] (v2) at (-1.75,0) {$\ldots$};
		\node[hollow node, fill=gray] (u2) at (1.75,0) {$\textcolor{white}{v_2}$};
		\node[] (w1) at (0.85,1.3) {$\ldots$};
		\node[hollow node, fill=gray] (u3) at (3.5,0) {$\textcolor{white}{v_3}$};
		\node[] (w2) at (2.6,1.3) {$\ldots$};
		\node[] (u4) at (5.25,0) {$\ldots$};
		\node[] (w3) at (4.3,1.3) {$\ldots$};
		\node[] (w4) at (6,1.3) {$\ldots$};
		\node[hollow node] (u5) at (7,0) {$s$};
		\path[->, line width=1.2pt] (ubar) edge (u2)  ;
		\path[->, line width=1.2pt] (u2) edge (u3)  ;
		\path[->, line width=1.2pt] (u3) edge (u4)  ;
		\path[->, line width=1.2pt] (u4) edge (u5)  ;
		\path[->] (v1) edge (ubar)  ;
		\path[->] (ubar) edge (v2)  ;
		\path[->] (v1) edge (v2)  ;
		\path[->] (ubar) edge (w1)  ;
		\path[->] (u2) edge (w1)  ;
		\path[->] (u2) edge (w2)  ;
		\path[->] (u3) edge (w2)  ;
		\path[->] (u3) edge (w3)  ;
		\path[->] (u4) edge (w3)  ;
		\path[->] (u4) edge (w4)  ;
		\path[->] (u5) edge (w4)  ;
		\node[] (tau) at (0.5,1.7)  {$\tau_1$};
		\node[] (tau) at (2,1.7)  {$\tau_2$};
		\begin{scope}[dashed]
			\draw[color=black] (0.85,0.3) circle (1.2cm);
			\draw[color=black] (2.6,0.3) circle (1.2cm);
		\end{scope}
	\end{tikzpicture}
\caption{A unique path on nodes $\{\bar{u}=v_1, v_2, v_3, \ldots, s = v_n\}$ under Lemma~\ref{prop:properties2}. Each of the nodes $\bar{u}, v_2, \ldots, v_{n-1}$ belongs to $\bar{U}$. Each of the nodes $v_2, \ldots, v_{n-1}, s$ has a unique parent. The node $\bar{u}\in \bar{U}$ may have parents as illustrated here, but then there is at least one tournament with respect to which it is a source node, e.g., $\tau_1$. Let $v_2$ be the node with unique parent $\bar{u}$ in $\tau_1$. When $v_2$ belongs to $\bar{U}$, it must participate in at least one another tournament, say $\tau_2$. In $\tau_2$ the node with unique parent $v_2$ is $v_3$. In this principle the path continues until we find a node in $U$, which is $v_n = s$ in this case.}
\label{fig:unipath}
\end{figure}	
			\begin{proof}
			Let $\bar{u}\in \bar{U}$ and suppose $\bar{u}$ has no parents, so $\bar{u}=u_0$. Take a node whose unique parent is $\bar{u}$, say $v_2$. By \citet[Corollary~5.a]{harari1966the} such a node exists in every tournament in which $\bar{u}$ takes part. If $v_2\in U$ then $s=v_2$ and we are done. If $v_2\in \bar{U}$ then by \ref{ident2} $v_2$ must be a source of at least one another tournament. 
			In each of these, there is a node whose only parent is $v_2$. Take such a node, say $v_3$. If $v_3\in U$ then $v_3=s$ and we are done; otherwise continue in the same way until we find a node which is in $U$. Because the graph is finite and because of condition \ref{ident2} such a node must exist. It is clear that the path constructed in this way has the stated properties.
			
			Next suppose that $\bar{u}$ belongs to $\bar{U}$ and that $\bar{u}$ has at least one parent. By~\ref{ident2} it must be a source of at least one tournament 
			In each of these tournaments there is a node with single parent $\bar{u}$. Take one of them, say $v_2$, and if $v_2\in U$ then we are done, otherwise repeat the same procedure as above until we find a node which is in $U$. Because the graph is finite and because of condition \ref{ident2}
			such a node must exist. It is clear that the path constructed in this way has the stated properties too.
			\end{proof}

			\begin{lema1} \label{prop:properties2-1}
		Let $\T= (V, E)$ be a ttt as in Definition~\ref{def:ttt} and let $\T$ have a unique source, $u_0$. Let $U \subset V$ be non-empty and suppose that $\bar{U} = V \setminus U$ satisfies conditions~\ref{ident1} and \ref{ident2}. 
		Let $i, j \in V$ be two distinct nodes. Upon switching the roles of $i$ and~$j$ if needed, the equality $\Desc(i) \cap U = \Desc(j) \cap U$ implies the following properties:
			\begin{compactenum}
				\item $i \in \bar{U}$;
				\item $\desc(i) = \Desc(j)$;
				\item $\{i\} = \pa(j)$;
				\item there exists $u \in V$ such that $i, j \in \pa(u)$;
				\item for $k \in V \setminus \{i, j\}$, the set $\Desc(k) \cap U$ is different from $\Desc(i) \cap U = \Desc(j) \cap V$;
				\item $|\Desc(i)\cap U|=|\Desc(j)\cap U|\geq 2$.
			\end{compactenum}
	\end{lema1}
	
	
	
	\begin{proof}

		1. Note first that $\Desc(i) \cap U$ cannot be empty, for otherwise, we would have $\Desc(i) \subseteq \bar{U}$, but this is impossible, since $\Desc(i)$ contains at least one leaf node (a node without children), in contradiction to \ref{ident1}.
		
		Since $\Desc(i) \cap U = \Desc(j) \cap U$ and since this set is non-empty, the intersection $\Desc(i) \cap \Desc(j)$ is not empty too. 
		In relation to 
		Lemma~\ref{lem:oned-1}-3 this means either $\Desc(i) \subseteq \Desc(j)$ or $\Desc(j) \subseteq \Desc(i)$.	
		In the remainder of the proof, we suppose $\Desc(j) \subseteq \Desc(i)$. Then we must have $i \not\in \Desc(j)$, since otherwise also $\Desc(i) \subseteq \Desc(j)$ and thus $\Desc(i) = \Desc(j)$, which is impossible since $i$ and $j$ are distinct; see \eqref{eq:Descdistinct}.
		From $\Desc(j)\subseteq\Desc(i)$ and $\Desc(i) \cap U = \Desc(j) \cap U$ it follows that 
		\[ 
		\Desc(i) \setminus \Desc(j) \subseteq \bar{U}.
		\]
		Because $i \not\in \Desc(j)$ we get $i\in \bar{U}$.
		\smallskip
		
		2--4. 
		First we show that all elements in $\Desc(i) \setminus \Desc(j)$ are ancestors of $j$. Let $v \in \Desc(i) \setminus \Desc(j)$. Because $j$ and $v$ are two different nodes, Lemma~\ref{lem:oned-1}-3 implies that one of three cases must occur: $\Desc(v) \subset \Desc(j)$; $\Desc(j) \subset \Desc(v)$; or $\Desc(j) \cap \Desc(v) = \varnothing$.
		The first case, $\Desc(v) \subset \Desc(j)$, is impossible, since $v \not\in \Desc(j)$.
		The third case, $\Desc(j) \cap \Desc(v) = \varnothing$, is impossible too, since it would imply that $\Desc(v) \subseteq \Desc(i) \setminus \Desc(j) \subseteq \bar{U}$, but this cannot happen since $\Desc(v)$ contains at least one leaf node while $\bar{U}$ does not contain any.
		Only the second case, $\Desc(j) \subset \Desc(v)$, remains.
		As a consequence, $v$ is an ancestor of $j$, and so all nodes of $\Desc(i)\setminus \Desc(j)$ are ancestors of $j$. By the proof of point~1, we get
		\begin{equation}
			\label{eq:Descij}
			\Desc(i) \setminus \Desc(j) \subseteq \an(j) \cap \bar{U}.
		\end{equation}
		
		Let again $v \in \Desc(i) \setminus \Desc(j)$. We show that there exists a unique path from $v$ to $j$ and that $j$ has only a single parent. Since $v \in \bar{U}$, there exists, by Lemma~\ref{prop:properties2}, a node $s(v) \in U$ such that there is only directed path $p(v, s(v))$ from $v$ to $s(v)$; moreover, this path satisfies properties~1 and~2 of the statement. Necessarily, 
		\[ 
		s(v) \in \Desc(v) \cap U \subseteq \Desc(i) \cap U = \Desc(j) \cap U.
		\]
		As $s(v)$ is a descendant of $j$ while $j$ is a descendant of $v$, the path $p(v, s(v))$ passes by $j$. As a consequence, there is a unique path $p(v, j)$ from $v$ to $j$; otherwise, there would be more than one path from $v$ to $s(v)$. Moreover, by Lemma~\ref{prop:properties2}-2, all nodes on the path $p(v, s(v))$, except possibly for $v$, have only one parent. In particular, $j$ has only one parent.
		
		
		
		Take again any $v \in \Desc(i) \setminus \Desc(j)$. By \ref{ident1}, $v$ has at least two children. They cannot both be ancestors of $j$, since then there would be two paths from $v$ to $j$, in contradiction to the previous paragraph. Let $u$ be a child of $v$ that is not an ancestor of $j$; then $u \in \Desc(j)$ because of \eqref{eq:Descij}. This means there are two paths from $v$ to $u$: the edge $(v, u)$ and a path passing through $j$. These paths must belong to the same tournament, as the skeleton of $\T$ is a block graph. But then $v$ and $j$ are adjacent, and thus $v$, which we already knew to be an ancestor of $j$, is actually a parent of $j$. But $j$ has only one parent, and so the set $\Desc(i) \setminus  \Desc(j)$ must be a singleton. As this set obviously contains node $i$, we get $v = i$ and thus $\Desc(i) \setminus \Desc(j) = \pa(j) = \{i\}$. Since $\Desc(i) = \{i\} \cup \desc(i)$ and $\Desc(j) \subset \Desc(i)$, it follows that $\desc(i) = \Desc(j) = \Desc(i) \setminus \{i\}$. \smallskip

		5. 
		%
		Let $k \in V \setminus \{i,j\}$ be such that $\Desc(k) \cap U = \Desc(i) \cap U = \Desc(j) \cap U$. By point~3 we have $\{i\}=\pa(j)$. From  $\Desc(k) \cap U = \Desc(i) \cap U$ we can have either $\{k\}=\pa(i)$ or $\{i\}=\pa(k)$, while from $\Desc(k) \cap U = \Desc(j) \cap U$ we have either $\{k\}=\pa(j)$ or $\{j\}=\pa(k)$. Because already $\{i\}=\pa(j)$, we cannot also have $\{k\}=\pa(j)$, whence we must have $\{j\}=\pa(k)$. But then $\{i\}=\pa(k)$ is impossible, so that necessarily $\{k\}=\pa(i)$. From $\{i\}=\pa(j)$, $\{j\}=\pa(k)$, and $\{k\}=\pa(i)$ we get a cycle between the three nodes $i,j,k$ which is a contradiction to the definition of a DAG.
		\smallskip
		
		6. We already know that $\Desc(i) \cap U = \Desc(j) \cap U$. We need to show that this set contains at least two elements. Consider the triple $\{i,j,u\}$ from point~4 that forms a triangle with directed edges $(i,j)$, $(i,u)$, and $(j,u)$. 
		By point~1 we have also $i\in \bar{U}$. There are four cases, according to whether $j$ and $u$ belong to $U$ or not.
		\begin{compactitem}
			\item 
			If $j,u\in U$, they are two distinct elements of $\Desc(j) \cap U$. 
			\item
			If $j \in U$ but $u \in \bar{U}$, then take $r \in \Desc(u) \cap U$ [which is non-empty by \ref{ident1}: all leaf nodes in $\Desc(u)$ are in $U$], and note that $j$ and $r$ are two distinct elements in $\Desc(j) \cap U$.
			\item 
			If $j \in \bar{U}$ and $u \in U$, then, as in Lemma~\ref{prop:properties2}, let $s \in U$ be such that there is unique path $p(j, s)$ from $j$ to $s$, this path satisfying properties~1-2 in the same lemma. Then $u$ does not belong to that path (since $u \in U$ and $u$ has at least two parents, $i$ and $j$), so that $s$ is different from $u$, and both are members of $\Desc(j) \cap U$.
			\item
			If $j, u \in \bar{U}$, then we can find by Lemma~\ref{prop:properties2} nodes $s\in \desc(j)\cap U$ and $r\in \desc(u)\cap U$ with paths $p(j,s)$ and $p(u,r)$ which satisfy the characteristics in this lemma. Nodes $s,r$ clearly belong to $\Desc(j)\cap U$. Moreover, they are distinct: node $u$, having at least two parents, cannot belong to the unique path $p(j, s)$ between $j$ and $s$, while by construction, there is a directed path from $j$ to $r$ that passes along $u$. \qedhere
		\end{compactitem}	 
	\end{proof}

	\begin{lema1} 
		\label{lem:bzz}
		Let $X$ be a max-linear model with respect to a ttt with unique source. The coefficient $b_{vv}$ depends only on the edge weights of the tournament shared by node $v\in V$ and its parents, and it is given by
		\begin{equation} \label{eqn:bvv_expression}
			b_{vv}=1+\sum_{u\in \pa(v)} \sum_{p\in \pi(u,v) }
			(-1)^{|p|}c_p.
		\end{equation}
	\end{lema1}	
	\begin{proof} 
Consider the node $v$. If $v$ has no parents we have $\an(v)=\varnothing$ and by \eqref{eq:cvv} we have $b_{vv}=1$. If $v$ has at least one parent then there is a tournament which contains the parents, say, $\tau=(V_{\tau}, E_{\tau})$. Let the nodes in $\tau$ be labelled according to their in-/out-degree ordering in $\tau$: the node with $|V_{\tau}|-1$ children in $\tau$ (the source of $\tau$) has index~1, the node with $|V_{\tau}|-2$ children in $\tau$ has index~2, and so on. 
 We can partition the set $\an(v)$ into $\An(1)$ and $\pa(v)\setminus \{1\}$. For $i\in \An(1)$ the shortest path from $i$ to $v$ passes necessarily through $1$, so $b_{vi}=c_{p(i,1)}c_{1v}b_{ii}$. Then we have by~\eqref{eq:bvi} and~\eqref{eq:cvv}
\begin{align} \label{eqn:bvv_expression0}
	b_{vv}=1-\sum_{i\in \an(v)}b_{vi}=
	1-\sum_{i\in \An(1)}c_{p(i,1)}c_{1v}b_{ii}
	-\sum_{i\in \pa(v)\setminus 1}b_{vi}
	=
	1-c_{1v}-\sum_{i\in \pa(v)\setminus 1}c_{iv}b_{ii}.
\end{align}			
		Let $C=\{c_{ij}\}_{i,j\in V_\tau, i<j}$ be the matrix of edge weights within $\tau$: it is lower triangular and has zero diagonal. Let $I_{m}$ denote the $m \times m$ identity matrix, write $\bm{b}=(1,b_{22},\ldots,b_{|V_\tau|,|V_\tau|})^\top$ (a column vector) and let $1_{|V_\tau|}$ be a column vector of ones of length $|V_\tau|$. Consider the system of linear equations 
		\begin{equation} \label{eqn:syseqnb}
		(I_{|V_{\tau}|}+C) \, \bm{b}=1_{|V_\tau|}.
		\end{equation}
		For $v\geq 2$ the expression in \eqref{eqn:bvv_expression0} is equivalent to the $v$-th equation in~\eqref{eqn:syseqnb}.
		A solution for $\bm{b}$ is
		\[ 
			\bm{b} 
			= (I_{|V_{\tau}|}+C)^{-1}1_{|V_\tau|}
			= \left( I_{|V_{\tau}|}-(-C) \right)^{-1} \, 1_{|V_\tau|}.
		\]
	From the equality
	\[
	\big(I_{|V_{\tau}|}+(-C)+(-C)^2+\cdots+(-C)^k\big)
	\big(I_{|V_{\tau}|}-(-C)\big)
	=
	I_{|V_{\tau}|}-(-C)^{k+1}
	\] 
	and the fact that for $|V_{\tau}|$-square lower triangular matrix with zero diagonal powers of $k\geq |V_{\tau}|$ are zero matrices we obtain
	\[
	\big(I_{|V_{\tau}|}+(-C)+(-C)^2+\cdots+(-C)^{|V_{\tau}|-1}\big)
	=
	\big(I_{|V_{\tau}|}-(-C)\big)^{-1}.
	\]
		If the matrix on the left is denoted by $K$ we have as solution $\bm{b}=K1_{|V_{\tau}|}$. For all $b_{vv}, v\geq 2,$ it can be shown that \eqref{eqn:bvv_expression} equals the $vv$-th element of this solution for $\bm{b}$. For $b_{11}$ consider the corresponding solution when the tournament $\tau$ is the one which node 1 shares with its parents. If node~1 has no parents in the ttt, then we have the solution $b_{11}=1$ which is indeed the case.
	\end{proof}

	\begin{lema1} \label{lem:vtau-max:2}
		Let $X$ follow a max-linear model as in Assumption~\ref{ass:max-mod} with respect to a ttt $\T$ consisting of a single tournament $\tau=(V,E)$. If the node $v\in V$ has at least one parent, then the parameter vector $\theta=(c_e)_{e\in E} \in \mathring{\Theta}_*$ is not identifiable from the distribution of $X_{V \setminus v}$.
		Specifically, there exists $\theta' = (c_e')_{e \in E} \in \mathring{\Theta}_*$ such that $\theta' \ne \theta$ and the distribution of $X_{V \setminus v}$ is the same under $\theta'$ as under $\theta$.
	\end{lema1}
	
	\begin{proof}
		Let $n = |V|$ denote the number of nodes. For convenience, rename the nodes to $V = \{1, \ldots, n\}$ in the ordering induced by the DAG, i.e., node $i$ has $i-1$ parents, for $i \in V$. The number of edges is $|E| = n(n-1)/2 =: m$, and the parameter set $\mathring{\Theta}_*$ is an open subset of $\mathbb{R}^E$. The distribution of $X$ is max-linear and is given by
		\begin{equation}
			\label{eq:XjbjiZi}
			X_{j} = \bigvee_{i=1}^j b_{ji} Z_i, \qquad j \in V, 
		\end{equation}
		where $b_{11} = 1$, $b_{jj} = 1 - \sum_{i=1}^{j-1} b_{ji}$ for $j \in V \setminus 1$, and where the $m$ coefficients $b = (b_{ji} : 1 \le i < j \le n)$ are determined by the edge parameters $\theta = (c_{ij} : 1 \le i < j \le n)$.
		
		Discarding the variable $X_v$ for some $v \in V \setminus 1$ yields the vector $X_{V \setminus v}$, the distribution of which is determined by the $m - (v-1)$ coefficients $(b_{ji} : 1 \le i < j \le n, j \ne v)$. For convenience, identify $\mathbb{R}^E$ with $\mathbb{R}^m$. Let $\pi : \mathbb{R}^m \to \mathbb{R}^{m-v+1}$ be the projection that sends $x = (x_{ij} : 1 \le i < j \le m)$ to $\pi(x) = (x_{ij} : 1 \le i < j \le m, j \ne v)$, i.e., the effect of $\pi$ is to leave out the coordinates $(i, v)$ with $i = 1,\ldots,v-1$. 
		By \eqref{eq:XjbjiZi} with $j = v$ removed, the distribution of $X_{V \setminus v}$ is determined by $\pi(b)$.
		
		The max-linear coefficients $b$ are a function of the edge parameters $\theta$. Formally, there exists a map $f : \mathring{\Theta}_* \to \mathbb{R}^m$ such that
		\[
		b = f(\theta).
		\]
		The function $f$ can be reconstructed from \eqref{eq:bvi-rev1} with $p(i,j) = (i, j)$ for $1 \le i < j \le n$. 
		Clearly, $f$ is continuous. Since the parameter $\theta$ is identifiable from the distribution of $X$ (Lemma~\ref{lem:Hth}), the function $f$ is also injective, i.e., $\theta \ne \theta'$ implies $f(\theta) \ne f(\theta')$. By the Invariance of Domain Theorem \citep[see, e.g.][]{kulpa1998}, the image $f(\mathring{\Theta}_*)$ is therefore an open subset of $\mathbb{R}^m$. But then, for any coefficient vector $b \in f(\mathring{\Theta}_*)$, there exists another coefficient vector $b' \in  f(\mathring{\Theta}_*)$ such that $b' \ne b$ but still $b_{ji} = b_{ji}'$ for all $1 \le i < j \le n$ and $j \ne v$ 
		--- in other words, such that $\pi(b) = \pi(b')$. Since $f$ is injective, the vectors $b$ and $b'$ originate from different edge parameter vectors $\theta = f^{-1}(b)$ and $\theta' = f^{-1}(b')$ in $\mathring{\Theta}_*$. But
		\[ 
		\pi(f(\theta)) = \pi(b) = \pi(b') = \pi(f(\theta')),
		\]
		so that the edge weight vectors $\theta$ and $\theta'$ induce the same distribution of $X_{V \setminus v}$. We conclude that the parameter $\theta$ is not identifiable from the distribution of $X_{V \setminus v}$.
		%
	\end{proof}

	\subsection{Proof of Proposition~\ref{prop:ident_block}}
	
	
	
	When reading the proof, the following perspective may help. Recall the notation in equations~\eqref{eq:HthU} and~\eqref{eqn:Hmuomega}. The knowledge of the (simple max-stable) distribution of $X_U$ implies the knowledge of its angular measure $H_U$ and thus of the unordered collection of pairs of atoms and masses $(\omega_r, \mu_r)$ for $r = 1, \ldots, s$. The vector $X_U$ can itself be represented as a max-linear model with $s$ independent factors and coefficient vectors $\beta_r = \mu_r \omega_r$ for $r = 1,\ldots,s$. We first need to ensure that we can match those vectors $\beta_r$ in a unique way to the max-linear coefficient vectors $(b_{vi})_{v \in U}$ for $i \in V$; note that the coordinates $v$ of those vectors are restricted to $U$. Next, from the latter vectors, we need to recover the edge coefficients $\theta = (c_e)_{e \in E}$.

	\begin{proof}[Proof of sufficiency (if) part of Proposition~\ref{prop:ident_block}]
		We assume~\ref{ident1} and~\ref{ident2}. In the first step of the proof we show that the angular measure of $X_U$ in \eqref{eqn:Hmuomega} is composed of $|V|$ distinct atoms and that we can associate every atom in $\{\omega_r : r=1, \ldots, |V|\}$ to some node $v \in V$ and accordingly be able to associate it to one of the atoms $a_{i,U}=(b_{vi}/m_{i,U})_{v\in U}$ for $i\in V$. 
		For this, we focus on the nature of the atoms $\{a_{i,U}\}$, given the conditions~\ref{ident1} and~\ref{ident2}.
		As a consequence, the max-linear coefficient matrix $b_{U \times V} = (b_{vi})_{v \in U, i \in V}$ can be recovered from the distribution of $X_U$. In Step~2, we show how  to recover from this matrix the edge parameters $\theta = (c_e)_{e \in E}$.
		
		
		\item \underline{Step 1.} Recall the representation $H_U = \sum_{i \in V} m_{i,U} \delta_{a_{i,U}}$  in \eqref{eq:HthU} of the angular measure of $X_U$. We shall show that all $|V|$ masses $m_{i,U}$ are positive and that all $|V|$ atoms $a_{i,U}$ are distinct. Moreover, we will show how to match the atoms to the nodes, that is, given an atom $\omega \in \{ \omega_r : r = 1, \ldots, |V|\}$ how to identify the node $i \in V$ such that $\omega = a_{i,U}$.

			\item	\emph{All $|V|$ vectors $\{a_{i,U}\}$ have positive  masses $\{m_{i,U}\}$.}
		Recall $m_{i,U} = \sum_{v \in U} b_{vi}$ and recall from \eqref{eq:equivs} that $b_{vi} > 0$ if and only if $v \in \Desc(i)$.
		It follows that $m_{i,U} = 0$ if and only if $\Desc(i) \cap U = \varnothing$ or, in other words, $\Desc(i) \subseteq \bar{U}$.
		But this is impossible 
		since $\Desc(i)$ contains at least one leaf node, that is, a node without children, and such a node belongs to $U$ by \ref{ident1}.
		We conclude that $m_{i,U} > 0$ for all $i \in V$. 

		\item \emph{All $|V|$ vectors $\{a_{i,U}\}$ are distinct.} 
	By \eqref{eq:omegaiU} it follows that whenever for two different nodes $i,j\in V$ we have $\Desc(i)\cap U\neq \Desc(j)\cap U$ then we can find two atoms, say $\omega'$ and $\omega''$, within the set $\{\omega_r\}$ such that $\omega'=a_{i,U}$ and $\omega''=a_{j,U}$. Because $\Desc(i)\cap U\neq \Desc(j)\cap U$ then necessarily $a_{i,U}\neq a_{j,U}$. Suppose however for two different nodes $i,j\in V$ we have $a_{i,U}=a_{j,U}$. This means that for the $u$-th and $j$-th elements of these vectors we have
	\[
	a_{i,u;U}=a_{j,u;U}
	\Longleftrightarrow
	\frac{b_{ui}}{m_i}=\frac{b_{uj}}{m_j} 
	\qquad \text{and} \qquad
	a_{i,j;U}=a_{j,j;U}
	\Longleftrightarrow
	\frac{b_{ji}}{m_i}=\frac{b_{jj}}{m_j}
	\] 
	Considering the ratios above, we should also have 
	\begin{equation} \label{eqn:aiuU}
	\frac{a_{i,u;U}}{a_{i,j;U}}=\frac{a_{j,u;U}}{a_{j,j;U}}
	\qquad \Longleftrightarrow \qquad
	\frac{b_{ui}}{b_{ji}}=\frac{b_{uj}}{b_{jj}}.
	\end{equation}
	
Because $a_{i,U}=a_{j,U}$ necessarily 
$		\Desc(i) \cap U
		= \Desc(j) \cap U. 
$
		By Lemma~\ref{prop:properties2-1} there exists a node $u$ 
		 such that one of the edge sets $\{(i,j), (i,u), (j,u)\}$ or $\{(j,i), (i,u), (j,u)\}$ is contained in $E$. Without loss of generality, suppose this holds for the first triple. Also, by Lemma~\ref{prop:properties2-1} there cannot be another node $k$ with $\Desc(k)\cap U=\Desc(i)\cap U=\Desc(j)\cap U$.  
		
		Suppose first $j, u \in U$. From the identities
		\begin{align} \label{eqn:identities}
		b_{ji} &= c_{ij} b_{ii}, &
		b_{ui} &= c_{iu} b_{ii}, &
		b_{uj} &= c_{ju} b_{jj}, 
		\end{align}
		and the criticality requirement 
		\[
		c_{iu} > c_{ij} c_{ju}
		\]
		we have 
		\begin{equation} \label{eqn:bui}
		\frac{b_{ui}}{b_{ji}}>\frac{b_{uj}}{b_{jj}}.
		\end{equation}
		This inequality shows that \eqref{eqn:aiuU} cannot happen, hence we cannot have $a_{i,U}=a_{j,U}$. This means that all $|V|$ atoms are distinct.


%
%
%
		
		Next suppose $j, u \in \bar{U}$. By Lemma~\ref{prop:properties2}, there exist nodes $j', u'\in U$ such that $j' \in \desc(j)$ and $u' \in \desc(u)$ and the paths $p(j,j')$ and $p(u,u')$ satisfy the properties in the said lemma. Because all nodes on the path $p(j,j')$ except possibly for $j$ have a unique parent, the path $p(j,j')$ cannot pass through $u$ (which has parents $i$ and $j$), and thus $j' \neq u'$. As $i$ is a parent of $j$, the shortest (and in fact the only) path from $i$ to $j'$ is the one that concatenates the edge $(i, j)$ with $p(j, j')$ (Lemma~\ref{prop:properties2}). It follows that
		\begin{equation} \label{eqn:bjprimei} 
		b_{j'i} = c_{p(j,j')} c_{ij} b_{ii}
		\qquad \text{and} \qquad
		b_{j'j} = c_{p(j,j')} b_{jj}.
		\end{equation} 
		By a similar argument, the path that concatenates the edge $(i,u)$ with the path $p(u,u')$ is the unique shortest path from $i$ to $u'$, while the path that concatenates $(j,u)$ with $p(u, u')$ is the unique shortest path from $j$ to $u'$. It follows that
		\begin{equation} \label{eqn:buprimei}
		b_{u'i} = c_{p(u,u')} c_{iu} b_{ii}
		\qquad \text{and} \qquad
		b_{u'j} = c_{p(u,u')} c_{ju} b_{jj}.
		\end{equation}
		Combining these equalities, $c_{iu}>c_{ij}c_{ju}$ implies that we should have
		\begin{equation} \label{eqn:buprime}
		\frac{b_{u'i}}{b_{u'j}}>\frac{b_{j'i}}{b_{j'j}}.
		\end{equation}
		However from $a_{i,U}=a_{j,U}$ we have for the $u'$-th and $j'$-th elements of these vectors 
		\[
		a_{i,u';U}=a_{j,u';U}
		\Longleftrightarrow
		\frac{b_{u'i}}{m_i}=\frac{b_{u'j}}{m_j} 
		\qquad \text{and} \qquad
		a_{i,j';U}=a_{j,j';U}
		\Longleftrightarrow
		\frac{b_{j'i}}{m_i}=\frac{b_{j'j}}{m_j}
		\] 
		Considering the ratios above, we should also have 
		\begin{equation} \label{eqn:aiuprime}
			\frac{a_{i,u';U}}{a_{i,j';U}}=\frac{a_{j,u';U}}{a_{j,j';U}}
			\qquad \Longleftrightarrow \qquad
			\frac{b_{u'i}}{b_{u'j}}=\frac{b_{j'i}}{b_{j'j}}.
		\end{equation}
Because of \eqref{eqn:buprime} the equalities in \eqref{eqn:aiuprime} cannot happen, hence we cannot have $a_{i,U}=a_{j,U}$.

%
		
		The analysis of the cases $(j, s) \in U \times \bar{U}$ and $(j, s) \in \bar{U} \times U$ is similar.
This shows that all $|V|$ vectors $a_{i,U}$ for $i \in V$ are different and because of \eqref{eqn:muom_ma}, all vectors $\omega_r$ for $r \in \{1,\ldots,|V|\}$ are different too.

		\item \emph{Distinguishing all atoms $a_{i,U}$ with zeroes on the same positions.} 
		For two different nodes $i, j \in V$, the atoms $a_{i,U}$ and $a_{j,U}$ have the same supports $\{v\in U: a_{vi}>0\}=\{v\in U: a_{vj}>0\}$ when $\Desc(i) \cap U = \Desc(j)\cap U$. By Lemma~\ref{prop:properties2-1}-5 there cannot be any other node $k \in V \setminus \{i, j\}$ with the same descendants in $U$. In the representation
		\[
		H_{\theta,U} 
		= \sum_{t \in V} m_{t,U} \delta_{a_{t,U}}
		= \sum_{r = 1}^{|V|} \mu_r \delta_{\omega_r}
		\]
		there are thus exactly two atoms, $\omega$ and $\omega'$, say, with the same indices of non-zero coordinates as $a_{i,U} = (b_{vi} / m_{i,U})_{v \in V}$ and $a_{j,U} = (b_{vj} / m_{j,U})_{v \in V}$. The question is then how to know whether $\omega = a_{i,U}$ and $\omega' = a_{j,U}$ or vice versa, $\omega = a_{j,U}$ and $\omega' = a_{i,U}$. Let $\mu$ and $\mu'$ be the masses of $\omega$ and $\omega'$, respectively, and consider the vectors $\beta = \mu \omega$ and $\beta' = \mu' \omega'$. An equivalent question is then how to identify $\beta$ and $\beta'$ with the two max-linear coefficient vectors $(b_{vi})_{v \in U}$ and $(b_{vj})_{v \in U}$.
		
		By Lemma~\ref{prop:properties2-1}-2 to~4, we can suppose that $i$ is the unique parent of $j$ and that $i$ and $j$ have a common child $u$. The analysis is now to be split up into different cases, according to whether $j$ and $u$ belong to $U$ or not. Recall that $b_{vi} = c_{p(i,v)} b_{ii}$ and $b_{vj} = c_{p(j,v)} b_{jj}$ for $v \in \Desc(j) \subset \Desc(i)$.
		

		Suppose first that $j, u \in U$. From \eqref{eqn:identities} we deduce~\eqref{eqn:bui} thanks to the criticality assumption. 
		In order to make the correct assignment of the two vectors $\beta=(\beta_v)_{v\in U}$ and $\beta'=(\beta'_v)_{v\in U}$ to the nodes $i$ and $j$, we need to check the inequality \eqref{eqn:bui}. If $\beta_u / \beta_j > \beta'_u / \beta'_j$ then we assign the vector $\beta$ to the node $i$ and the vector $\beta'$ to the node $j$. If the equality is reversed, we do the assignment the other way around.
		
		Next suppose that $j, u \in \bar{U}$. According to Lemma~\ref{prop:properties2}, there exist nodes $j', u'\in U$ so that there is a unique path from $j$ to $j'$ and from $u$ to $u'$. 
		By Lemma~\ref{prop:properties2}, the paths from $i$ to $u'$ and $j'$ and from $j$ to $j$ to $u'$ are
		\begin{align*}
			p(i,u') &= \{(i,u)\} \cup p(u,u'), &
			p(i,j') &= \{(i,j)\} \cup p(j,j'), &
			p(j,u') &= \{(j,u)\} \cup p(u,u'). 
		\end{align*}
	We have the same identities in \eqref{eqn:bjprimei} and \eqref{eqn:buprimei} which, together with the criticality assumption, lead to the inequality~\eqref{eqn:buprime}.
		In order to make the correct assignment of the two vectors $\beta=(\beta_v)_{v\in U}$ and $\beta'=(\beta'_v)_{v\in U}$ we do as above for the case $j,u\in U$.

		For the cases $(j,u) \in U \times \bar{U}$ and $(j, u) \in \bar{U} \times U$, we combine methods from the cases $(j,u) \in U \times U$ and $(j,u) \in \bar{U} \times \bar{U}$.  
		
		With this we finish the proof that we can learn the structure of every atom $\{\omega_r : r = 1, \ldots, |V|\}$, i.e., for every $r=1, \ldots, |V|$ we can identify the unique node $i\in V$ such that $\omega_r=a_{i,U}=(b_{vi}/m_{i,U})_{v\in U}$. This means that we can also match every element $\beta$ in the collection of vectors $\{\beta_r : r = 1, \ldots, |V|\}$ to the correct node $i \in V$ such that $\beta = (b_{vi})_{v \in U}$.
		
		
		\item \underline{Step 2.}
		In the previous step, we have shown that the distribution of $X_U$ (together with the knowledge of the graph structure) determines the max-linear coefficient matrix $b_{U \times V} = (b_{vi})_{v \in U, i \in V}$. Here, we show that this matrix suffices to reconstruct the vector of edge coefficients $\theta = (c_e)_{e \in E}$.
		
		If $v$ is a child of $i$, then $p(i,v) = \{(i,v)\}$ and thus $b_{vi} = c_{iv} b_{ii}$. If both $i$ and $v$ belong to $U$, then, clearly, we can identify $c_{iv} = b_{vi} / b_{ii}$. 

Let $i\in U$ with child $v \in \ch(i) \cap \bar{U}$. By Lemma~\ref{prop:properties2} there exists a node $v'\in U$ such that there is a unique path from $v$ to $v'$. Rewrite $v=v_1$ and $v'=v_n$ and consider the node set $\{v_1, v_2, \ldots, v_n\}$ on that  unique path. Using the fact that if a node $\ell$ has a single parent $k$, then $b_{\ell\ell} = 1 - c_{k \ell}$ (see \eqref{eqn:unique_parent}), we find the following identities for the max-linear coefficients $b_{v_nj}$ for $j \in \{v_n,\ldots,v_2,i\}$:
	\begin{equation} \label{eqn:binin}
	\begin{split}
		b_{v_nv_n} &= 1 - c_{v_{n-1}v_n},\\
		b_{v_nv_{n-1}} &= c_{p(v_{n-1},v_n)} b_{v_{n-1}v_{n-1}}
		= c_{v_{n-1} v_n} \left( 1-c_{v_{n-2}v_{n-1}} \right),\\
		&\vdots\\
		b_{v_nv_{2}} &= c_{p(v_{2},v_n)} b_{v_2v_2} 
		= c_{v_2 v_3} \cdots c_{v_{n-1}v_n} \left( 1-c_{v_{1}v_{2}} \right),\\
			b_{v_ni}&=c_{p(i,v_n)}b_{ii}
			=c_{iv_1}c_{v_1v_2}c_{v_2v_3}\cdots c_{v_{n-1}v_n}b_{ii}.\\
	\end{split}
\end{equation} 
From the first equation we identify $c_{v_{n-1}v_n}$, from the second $c_{v_{n-2}v_{n-1}}$ and so on until we identify $c_{v_1v_2}$ from the penultimate equation. From the last equation we can identify $c_{iv_1}$ because $b_{ii}$ is available from $b_{U \times V}$ in view of $i \in U$. 

The next step of the proof is to extract the edge parameters between a node with latent variable and its children.

Let $i\in \bar{U}$. We will show that we can identify all edge weights $c_{ij}$ for $j \in \ch(i)$. Because $i$ belongs to $\bar{U}$, it should have at least two children, say $v$ and $\bar{v}$. Take $v$ to be a node whose only parent is $i$. Note that we can always find such a node by Lemma~\ref{prop:properties2}. 
Let us first assume $v,\bar{v}\in U$. Because $v$ has only one parent, we have $b_{vv} = 1-c_{iv}$ and thus $c_{iv}=1-b_{vv}$. We also know $b_{vi}=c_{iv}b_{ii}$ and from here $b_{ii}=b_{vi}/c_{iv}=b_{vi}/(1-b_{vv})$. From $b_{\bar{v} i}=c_{i\bar{v}}b_{ii}$ we deduce $c_{i\bar{v}}=b_{\bar{v} i}(1-b_{vv})/b_{vi}$. Hence we have identified all the edge parameters related to children of $i$ which are observable, provided $i$ has two or more children in $U$. 

Next assume that both $v,\bar{v}\in \bar{U}$. By Lemma~\ref{prop:properties2}, there exists a node $v'\in \desc(v) \cap U$ such that there is a unique path from $v$ to $v'$ and which has the properties in the cited statement. Let the sequence of nodes along which the path passes be denoted by $\{v_1=v, v_2, \ldots, v_n=v'\}$. 
Using again that for a node $\ell$ with single parent $k$, $b_{\ell\ell} = 1 - c_{k \ell}$ (see \eqref{eqn:unique_parent}), we find the same identities as in \eqref{eqn:binin} for the max-linear coefficients $b_{v_nj}$ for $j \in \{v_n,\ldots,v_2\}$. From the first equation we obtain $c_{v_{n-1} v_n} = 1 - b_{v_nv_n}$, from the second equation $c_{v_{n-2}v_{n-1}} = 1 - b_{v_nv_{n-1}}/(1-b_{v_nv_n})$ and so on until we obtain $c_{v_1 v_2}$ from the penultimate equation in \eqref{eqn:binin}. Because we assumed $\pa(v_1)=\{i\}$ we have $b_{v_1v_1}=1-c_{iv_1}$ and thus
		\begin{equation*}
				b_{v_nv_{1}} = c_{p(v_{1},v_n)} b_{v_1v_1} 
				= c_{v_1 v_2} c_{v_2 v_3} \cdots c_{v_{n-1} v_n } \left( 1-c_{iv_1} \right),
		\end{equation*} 
from where we identify $c_{iv_1}$. Since $v_n \in U$, all coefficients $b_{v_nj}$ for $j \in V$ are contained in the max-linear coefficient matrix $b_{U \times V}$. The procedure just described thus allows us to compute all edge coefficients $c_{iv_1}, c_{v_1v_2}, \ldots, c_{v_{n-1}v_n}$.
		
Because the path $p(v_1, v_n)$ satisfies the properties in Lemma~\ref{prop:properties2}, and in view of the same lemma, plus the fact that the only parent of $v_1=v$ is $i$, it follows that the path $\{(i,v_1)\} \cup p(v_1,v_n)$ is the unique path between $i$ and $v_n$. Hence we have also
		\begin{equation} \label{eqn:bvv}
			b_{v_ni} = c_{p(i,{v_n})}b_{ii}
			=c_{iv_1}c_{v_1 v_2}c_{v_2 v_3}\cdots c_{v_{n-1} v_n}b_{ii}.
		\end{equation}   
		As $v_n \in U$, the value of $b_{v_n i}$ is known from $b_{U \times V}$. The edge coefficients on the right-hand side were expressed in terms of $b_{U \times V}$ in the previous paragraph. From there, we obtain the value of $b_{ii}$, which will be used next.
		
		Now consider the node $\bar{v}$, renamed to $\bar{v}_1$, which was an arbitrary child of $i$. For $\bar{v}=\bar{v}_1$ too, we can find a node $\bar{v}_m \in U$ and a sequence of nodes $\{\bar{v}_2, \ldots, \bar{v}_m\}$ according to Lemma~\ref{prop:properties2} satisfying the properties stated there.
		For all $r \in \{2,\ldots,m\}$, the node $\bar{v}_r$ has a unique parent, $\bar{v}_{r-1}$. We have the following equalities, using again by $b_{\ell\ell} = 1 - c_{k\ell}$ for a node $\ell$ with unique parent $k$
		\begin{equation*}
			\begin{split}
				b_{\bar{v}_m\bar{v}_m} 
				&= 1-c_{\bar{v}_{m-1}\bar{v}_m}, \\
				b_{\bar{v}_m\bar{v}_{m-1}} 
				&= c_{p(\bar{v}_{m-1},\bar{v}_m)} b_{\bar{v}_{m-1}\bar{v}_{m-1}}
				= c_{\bar{v}_{m-1} \bar{v}_m} \, \left( 1-c_{\bar{v}_{m-2}\bar{v}_{m-1}} \right),\\
				&\vdots\\
				b_{\bar{v}_m\bar{v}_{2}}
				&= c_{p(\bar{v}_{2},\bar{v}_m)} b_{\bar{v}_2 \bar{v}_2} 
				= c_{\bar{v}_2 \bar{v}_3} \cdots c_{\bar{v}_{m-1} \bar{v}_m} \left( 1-c_{\bar{v}_{1}\bar{v}_{2}} \right),\\
				b_{\bar{v}_mi}&=c_{p(i,\bar{v}_m)}b_{ii}
				=c_{i\bar{v}_1}c_{\bar{v}_1 \bar{v}_2}c_{\bar{v}_2 \bar{v}_3}\cdots c_{\bar{v}_{m-1} \bar{v}_m}b_{ii} .
\end{split}
		\end{equation*} 
		In the last equality we used that the path $\{(i,\bar{v}_1)\}\cup p(\bar{v}_1, \bar{v}_m)$ is the unique shortest one between $i$ and $\bar{v}_m$, because of Lemma~\ref{prop:properties2} and because $p(\bar{v}_1, \bar{v}_m)$ satisfies the properties~1-2 of the same lemma.  
		Since $\bar{v}_m \in U$, the values of the left-hand sides in the previous equations are contained in the given matrix $b_{U \times V}$.
		From the first equality we obtain $c_{\bar{v}_{m-1}\bar{v}_m}$, from the second one $c_{\bar{v}_{m-2}\bar{v}_{m-1}}$ and so on until we identify $c_{\bar{v}_{1}\bar{v}_{2}}$ from the penultimate equality, i.e., all edge parameters linked to $p(\bar{v}_1,\bar{v}_m)$. In the last equation above we replace $b_{ii}$ with the expression derived from \eqref{eqn:bvv} and we obtain the parameter $c_{i\bar{v}_1}$.
		
		If some of the children of $i$ are in $U$ and some others are in $\bar{U}$, we apply a combination of the techniques used in the two cases described above -- when two children are in $U$ or when two children are in $\bar{U}$.  
		
		This concludes the proof of the sufficiency (if) part.
\end{proof}

\begin{proof}[Proof of necessity (only if) part in Proposition~\ref{prop:ident_block}]
	Let $\bar{U} \subset V$ be such that at least one of the two conditions \ref{ident1} and \ref{ident2} is not satisfied and let $\theta = (c_e)_{e \in E} \in \mathring{\Theta}_*$. We will show that there exists another parameter $\theta' = (c_e')_{e \in E} \in \mathring{\Theta}_*$ such that $\theta' \ne \theta$ but the distribution of $X_{V \setminus u}$ under $\theta'$ is the same as the one under $\theta$.
	
	We consider two cases: case~(1) \ref{ident2} does not hold, i.e., there exists $u \in \bar{U}$ which is not the source of any tournament in $\T$; case~(2) \ref{ident2} holds but not \ref{ident1}, i.e., every $u \in \bar{U}$ is the source of some tournament in $\T$ but there exists $u \in \bar{U}$ with less than two children.
	\smallskip
	
	\noindent\emph{Case~(1): there exists $u \in \bar{U}$ which is not the source of any tournament in $\T$.}
	Then $u$ belongs to only a single tournament, say $\tau = (V_\tau, E_\tau)$; indeed, a node that belongs to two different tournaments must be the source of at least one of them, because otherwise it would have parents from two different tournaments, yielding a forbidden v-structure.
		
	Let $X_{V_{\tau}}$ be a max-linear model restricted to a single tournament, $\tau$. The coefficients associated to $X_{V_{\tau}}$, denote by $b_{vi}^\tau$ for $v,i \in V_\tau$, are determined by the weights of the edges $e \in E_\tau$. 
	By Lemma~\ref{lem:vtau-max:2}, we can modify the edge weights $c_e$ for $e \in E_\tau$ in such a way that the max-linear coefficients $b_{vi}^\tau$ for $v \in V_\tau \setminus u$ and $i \in V_\tau$ remain unaffected. 
	Let $\tilde{c}_e$ for $e \in E_\tau$ denote such a modified vector of edge weights. Define $\theta' = (c_e')_{e \in E} \in \mathring{\Theta}_*$ by
	\begin{align*}
		c_e' &= 
		\begin{cases} 
			\tilde{c}_e & \text{if $e \in E_\tau$,} \\
			c_e & \text{if $e \in E \setminus E_\tau$.}
		\end{cases}
	\end{align*}
	Then $\theta' \in \mathring{\Theta}_*$ too: indeed, $c_{ii}' > 0$ for all $i \in V$ by assumption on $\theta$ 
	and the fact the vector $(\tilde{c}_e : e \in E_\tau)$ satisfies $\tilde{c}_{ii} > 0$ for $i \in V_\tau$. By construction, the distribution of $X_{V_{\tau} \setminus u}$ is the same under $\theta'$ as under $\theta$.
	
	We show that the distribution of $X_{V \setminus u}$ under $\theta'$ is the same as the one under $\theta$. We proceed by induction on the number of tournaments.
	
	If $\T$ consists of a single tournament, then $\T = \tau$ and there is nothing more to show.
	
	So suppose $\T$ consists of $m \ge 2$ tournaments. The skeleton graph of $\T$ is a block graph and thus a decomposable graph. By the running intersection property, we can order the tournaments $\tau_1,\ldots,\tau_m$ with node sets $V_1,\ldots,V_m$ in such a way that $\tau_1 = \tau$, the tournament containing $u$, and such that $V_m \cap (V_1 \cup \ldots \cup V_{m-1})$ is a a singleton, say $\{s\}$. Then $s \ne u$ since $u$ belongs to only a single tournament. 
	
	Write $W = V_1 \cup \ldots \cup V_{m-1} \setminus u$. The joint distribution of $X_{V \setminus u}$ can be factorized into two parts: first, the distribution of $X_{W}$ and second, the conditional distribution of $X_{V_m \setminus s}$ given $X_{W}$. It is sufficient to show that both parts remain the same when $\theta$ is replaced by $\theta'$. 	
	\begin{itemize}
	\item
		By the induction hypothesis, the distribution of $X_{W}$ is the same under $\theta'$ as under $\theta$.
	\item
		By the global Markov property (Proposition~\ref{prop:markov_ml}), the conditional distribution of $X_{V_m \setminus s}$ given $X_{W}$ is the same as the conditional distribution of $X_{V_m \setminus s}$ given $X_s$. But the latter is determined by the joint distribution of $X_{V_m}$, which, in turn, only depends on the weights of the edges $e$ in $\tau_m$. By construction, these edge weights are the same under $\theta$ as under $\theta'$. It follows that the conditional distribution of $X_{V_m \setminus u}$ given $X_W$ is the same under $\theta'$ as under $\theta$.
	\end{itemize}
	We conclude that the distribution of $X_{V \setminus u}$ is the same under $\theta'$ as under $\theta$. Since $\theta \ne \theta'$, the parameter is not identifiable.
	\medskip

	\noindent\emph{Case~(2): any $u \in \bar{U}$ is the source of some tournament in $\T$ but there exists $u \in \bar{U}$ with less than two children.}
	Any $u \in \bar{U}$ must have at least one child (a node without children cannot be the source of a tournament). But then there exists $u \in \bar{U}$ with exactly one child: $\ch(u) = \{w\}$. The tournament of which $u$ is the source can only consist of the nodes $u$ and $w$ and the edge $(u, w)$. Now there are two subcases, depending on whether $u$ has any parents or not.
	\smallskip
	
	\emph{Case~(2).a: $u$ has no parents.}
	Then $u$ is the source node of $\T$ with a single child $w$.
	Removing the node $u$ yields the ttt $\T_{\setminus u} := (V \setminus u, E \setminus \{(u, w)\})$ with single source $w$. The random vector $X_{V \setminus u}$ follows the recursive max-linear model \eqref{eqn:mlgm} with respect to $\T_{\setminus u}$. Its distribution is determined by the coefficients $c_e$ for $e \in E \setminus \{(u, w)\}$. The value of $c_{uw}$ can thus be chosen arbitrarily in $(0, 1)$ without affecting the distribution of $X_{V \setminus u}$.
	\smallskip
	
	\emph{Case~(2).b: $u$ has parents.}
	Any ancestor of $w$ different from $u$ must be an ancestor of $u$ too, since otherwise there would be a v-structure at $w$
; therefore,
	\begin{equation*} 
		\An(w) = \an(u) \cup \{u, w\}. 
	\end{equation*}
	Let $\lambda > 0$ and close enough to $1$ (as specified below). Define $\theta' = (c_e')_{e \in E}$ by modifying the weights of edges adjacent to $u$: specifically,
	\begin{align*}
		c_{ju}' &= \lambda c_{ju}, \qquad j \in \pa(u); \\
		c_{uw}' &= \lambda^{-1} c_{uw}; \\
		c_{e}' &= c_e, \qquad e \in E \setminus [\{(j,u) : j \in \pa(u)\} \cup \{(u, w)\}].
	\end{align*}
	In words, $\theta'$ coincides with $\theta$ for edges $e$ that do not involve $u$, and $\theta' = \theta$ if and only if $\lambda = 1$. Since the parameter space $\mathring{\Theta}_*$ is open, $\theta'$ belongs to $\mathring{\Theta}_*$ for $\lambda$ sufficiently close to $1$. We claim that the distribution of $X_{V \setminus u}$ is invariant under $\lambda$. Hence, for $\lambda$ different from but sufficiently close to $1$, we have found a parameter $\theta' \neq \theta$ producing the same distribution of $X_{V \setminus u}$ as $\theta$.
	
	Under $\theta'$, the random vector $X_{V \setminus u}$ follows the max-linear model
	\[
		X_v = \bigvee_{i \in \An(v)} b_{vi}' Z_i, \qquad v \in V \setminus u,
	\]
	where $(Z_i)_{i \in V}$ is a vector of independent unit-Fréchet random variables and where the coefficients $b_{vi}'$ are given by equations \eqref{eq:bvi}, \eqref{eq:sumibvi} and \eqref{eq:cvv} with $c_e$ replaced by $c_e'$. 
	
	First, suppose $v \in V \setminus u$ is not a descendant of $u$. Then for any $i \in \An(v)$, the coefficient $b_{vi}'$ is a function of edge weights $c_e'$ for edges $e \in E$ different from $(u, w)$ and from $(j, u)$ for $j \in \pa(u)$. It follows that $c_e' = c_e$ for such edges and thus $b_{vi}' = b_{vi}$ for $v \in V \setminus \Desc(u)$ and $i \in \An(v)$.
	
	Second, suppose $v \in \desc(u)$. Then necessarily $v \in \Desc(w)$ too and for any $i \in \An(u)$, the path $p(i, v)$ passes by (or arrives in) $w$. Furthermore, any ancestor of $v$ not in $\An(w)$ is a descendant of $w$:
	\begin{equation}
	\label{eq:Anv}
		\An(v) = \an(u) \cup \{u, w\} \cup [\desc(w) \cap \An(v)].
	\end{equation}
	It follows that
	\[
		X_v = \bigvee_{i \in \an(u)} b_{vi}' Z_i
		\vee \left( b_{vu}' Z_u \vee b_{vw}' Z_w \right)
		\vee \bigvee_{i \in \desc(w) \cap \An(v)} b_{vi}' Z_i,
	\]
	where the last term on the right-hand side is to be omitted if $v = w$. We treat each of the three terms on the right-hand side separately.
	\begin{itemize}
	\item
		For $i \in \an(u)$, we have
		\[
			b_{vi}' = c_{ii}' c_{p(i,v)}' = c_{ii}' c_{p(i,u)}' c_{uw}' c_{p(w,v)}'
		\]
		where $c_{p(w,v)}' = 1$ if $v = w$. The coefficients $c_{ii}'$ and $c_{p(w, v)}'$ only involve weights $c_e'$ for edges $e \in E$ different from $(u, w)$ and $(j, u)$ for $j \in \pa(u)$; it follows that $c_{ii}' = c_{ii}$ and $c_{p(w,v)}' = c_{p(w,v)}$. Further, given $i \in \an(u)$ there exists $j \in \pa(u)$ such that $p(i, u)$ passes by $j$ right before reaching $u$ (with $i = j$ if $i \in \pa(u)$), and then
		\[
			c_{p(i,u)}' c_{uw}' 
			= c_{p(i,j)}' c_{ju}' c_{uw}'
			= c_{p(i,j)}' (\lambda c_{ju}) (\lambda^{-1} c_{uw}).
		\]
		Since $p(i,j)$ does not involve edges meeting $u$, we find that $c_{p(i,j)}' = c_{p(i,j)}$, so that the above expression does not depend on $\lambda$. We conclude that $b_{vi}' = b_{vi}$ for $i \in \an(u)$.
	\item
		If $v \in \desc(w)$ and $i \in \desc(w) \cap \An(v)$, the coefficient $b_{vi}'$ is
		\[
			b_{vi}' = c_{ii}' c_{p(i,v)}'.
		\]
		The path $p(i,v)$ does not involve edges touching $u$, so $c_{p(i,v)'} = c_{p(i,v)}$.
		By Lemma~\ref{lem:bzz}, the coefficient $c_{ii}' = b_{ii}'$ is a function of the edge weights $c_e'$ for edges $e$ in the tournament shared by $i$ and its parents. Since $i \in \desc(w)$ and since $w$ is the only child of $u$, none of these edges touches $u$, and thus $c_{e}' = c_{e}$ for all such edges.
		It follows that $c_{ii}' = c_{ii}$ too. We conclude that $b_{vi}' = b_{vi}$ for $v \in \desc(w)$ and $i \in \desc(w) \cap \An(v)$.
	\item
		The random variable $b_{vu}' Z_u \vee b_{vw}' Z_w$ is independent of all other variables $Z_i$ for $i \in V \setminus \{u,w\}$ and its distribution is equal to $(b_{vu}' + b_{vw}') Z$ for $Z$ a unit-Fréchet variable.
		We will show that $b_{vu}' + b_{vw}'$ does not depend on $\lambda$. 
		Since $1 = \sum_{i \in \An(v)} b_{vi}'$, the partition \eqref{eq:Anv} yields
		\[
			b_{vu}' + b_{vw}'
			= 1 - \sum_{i \in \an(u)} b_{vi}' - \sum_{i \in \desc(w) \cap \An(v)} b_{vi}'.
		\]
		(The last sum on the right-hand side is zero if $v$ is not a descendant of $w$.)
		In the two previous bullet points, we have already shown that the coefficients $b_{vi}'$ for $i$ in $\an(u)$ or $\desc(w) \cap \An(v)$ do not depend on $\lambda$. 
		By the stated identity, the sum $b_{vu}' + b_{vw}'$ then does not depend on $\lambda$ either.
	\end{itemize}

	We have thus shown that if $\bar{U}$ does not satisfy \ref{ident1}--\ref{ident2}, then we can find $u \in \bar{U}$ such that the distribution of $X_{V \setminus u}$ is the same under $\theta'$ as under $\theta$. As $\theta' \ne \theta$ by construction, this means that the parameter $\theta$ is not identifiable from the distribution of $X_{V \setminus u}$. But as $U = V \setminus \bar{U} \subseteq V \setminus \{u\}$, the parameter $\theta$ is not identifiable from the distribution of $X_U$ either. This confirms the necessity of \ref{ident1}--\ref{ident2} for the identifiability of $\theta$ from the distribution of $X_U$.
\end{proof}

\section*{Acknowledgements}

The comments and suggestions by two anonymous Reviewers have been greatly helpful in the preparation of the final version of the manuscript.
Stefka Asenova is particularly grateful to Eugen Pircalabelu and Ngoc Tran for their availability and precious help.

	\bibliography{bibliography_pr3}

\end{document}